\documentclass[12pt]{article}
\usepackage{graphicx,amsfonts,amsmath,amssymb,amsthm,mathrsfs,mathabx,url,hyperref}
\usepackage[usenames,dvipsnames]{xcolor}
\usepackage[toc,page]{appendix}
\usepackage{algorithm}
\usepackage[noend]{algpseudocode}

\title{Friedrichs Diagrams---Bosonic and Fermionic}
\author{
Morris Brooks\footnote{Institute of Mathematics, University of Zurich, Winterthurerstrasse 190, 8057 Zurich}\ \footnote{E-Mail: \url{Morris.Brooks@math.uzh.ch}}~and
Sascha Lill\footnote{Universit\`a degli Studi di Milano, Dipartimento di Matematica, Via Cesare Saldini 50, 20133 Milano, Italy}\ \footnote{E-Mail: \url{sascha.lill@unimi.it}}
} 
\date{\today}

\usepackage[utf8]{inputenc}
\usepackage{tikz}
\usetikzlibrary{decorations.pathreplacing}

\newcommand*\circled[1]{\tikz[baseline=(char.base)]{
            \node[shape=circle,draw,inner sep=1pt, outer sep = 0] (char) {#1};}}	

\usepackage[section]{placeins}
\usepackage{capt-of, caption}
\usepackage{paralist}
\usepackage{booktabs}
\usepackage{hhline}
\usepackage{color}                  

\newcommand{\bx}{\boldsymbol{x}}
\newcommand{\by}{\boldsymbol{y}}

\newcommand{\bX}{\boldsymbol{X}}
\newcommand{\bY}{\boldsymbol{Y}}

\newcommand{\cA}{\mathcal{A}}

\newcommand{\cC}{\mathcal{C}}

\newcommand{\cI}{\mathcal{I}}

\newcommand{\cN}{\mathcal{N}}

\newcommand{\cS}{\mathcal{S}}

\newcommand{\cU}{\mathcal{U}}

\newcommand{\fh}{\mathfrak{h}}

\newcommand{\sF}{\mathscr{F}}

\newcommand{\NNN}{\mathbb{N}}

\newcommand{\RRR}{\mathbb{R}}

\newcommand{\ZZZ}{\mathbb{Z}}

\newcommand{\fin}{\mathrm{fin}}

\newcommand{\imag}{\mathrm{imag}}

\newcommand{\sgn}{\mathrm{sgn}}

\newcommand{\normalordered}[1]{: \! #1 \! :}

\newcommand{\cont}[0]{
\!\! \raisebox{1pt}{
\begin{tikzpicture}
\filldraw[fill = white] (0, 0) circle (0.08);
\draw (-0.3,0) -- (-0.08,0);
\draw (0.3,0) -- (0.08,0);
\end{tikzpicture}
} \!\!
}
\newcommand{\contfer}[0]{
\!\! \raisebox{1pt}{
\begin{tikzpicture}
\filldraw[fill = white] (0, 0) circle (0.08);
\draw (-0.3,0) -- (-0.08,0);
\draw (0.3,0) -- (0.08,0);
\draw[thick] (0.3,-0.06) -- (0.15,-0.06);
\end{tikzpicture}
} \!\!
}

\newtheorem{lemma}{Lemma}
\numberwithin{lemma}{section}

\numberwithin{corollary}{section}
\newtheorem{theorem}{Theorem}
\numberwithin{theorem}{section}
\newtheorem{proposition}{Proposition}
\numberwithin{proposition}{section}

\numberwithin{remark}{section}
\theoremstyle{definition}
\numberwithin{definition}{section}
\theoremstyle{definition}
\numberwithin{assumption}{section}

\newcounter{remarks}
\begin{document}
\maketitle
\begin{abstract}
We give a mathematically precise review of a diagrammatic language introduced by Friedrichs in order to simplify computations with creation-- and annihilation operator products. In that language, we establish explicit formulas and algorithms for evaluating bosonic and fermionic commutators. Further, as an application, we demonstrate that the non--linear Hartree dynamics can be seen as a subset of the diagrams arising in the time evolution of a Bose gas. \\

\medskip

\noindent Key words: Friedrichs Diagrams, Many-Body Physics, Quantum Field Theory, Hartree Equation, Feynman Diagrams.
\end{abstract}

\tableofcontents

\clearpage

\section{Introduction}	
\label{sec:intro}

The evaluation of products or commutators between creation-- and annihilation operators, $ a^*, a $ is a frequently occurring task in many--body quantum physics and quantum field theory (QFT). In 1965, Friedrichs \cite{Friedrichs1965} suggested a diagrammatic method to track and simplify these calculations. A Friedrichs diagram represents a normal ordered product 
\begin{equation*}
	A = \int f_A(\bX_A, \bY_A) a^*_{\bx_{A, n_A}} \ldots a^*_{\bx_{A, 1}} a_{\by_{A, 1}} \ldots a_{\by_{A, m_A}} \; \mathrm{d}\bX_A \mathrm{d}\bY_A
\end{equation*}
by a graph with $ n_A $ legs (or ``prongs'' in Friedrichs' original language) that point to the left and represent creation operators, as well as $ m_A $ legs that point to the right and represent annihilation operators:

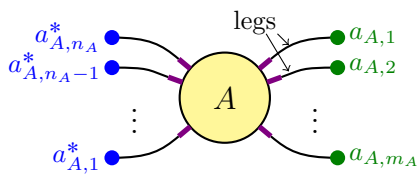
\begin{figure}[H]
	\centering
	\scalebox{1.0}{\def\r{0.6} 
\def\extl(#1, #2){\fill (#1, #2) circle (0.05); \fill[opacity = 0.3, blue] (#1, #2) circle (0.1);  } 
\def\extr(#1, #2){\fill (#1, #2) circle (0.05); \fill[opacity = 0.3, green!50!black] (#1, #2) circle (0.1);  } 
\def\conn(#1){({\r*cos(#1)},{\r*sin(#1)}) --  ({(\r+0.2)*cos(#1)},{(\r+0.2)*sin(#1)})} 
\def\connl(#1, #2){({(\r+0.2)*cos(#1)},{(\r+0.2)*sin(#1)}) .. controls  ({(\r+#2)*cos(#1)},{(\r+#2)*sin(#1)}) and} 
\begin{tikzpicture}
\useasboundingbox (-2.7,-1) rectangle (2.5,1.2);

\filldraw[fill = yellow!50!white, thick] (0,0) circle (\r) node{$A$} ;

\draw[thick] \connl(140, 0.5) ++(0.5,0) .. (-1.5,0.8) node[blue, anchor = east]{\footnotesize $ a^*_{A, n_A} $}; \extl(-1.5, 0.8)
\draw[thick] \connl(160, 0.5) ++(0.5,0) .. (-1.5,0.4) node[blue, anchor = east]{\footnotesize $ a^*_{A, n_A - 1} $}; \extl(-1.5, 0.4)
\node at (-1.2,-0.2) {$ \vdots $};
\draw[thick] \connl(220, 0.5) ++(0.5,0) .. (-1.5,-0.8) node[blue, anchor = east]{\footnotesize $ a^*_{A, 1} $}; \extl(-1.5, -0.8)

\draw[thick] \connl(40, 0.5) ++(-0.5,0) .. (1.5,0.8) node[green!50!black, anchor = west]{\footnotesize $ a_{A, 1} $} ; \extr(1.5, 0.8)
\draw[thick] \connl(20, 0.5) ++(-0.5,0) .. (1.5,0.4) node[green!50!black, anchor = west]{\footnotesize $ a_{A, 2} $}; \extr(1.5, 0.4)
\node at (1.2,-0.2) {$ \vdots $};
\draw[thick] \connl(-40, 0.5) ++(-0.5,0) .. (1.5,-0.8) node[green!50!black, anchor = west]{\footnotesize $ a_{A, m_A} $}; \extr(1.5, -0.8)

\node at (0.4,1) {\footnotesize legs};
\draw[->] (0.55,0.85) -- ++(0.3,-0.5);
\draw[->] (0.7,0.85) -- ++(0.15,-0.15);

\draw[line width = 2, red!50!blue] \conn(140);
\draw[line width = 2, red!50!blue] \conn(160);
\draw[line width = 2, red!50!blue] \conn(220);

\draw[line width = 2, red!50!blue] \conn(40);
\draw[line width = 2, red!50!blue] \conn(20);
\draw[line width = 2, red!50!blue] \conn(-40);

\end{tikzpicture}}
	\caption{A Friedrichs diagram.}
	\label{fig:Friedrichsdiagram}
\end{figure}

Formulas for evaluating certain bosonic products in terms of diagrams have been derived in \cite{Friedrichs1965, Hepp1969} and were subsequently applied in the context of constructive QFT \cite{Hepp1969, Glimm1968, GlimmJaffe1973, FeldmanOsterwalder1976}. In particular, \cite[Thm.~1.1]{Hepp1969} and \cite[Thm.~20.18]{DerezinskiGerard2013} imply a simple formula \eqref{eq:contthm} for the commutator $ [A, B] $ of two operators $ A $ and $ B $ that involves a ``sum over all ways to contract legs of the diagrams corresponding to $ A $ and $ B $'' (we specify in Section \ref{sec:diagrams}, what is meant by that, mathematically). A mathematically precise account of Friedrichs Diagrams, among other diagrammatic types, can be found in \cite[Ch.~20]{DerezinskiGerard2013}.\\
Friedrichs diagrams have also been used within mathematical investigations on the Yukawa model in $ 1+1 $ spacetime dimensions \cite{GlimmY2I, GlimmJaffeY2I, GlimmJaffeY2II, McBryanPark1975}, which includes fermions. For fermionic operators, one may expect a formula similar to \eqref{eq:contthm} to be true, up to the multiplication of certain contributions by a factor of $ (-1) $. However, we are not aware of such an explicit formula in the existing literature\footnote{In principle, such a formula could be derived from \cite[Thm.~20.18]{DerezinskiGerard2013} after plugging in all definitions for the fermionic case which is, however, not explicitly done in that reference.}.\\

In the present work, we therefore provide an analogue to the bosonic commutator formula \eqref{eq:contthm} in the fermionic case, which is \eqref{eq:contferthm}. This fermionic formula is our \textbf{main result} and given in Theorem \ref{thm:contfer}. We also state the bosonic commutator formula in Theorem \ref{thm:cont} and give two different proofs for it.\\

The bosonic formula is based on a so--called attached product $ A \cont B $ \eqref{eq:cont}, which is essentially a sum over all possible ways to contract legs of the diagrams of $ A $ and $ B $ in a certain way. Its fermionic analogue, $ A \contfer B $ \eqref{eq:contfer} agrees with $ A \cont B $ up to sign changes.\\
Further, the bosonic formula always holds for commutators $ [A, B] $, while the fermionic analogue describes either $ [A, B] $ or the anticommutator $ \{A, B\} $, depending on the numbers of legs.\\

We believe that the above--mentioned formulas can be useful for speeding up commutator calculations as they frequently appear when conjugating operators (e.g., Hamiltonians) with exponentials, such as in dressing or generalized Bogoliubov transformations. Recent examples for such conjugations can be found in the context of Bose gases \cite{BoccatoBrenneckeCenatiempoSchlein2020, AdhikariBrenneckeSchlein2021, NamTriay2021} and Fermi gases \cite{Giacomelli2022short, Giacomelli2022, FalconiGiacomelliHainzlPorta2021, ChristiansenHainzlNam2022, BenedikterPortaSchleinSeiringer2022} for investigating the ground state and low--excitation spectrum of a given Hamiltonian. Friedrichs diagrams provide a quick insight into which kinds of terms appear, and can also be used to anticipate and check results that are obtained non--diagrammatically. We comment further on this within Section \ref{subsec:multicomm}.\\

The mathematical description used in this article slightly differs from the one used in \cite{DerezinskiGerard2013}: We keep it less general to allow for an easier access to the topic and to a practical employment of Friedrichs diagrams.\\
Further, we remark that the Friedrichs diagrams as in Figure \ref{fig:Friedrichsdiagram} may closely resemble the so--called Feynman diagrams, which are ubiquitous in the QFT literature \cite{Feynman1949, Schweber1961, Weinberg1995, Schwartz2014}. For a mathematical discussion on Feynman diagrams, see also \cite{DerezinskiGerard2013, BogoliubovParasiuk1957, Hepp1966, Zimmermann1969, Scharf1995, Salmhofer1999}. Typically, in Feynman diagrams, internal lines translate into spacetime propagators or covariance matrices, while for Friedrichs diagrams, they correspond to spacelike Dirac distributions. Further, Feynman diagrams appearing in the literature usually represent complex numbers or distributions, while Friedrichs diagrams are typically employed in operator calculations. However, the term ``Feynman diagram'' is sometimes also used in a more restrictive sense (e.g., only denoting diagrams in QFT or those with time ordered propagators) or it may include certain diagrams representing operators, see \cite[Sect.~20.6.2]{DerezinskiGerard2013}.\\

The rest of this article is structured as follows: In Section \ref{sec:diagrams}, we introduce the notation and give the rules for translating diagrams into mathematical expressions. In Section \ref{sec:mainresults}, we state our main results, including formulas for $ [A, B], \{A, B\} $ and $ AB $ in terms of normal ordered products, and provide algorithms for their practical employment in commutator computations. Proofs of the main results are given in Section \ref{sec:proofs}. Section \ref{sec:applications} concerns applications of the graphical framework presented here. We demonstrate that the various contributions arising in the non--linear Hartree dynamics can naturally be identified as a subset of all the diagrams arising in the quantum time evolution and sketch how to evaluate multicommutators using Friedrichs diagrams.\\

\section{Diagrams and Contractions}	
\label{sec:diagrams}

In this section, we give a mathematically precise description of the diagrammatic formalism following \cite{Friedrichs1965}, that encodes normal ordered products of bosonic or fermionic creation-- and annihilation operators. We consider a general description of quantum many--body systems, where the degrees of freedom per particle are modeled by a measure space $ (X, \mu) $ with either $ X = \NNN $ (where $ \mu $ is the counting measure) or $ X = \RRR^d $ (where $ \mu $ is the Lebesgue measure). The one particle Hilbert space and Fock space are then given by
\begin{equation}
	\fh := L^2(X, \mu), \qquad \sF_\pm := \bigoplus_{N = 0}^\infty S_\pm\fh^{\otimes N},
\end{equation}
where $S_+$ is the symmetrization operator and $S_-$ is the anti--symmetrization operator. For some $ \Psi \in \sF_\pm $, we denote the $ N $--particle sector by $ \Psi^{(N)} \in \fh^{\otimes N} $. The finite--particle space is then defined as
\begin{equation}
	\sF_{\fin} := \{ \Psi \in \sF_\pm \; \mid \; \exists N_{\max} \in \NNN : \Psi^{(N)} = 0 \; \forall N > N_{\max} \}.
\label{eq:cD}
\end{equation}
On this space one defines the operator--valued distribution $ a^* $ by
\begin{equation}
    a^*(f)\Psi^{(N)}:=\sqrt{N+1}\, S_{\pm }\left(f\otimes \Psi^{(N)}\right), \qquad f \in \fh,
\end{equation}
and $ a(f) $ as the adjoint of $a^*(f)$. Further, we write the smeared--out operators as a formal integral $a^*(f)=\int f(\bx)a^*_{\bx} \mathrm{d}\bx$, respectively $a(f)=\int \overline{f(\bx)}a_{\bx} \mathrm{d}\bx$, over some (formal) pointwise operators $a^*_{\bx}$ and $a_{\bx}$. We further define for $ n_A, m_A \in \NNN_0 $ the operator valued distribution $D^{(n_A,m_A)}$ acting on $ \fh^{\otimes n_A} \otimes \fh^{\otimes m_A} $ as 
\begin{equation}
 D^{(n_A,m_A)}\left(f_1\otimes \dots \otimes f_{n_A} \otimes \overline{g_1} \otimes \dots \otimes \overline{g_{m_A}}\right):=a^*(f_{n_A})\dots a^*(f_1)a(g_1)\dots a(g_{m_A}).
\end{equation}
So for any $f_A\in \fh^{\otimes n_A}\otimes \fh^{\otimes m_A}$ and $ n_A, m_A \in \NNN_0 $, the normal ordered distribution pairing
\begin{equation}
	A :=D^{(n_A,m_A)}\left(f_A\right)= \int f_A(\bX_A, \bY_A) a^*_{\bx_{A, n_A}} \ldots a^*_{\bx_{A, 1}} a_{\by_{A, 1}} \ldots a_{\by_{A, m_A}} \; \mathrm{d}\bX_A \mathrm{d}\bY_A
\label{eq:A}
\end{equation}
renders a well--defined operator $ A: \sF_{\fin} \to \sF_{\fin} $, where we have written the smeared--out operator $D^{(n_A,m_A)}\left(f_A\right)$ as a formal integral over the pointwise operators. Here, we used the abbreviations
\begin{equation}
	\bX_A := (\bx_1, \ldots, \bx_{n_A}) \in X^{n_A}, \quad
	\bY_A := (\by_1, \ldots, \by_{m_A}) \in X^{m_A}.
\label{eq:bXA}
\end{equation}
$ a^*_{\bx} $ and $ a_{\bx} $ satisfy the canonical commutation/anticommutation relations (CCR/CAR) in the sense of distributions\footnote{That means, \eqref{eq:CCRCAR} entails rigorous commutation relations for the operators defined in \eqref{eq:A}.}:
\begin{equation}
	[a_{\bx}, a^*_{\by}]_\pm = \delta(\bx - \by), \qquad [a_{\bx}, a_{\by}]_\pm = [a^*_{\bx}, a^*_{\by}]_\pm = 0,
\label{eq:CCRCAR}
\end{equation}
where the commutator (for bosons) is given by $ [A, B]_+ := [A, B] = AB - BA $ and the anticommutator (for fermions) is given by $ [A, B]_- := \{A, B\} = AB + BA $. The set of operators $ A $ of the form \eqref{eq:A} generates a CCR/CAR $ ^* $--algebra $ \cA_\pm $.\\

As in \cite{Friedrichs1965}, we represent an operator of the form \eqref{eq:A} by a diagram (see Figure \ref{fig:Friedrichs_A}) consisting of
\begin{itemize}
\item One central vertex, representing $ f_A $, with $ n_A $ ``connectors''\footnote{Connectors on vertices are necessary to keep track of the order of the coordinates $ \bx_{A, j}, \by_{A, k} $, since those can generally not be commuted in $ f_A(\bX_A, \bY_A) $. In the language of graph theory, a ``vertex with $ n $ left-- and $ m $ right--connectors'' can be defined as a sub--graph with $ n + m $ vertices that are colored in 2 colors.} to the left and $ m_A $ ``connectors'' to the right, representing the coordinates $ \bx_{A, j} $ and $ \by_{A, k} $. The connectors are arranged, with respect to $ j $ and $ k $, from bottom to top for $ \bX_A $ and from top to bottom for $ \bY_A $.
\item $ n_A $ external vertices on the left of the diagram, representing the expressions $ a^*_{\bx} $ and ordered from top to bottom in the same order as they appear in \eqref{eq:A}.
\item $ m_A $ external vertices on the right of the diagram, representing the expressions $ a_{\by} $ and ordered from top to bottom in the same order as they appear in \eqref{eq:A}.
\item $ n_A + m_A $ edges (called ``external legs'') connecting each external vertex to a connector, and thus specifying which coordinate $ \bx $ or $ \by $ to plug into $ a^*_{\bx} $ or $ a_{\by} $.
\end{itemize}

\begin{figure}
	\centering
	\scalebox{1.0}{\def\r{0.6} 
\def\extl(#1, #2){\fill (#1, #2) circle (0.05); \fill[opacity = 0.3, blue] (#1, #2) circle (0.1);  } 
\def\extr(#1, #2){\fill (#1, #2) circle (0.05); \fill[opacity = 0.3, green!50!black] (#1, #2) circle (0.1);  } 
\def\conn(#1){({\r*cos(#1)},{\r*sin(#1)}) --  ({(\r+0.2)*cos(#1)},{(\r+0.2)*sin(#1)})} 
\def\connl(#1, #2){({\r*cos(#1)},{\r*sin(#1)}) .. controls  ({(\r+#2)*cos(#1)},{(\r+#2)*sin(#1)}) and} 
\begin{tikzpicture}
\useasboundingbox (-2,-1) rectangle (4,1.2);

\filldraw[fill = yellow!50!white, thick] (0,0) circle (\r) node{$A$} ;

\node at (-0.7,0) {$ \vdots $};
\node at (0.7,0) {$ \vdots $};

\node at (-0.5,1) {\footnotesize central vertex};
\draw[->] (-0.5,0.85) -- ++(0.2,-0.3);

\draw[line width = 2, red!50!blue] \conn(140) node[anchor = east]{\scriptsize $ \boldsymbol{x}_{A, n_A} $};
\draw[line width = 2, red!50!blue] \conn(160) node[anchor = east]{\scriptsize $ \boldsymbol{x}_{A, n_A-1} $};
\draw[line width = 2, red!50!blue] \conn(220) node[anchor = east]{\scriptsize $ \boldsymbol{x}_{A, 1} $};

\draw[line width = 2, red!50!blue] \conn(40) node[anchor = west]{\scriptsize $ \boldsymbol{y}_{A, 1} $};
\draw[line width = 2, red!50!blue] \conn(20) node[anchor = west]{\scriptsize $ \boldsymbol{y}_{A, 2} $};
\draw[line width = 2, red!50!blue] \conn(-40) node[anchor = west]{\scriptsize $ \boldsymbol{y}_{A, m_A} $};

\draw [decorate,decoration={brace, amplitude=5pt}, red!50!blue]  (2,0.7) -- ++ (0,-1.4);
\node[red!50!blue, anchor = west] at (2.1,0) {\footnotesize connectors};

\end{tikzpicture}}
	\scalebox{1.0}{\def\r{0.6} 
\def\extl(#1, #2){\fill (#1, #2) circle (0.05); \fill[opacity = 0.3, blue] (#1, #2) circle (0.1);  } 
\def\extr(#1, #2){\fill (#1, #2) circle (0.05); \fill[opacity = 0.3, green!50!black] (#1, #2) circle (0.1);  } 
\def\conn(#1){({\r*cos(#1)},{\r*sin(#1)}) --  ({(\r+0.2)*cos(#1)},{(\r+0.2)*sin(#1)})} 
\def\connl(#1, #2){({(\r+0.2)*cos(#1)},{(\r+0.2)*sin(#1)}) .. controls  ({(\r+#2)*cos(#1)},{(\r+#2)*sin(#1)}) and} 
\begin{tikzpicture}
\useasboundingbox (-2.7,-1) rectangle (2.5,1.2);

\filldraw[fill = yellow!50!white, thick] (0,0) circle (\r) node{$A$} ;

\draw[thick] \connl(140, 0.5) ++(0.5,0) .. (-1.5,0.8) node[blue, anchor = east]{\footnotesize $ a^*_{A, n_A} $}; \extl(-1.5, 0.8)
\draw[thick] \connl(160, 0.5) ++(0.5,0) .. (-1.5,0.4) node[blue, anchor = east]{\footnotesize $ a^*_{A, n_A - 1} $}; \extl(-1.5, 0.4)
\node at (-1.2,-0.2) {$ \vdots $};
\draw[thick] \connl(220, 0.5) ++(0.5,0) .. (-1.5,-0.8) node[blue, anchor = east]{\footnotesize $ a^*_{A, 1} $}; \extl(-1.5, -0.8)

\draw[thick] \connl(40, 0.5) ++(-0.5,0) .. (1.5,0.8) node[green!50!black, anchor = west]{\footnotesize $ a_{A, 1} $} ; \extr(1.5, 0.8)
\draw[thick] \connl(20, 0.5) ++(-0.5,0) .. (1.5,0.4) node[green!50!black, anchor = west]{\footnotesize $ a_{A, 2} $}; \extr(1.5, 0.4)
\node at (1.2,-0.2) {$ \vdots $};
\draw[thick] \connl(-40, 0.5) ++(-0.5,0) .. (1.5,-0.8) node[green!50!black, anchor = west]{\footnotesize $ a_{A, m_A} $}; \extr(1.5, -0.8)

\node at (0.4,1) {\footnotesize legs};
\draw[->] (0.55,0.85) -- ++(0.3,-0.5);
\draw[->] (0.7,0.85) -- ++(0.15,-0.15);

\draw[line width = 2, red!50!blue] \conn(140);
\draw[line width = 2, red!50!blue] \conn(160);
\draw[line width = 2, red!50!blue] \conn(220);

\draw[line width = 2, red!50!blue] \conn(40);
\draw[line width = 2, red!50!blue] \conn(20);
\draw[line width = 2, red!50!blue] \conn(-40);

\end{tikzpicture}}
	\caption{Left: A vertex with connectors representing $ f_A(\bX_A, \bY_A) $.\\ Right: A Friedrichs diagram including external vertices and legs.}
	\label{fig:Friedrichs_A}
\end{figure}
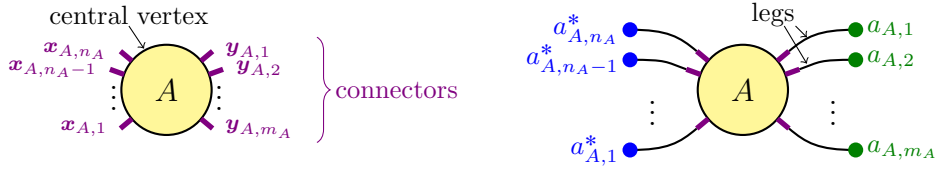

To simplify the notation, we will also write $ a^*_{\bx_{A, j}} =: a^*_{A, j} $ and $ a_{\by_{A, k}} := a_{A, k} $, whenever it is obvious whether an $ \bx $-- or a $ \by $--coordinate is used.\\

When computing commutators of operators \eqref{eq:A}, factors of the form $ \delta(\bx - \by) $ will appear naturally. These $ \delta(\bx - \by) $ can be viewed as ``contractions'' between the operators $ a_{\bx}^\sharp $ and $ a_{\by}^\sharp $ or between the coordinates $ \bx $ and $ \by $. As in \cite{Friedrichs1965}, we represent them by internal lines between the connectors representing the coordinates $ \bx $ and $ \by $: Consider a set of operators $ A_1, \ldots, A_V $ of the form \eqref{eq:A} with kernels (``vertices'') $ \{ f_v \}_{v = 1}^V, f_v \in \fh^{\otimes n_v} \otimes \fh^{\otimes m_v} $. The sets
\begin{equation}
	\cI := \{ (v, j) \; \mid \; 1 \le v \le V, \; 1 \le j \le n_v \}, \qquad
	\cI' := \{ (v, k) \; \mid \; 1 \le v \le V, \; 1 \le k \le m_v \}
\end{equation}
index all left/right pointing connectors and hence all coordinates. We would now like to contract $ C \in \NNN_0 $ pairs of coordinates $ \big\{ (\bx_{\pi(c)}, \by_{\pi'(c)}) \big\}_{c = 1}^C $, determined by the pairing maps
\begin{equation}
	\pi: \{1, \ldots, C\} \to \cI, \qquad
	\pi': \{1, \ldots, C\} \to \cI'.
\end{equation}
To indicate the order of external legs corresponding to uncontracted coordinates, we introduce the ordering maps
\begin{equation}
	\sigma: \{1, \ldots, |\cI| - C\} \to \cI \setminus  \imag (\pi), \qquad
	\sigma': \{1, \ldots, |\cI'| - C\} \to \cI' \setminus  \imag (\pi').
\end{equation}
Further, we use the abbreviations
\begin{equation}
	\bX_v := (\bx_{v, 1}, \ldots, \bx_{v, n_v}), \quad
	\bY_v := (\by_{v, 1}, \ldots, \by_{v, m_v}), \quad
	\mathrm{d}\bX := \prod_{v = 1}^V \mathrm{d}\bX_v, \quad
	\mathrm{d}\bY := \prod_{v = 1}^V \mathrm{d}\bY_v.
\label{eq:bXvbYv}
\end{equation}
The final contracted operator corresponding to $ \{ f_v \}_{v = 1}^V, \pi, \pi', \sigma, \sigma' $ is then given by
\begin{equation}
	G = \int \left( \prod_{v = 1}^V f_v(\bX_v, \bY_v) \right)
	\left( \prod_{c = 1}^C \delta(\bx_{\pi(c)} - \by_{\pi'(c)}) \right)
	\left( \prod_{\ell = 1}^{|\cI| - C} a_{\sigma(\ell)} \right)^*
	\left( \prod_{\ell' = 1}^{|\cI'| - C} a_{\sigma'(\ell')} \right) \; \mathrm{d}\bX \mathrm{d}\bY.
\label{eq:G}
\end{equation}
We represent $ G $ by a diagram (see Figure \ref{fig:Friedrichs_cont}) with:
\begin{itemize}
\item $ V $ internal vertices, each having $ n_v $ connectors to the left and $ m_v $ connectors to the right, encoding the functions $ f_v $ and their coordinates $ \bx_{v, j}, \by_{v, k} $.
\item $ |\cI| - C $ external vertices on the left of the diagram, encoding the expressions $ a^*_{\bx} $, ordered from bottom to top by $ \sigma $ (with $ a^*_{\sigma(|\cI| - C)} $ on top and $ a^*_{\sigma(1)} $ at the bottom). 
\item $ |\cI'| - C $ external vertices on the right of the diagram, encoding the expressions $ a_{\by} $, ordered from top to bottom by $ \sigma' $ (with $ a_{\sigma'(1)} $ on top and $ a_{\sigma'(|\cI'| - C)} $ at the bottom). 
\item $ C $ internal lines between pairs of connectors $ (\bx_{\pi(c)}, \by_{\pi'(c)}) $, encoding the factors $ \delta(\bx_{\pi(c)} - \by_{\pi'(c)}) $ and thus representing the contractions.
\item $ |\cI| + |\cI'| - 2C $ external lines between a left--connector and an external vertex on the left or a right--connector and an external vertex on the right. Those specify $ \sigma, \sigma' $, and thus, which coordinates $ \bx_{v, j} = \sigma(c), \by_{v, k} = \sigma'(c) $ have to be plugged into $ a^*_{\sigma(c)} $ and $ a_{\sigma'(c)} $.
\end{itemize}

\begin{figure}[t]
	\centering
	\hspace{-1cm}
	\scalebox{1.0}{\def\r{0.6} 
\def\rB{0.8} 
\def\rC{0.5} 
\def\extl(#1, #2){\fill (#1, #2) circle (0.05); \fill[opacity = 0.3, blue] (#1, #2) circle (0.1);  } 
\def\extr(#1, #2){\fill (#1, #2) circle (0.05); \fill[opacity = 0.3, green!50!black] (#1, #2) circle (0.1);  } 
\def\conn(#1, #2, #3, #4){({#4*cos(#1) + #2},{#4*sin(#1) + #3}) --  ({(#4+0.2)*cos(#1) + #2},{(#4+0.2)*sin(#1) + #3})} 
\def\connt(#1, #2, #3, #4, #5){({#4*cos(#1) + #2},{#4*sin(#1) + #3}) .. controls  ({(#4+#5)*cos(#1) + #2},{(#4+#5)*sin(#1) + #3}) and} 
\def\connte(#1, #2, #3, #4, #5){ ({(#4+#5)*cos(#1) + #2},{(#4+#5)*sin(#1) + #3}) .. ({#4*cos(#1) + #2},{#4*sin(#1) + #3}) } 
\begin{tikzpicture}
\useasboundingbox (-1.5,-2.6) rectangle (5.2,1.2);

\filldraw[fill = yellow!50!white, thick] (0,0) circle (\r) node{$A_1$} ;

\draw[line width = 2, red!50!blue] \conn(160, 0, 0, \r) node[anchor = east]{\scriptsize $ \boldsymbol{x}_{1, 4} $};
\draw[line width = 2, red!50!blue] \conn(180, 0, 0, \r) node[anchor = east]{\scriptsize $ \boldsymbol{x}_{1, 3} $};
\draw[line width = 2, red!50!blue] \conn(200, 0, 0, \r) node[anchor = east]{\scriptsize $ \boldsymbol{x}_{1, 2} $};
\draw[line width = 2, red!50!blue] \conn(220, 0, 0, \r) node[anchor = east]{\scriptsize $ \boldsymbol{x}_{1, 1} $};

\draw[line width = 2, red!50!blue] \conn(0, 0, 0, \r) node[anchor = west]{\scriptsize $ \boldsymbol{y}_{1, 1} $};
\draw[line width = 2, red!50!blue] \conn(-20, 0, 0, \r) node[anchor = west]{\scriptsize $ \boldsymbol{y}_{1, 2} $};
\draw[line width = 2, red!50!blue] \conn(-40, 0, 0, \r) node[anchor = west]{\scriptsize $ \boldsymbol{y}_{1, 3} $};

\filldraw[fill = yellow!50!white, thick] (3.5,-0.5) circle (\rB) node{$A_2$} ;

\draw[line width = 2, red!50!blue] \conn(140, 3.5, -0.5, \rB) node[anchor = east]{\scriptsize $ \boldsymbol{x}_{2, 5} $};
\draw[line width = 2, red!50!blue] \conn(160, 3.5, -0.5, \rB) node[anchor = east]{\scriptsize $ \boldsymbol{x}_{2, 4} $};
\draw[line width = 2, red!50!blue] \conn(180, 3.5, -0.5, \rB) node[anchor = east]{\scriptsize $ \boldsymbol{x}_{2, 3} $};
\draw[line width = 2, red!50!blue] \conn(200, 3.5, -0.5, \rB) node[anchor = east]{\scriptsize $ \boldsymbol{x}_{2, 2} $};
\draw[line width = 2, red!50!blue] \conn(220, 3.5, -0.5, \rB) node[anchor = east]{\scriptsize $ \boldsymbol{x}_{2, 1} $};

\draw[line width = 2, red!50!blue] \conn(40, 3.5, -0.5, \rB) node[anchor = west]{\scriptsize $ \boldsymbol{y}_{2, 1} $};
\draw[line width = 2, red!50!blue] \conn(20, 3.5, -0.5, \rB) node[anchor = west]{\scriptsize $ \boldsymbol{y}_{2, 2} $};
\draw[line width = 2, red!50!blue] \conn(-20, 3.5, -0.5, \rB) node[anchor = west]{\scriptsize $ \boldsymbol{y}_{2, 3} $};
\draw[line width = 2, red!50!blue] \conn(-40, 3.5, -0.5, \rB) node[anchor = west]{\scriptsize $ \boldsymbol{y}_{2, 4} $};

\filldraw[fill = yellow!50!white, thick] (1,-2) circle (\rC) node{$A_3$} ;

\draw[line width = 2, red!50!blue] \conn(120, 1, -2, \rC) node[anchor = east]{\scriptsize $ \boldsymbol{x}_{3, 3} $};
\draw[line width = 2, red!50!blue] \conn(150, 1, -2, \rC) node[anchor = east]{\scriptsize $ \boldsymbol{x}_{3, 2} $};
\draw[line width = 2, red!50!blue] \conn(180, 1, -2, \rC) node[anchor = east]{\scriptsize $ \boldsymbol{x}_{3, 1} $};

\draw[line width = 2, red!50!blue] \conn(40, 1, -2, \rC) node[anchor = west]{\scriptsize $ \boldsymbol{y}_{3, 1} $};
\draw[line width = 2, red!50!blue] \conn(0, 1, -2, \rC) node[anchor = west]{\scriptsize $ \boldsymbol{y}_{3, 2} $};

\end{tikzpicture}}
	\scalebox{1.0}{\def\r{0.6} 
\def\rB{0.8} 
\def\rC{0.5} 
\def\extl(#1, #2){\fill (#1, #2) circle (0.05); \fill[opacity = 0.3, blue] (#1, #2) circle (0.1);  } 
\def\extr(#1, #2){\fill (#1, #2) circle (0.05); \fill[opacity = 0.3, green!50!black] (#1, #2) circle (0.1);  } 
\def\conn(#1, #2, #3, #4){({#4*cos(#1) + #2},{#4*sin(#1) + #3}) --  ({(#4+0.2)*cos(#1) + #2},{(#4+0.2)*sin(#1) + #3})} 
\def\connt(#1, #2, #3, #4, #5){({#4*cos(#1) + #2},{#4*sin(#1) + #3}) .. controls  ({(#4+#5)*cos(#1) + #2},{(#4+#5)*sin(#1) + #3}) and} 
\def\connte(#1, #2, #3, #4, #5){ ({(#4+#5)*cos(#1) + #2},{(#4+#5)*sin(#1) + #3}) .. ({#4*cos(#1) + #2},{#4*sin(#1) + #3}) } 
\begin{tikzpicture}
\useasboundingbox (-2.2,-2.6) rectangle (5,1.2);

\draw[thick] \connt(140, 2.5, -0.5, (\rB+0.2), 0.8) ++(1,0) .. (0,0.8) -- (-1.5,0.8) node[blue, anchor = east]{\footnotesize $ a^*_{2, 5} $}; \extl(-1.5, 0.8)
\draw[thick] \connt(160, 0, 0, (\r+0.2), 0.3) ++(0.5,0) .. (-1.5,0.4) node[blue, anchor = east]{\footnotesize $ a^*_{1, 4} $}; \extl(-1.5, 0.4)
\draw[thick] \connt(180, 0, 0, (\r+0.2), 0.3) ++(0.5,0) .. (-1.5,0.0) node[blue, anchor = east]{\footnotesize $ a^*_{1, 3} $}; \extl(-1.5, 0)
\draw[thick] \connt(120, 1, -2, (\rC+0.2), 0.8) ++(0.5,0) .. (-1.2,-0.8) -- (-1.5,-0.8) node[blue, anchor = east]{\footnotesize $ a^*_{3, 3} $}; \extl(-1.5, -0.8)
\draw[thick] \connt(150, 1, -2, (\rC+0.2), 0.8) ++(0.5,0) .. (-1, -1.2) -- (-1.5,-1.2) node[blue, anchor = east]{\footnotesize $ a^*_{3, 2} $}; \extl(-1.5, -1.2)
\draw[thick] \connt(220, 0, 0, (\r+0.2)), 0.8) ++(1,0) .. (-1.5,-1.6) node[blue, anchor = east]{\footnotesize $ a^*_{1, 1} $}; \extl(-1.5, -1.6)
\draw[thick] \connt(180, 1, -2, (\rC+0.2), 0.3) ++(0.5,0) .. (-1.5,-2) node[blue, anchor = east]{\footnotesize $ a^*_{3, 1} $}; \extl(-1.5, -2)

\draw[thick] \connt(40, 2.5, -0.5, (\r+0.2), 0.3) ++(-0.5,0) .. (4,0.4) -- (4.3,0.4) node[green!50!black, anchor = west]{\footnotesize $ a_{2, 1} $} ; \extr(4.3, 0.4)
\draw[thick] \connt(20, 2.5, -0.5, (\r+0.2), 0.3) ++(-0.3,0) .. (4.1,0) -- (4.3,0) node[green!50!black, anchor = west]{\footnotesize $ a_{2, 2} $}; \extr(4.3, 0)
\draw[thick] \connt(0, 1, -2, (\rC+0.2), 2) ++(-0.5,0) .. (4.2,-0.8) -- (4.3,-0.8) node[green!50!black, anchor = west]{\footnotesize $ a_{3, 2} $}; \extr(4.3, -0.8)
\draw[thick] \connt(-20, 2.5, -0.5, (\rB+0.2), 0.3) ++(-0.5,0) .. (4.2,-1.2) -- (4.3,-1.2) node[green!50!black, anchor = west]{\footnotesize $ a_{2, 3} $}; \extr(4.3, -1.2)

\draw[red, opacity = .8, line width = 1] \connt(200, 0, 0, (\r+0.2), 1.8) \connte(-40, 2.5, -0.5, (\rB+0.2), 1.5);
\draw[red, opacity = .8, line width = 1] \connt(0, 0, 0, (\r+0.2), 1) \connte(220, 2.5, -0.5, (\rB+0.2), 1);
\draw[red, opacity = .8, line width = 1] \connt(-20, 0, 0, (\r+0.2), 0.3) \connte(160, 2.5, -0.5, (\rB+0.2), 0.3);
\draw[red, opacity = .8, line width = 1] \connt(-40, 0, 0, (\r+0.2), 0.3) \connte(180, 2.5, -0.5, (\rB+0.2), 0.3);
\draw[red, opacity = .8, line width = 1] \connt(40, 1, -2, (\rC+0.2), 0.3) \connte(200, 2.5, -0.5, (\rB+0.2), 0.3);

\draw[red, opacity = .8, ->] (1.6,0.8) -- ++(-0.4,-0.8);
\node[red, opacity = .8] at (2,1.4) {internal lines represent};
\node[red, opacity = .8] at (2,1) {contractions};

\filldraw[fill = yellow!50!white, thick] (0,0) circle (\r) node{$A_1$} ;

\draw[line width = 2, red!50!blue] \conn(160, 0, 0, \r);
\draw[line width = 2, red!50!blue] \conn(180, 0, 0, \r);
\draw[line width = 2, red!50!blue] \conn(200, 0, 0, \r);
\draw[line width = 2, red!50!blue] \conn(220, 0, 0, \r);

\draw[line width = 2, red!50!blue] \conn(0, 0, 0, \r);
\draw[line width = 2, red!50!blue] \conn(-20, 0, 0, \r);
\draw[line width = 2, red!50!blue] \conn(-40, 0, 0, \r);

\filldraw[fill = yellow!50!white, thick] (2.5,-0.5) circle (\rB) node{$A_2$} ;

\draw[line width = 2, red!50!blue] \conn(140, 2.5, -0.5, \rB);
\draw[line width = 2, red!50!blue] \conn(160, 2.5, -0.5, \rB);
\draw[line width = 2, red!50!blue] \conn(180, 2.5, -0.5, \rB);
\draw[line width = 2, red!50!blue] \conn(200, 2.5, -0.5, \rB);
\draw[line width = 2, red!50!blue] \conn(220, 2.5, -0.5, \rB);

\draw[line width = 2, red!50!blue] \conn(40, 2.5, -0.5, \rB);
\draw[line width = 2, red!50!blue] \conn(20, 2.5, -0.5, \rB);
\draw[line width = 2, red!50!blue] \conn(-20, 2.5, -0.5, \rB);
\draw[line width = 2, red!50!blue] \conn(-40, 2.5, -0.5, \rB);

\filldraw[fill = yellow!50!white, thick] (1,-2) circle (\rC) node{$A_3$} ;

\draw[line width = 2, red!50!blue] \conn(120, 1, -2, \rC);
\draw[line width = 2, red!50!blue] \conn(150, 1, -2, \rC);
\draw[line width = 2, red!50!blue] \conn(180, 1, -2, \rC);

\draw[line width = 2, red!50!blue] \conn(40, 1, -2, \rC);
\draw[line width = 2, red!50!blue] \conn(0, 1, -2, \rC);

\end{tikzpicture}}
	\caption{Left: An example of $ V = 3 $ internal vertices with connectors.\\ Right: An example of a contracted diagram corresponding to some $ G $.}
	\label{fig:Friedrichs_cont}
\end{figure}
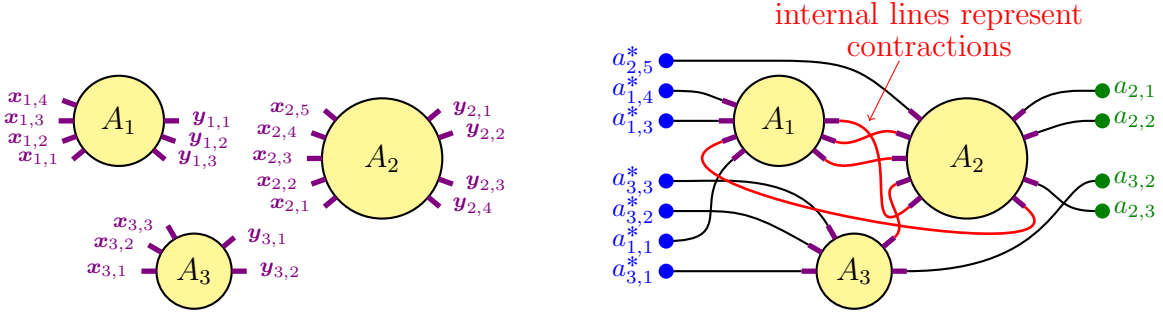
Every $ G $ as in \eqref{eq:G} can again be written as an operator of the form \eqref{eq:A}, i.e., with one single integral kernel. If we split $ \bX, \bY $ into contracted coordinates $ \bX_{\pi} := (\bx_{\pi(c)})_{c = 1}^C, \bY_{\pi} := (\bx_{\pi(c)})_{c = 1}^C $ and uncontracted coordinates $ \bX' := \bX \setminus \bX_{\pi}, \bY' := \bY \setminus \bY_{\pi} $, then the kernel associated with $ G $ is given by
\begin{equation}
     f_G(\bX', \bY') = \int \left( \prod_{v = 1}^V f_v(\bX_v, \bY_v) \right)
	\left( \prod_{c = 1}^C \delta(\bx_{\pi(c)} - \by_{\pi'(c)}) \right) \; \mathrm{d}\bX_{\pi} \mathrm{d}\bY_{\pi'}.
\end{equation}
In particular, $ f_G \in L^2 $. So $ G $ can also be represented by a diagram with a single vertex, as depicted in Figure \ref{fig:Friedrichs_contG}.\\
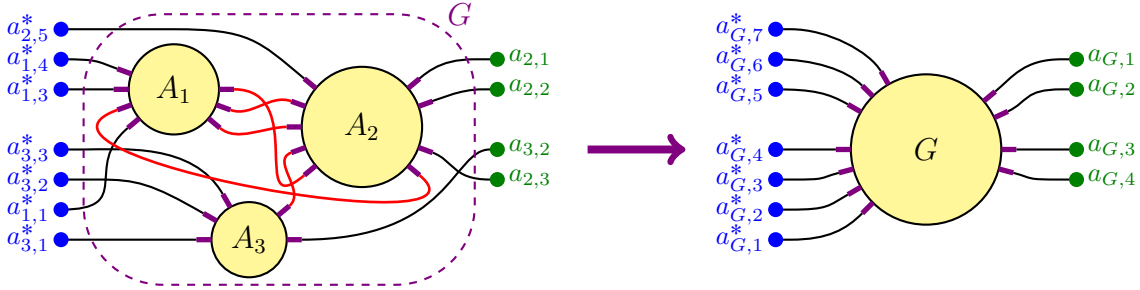
\begin{figure}[t]
	\hspace{-1cm}
	\scalebox{1.0}{\def\r{0.6} 
\def\rB{0.8} 
\def\rC{0.5} 
\def\rG{1} 
\def\xGA{8} 
\def\xGB{12} 
\def\extl(#1, #2){\fill (#1, #2) circle (0.05); \fill[opacity = 0.3, blue] (#1, #2) circle (0.1);  } 
\def\extr(#1, #2){\fill (#1, #2) circle (0.05); \fill[opacity = 0.3, green!50!black] (#1, #2) circle (0.1);  } 
\def\conn(#1, #2, #3, #4){({#4*cos(#1) + #2},{#4*sin(#1) + #3}) --  ({(#4+0.2)*cos(#1) + #2},{(#4+0.2)*sin(#1) + #3})} 
\def\connt(#1, #2, #3, #4, #5){({#4*cos(#1) + #2},{#4*sin(#1) + #3}) .. controls  ({(#4+#5)*cos(#1) + #2},{(#4+#5)*sin(#1) + #3}) and} 
\def\connte(#1, #2, #3, #4, #5){ ({(#4+#5)*cos(#1) + #2},{(#4+#5)*sin(#1) + #3}) .. ({#4*cos(#1) + #2},{#4*sin(#1) + #3}) } 
\begin{tikzpicture}
\useasboundingbox (-2.2,-2.6) rectangle (15,1.2);

\draw[thick] \connt(140, 2.5, -0.5, (\rB+0.2), 0.8) ++(1,0) .. (0,0.8) -- (-1.5,0.8) node[blue, anchor = east]{\footnotesize $ a^*_{2, 5} $}; \extl(-1.5, 0.8)
\draw[thick] \connt(160, 0, 0, (\r+0.2), 0.3) ++(0.5,0) .. (-1.5,0.4) node[blue, anchor = east]{\footnotesize $ a^*_{1, 4} $}; \extl(-1.5, 0.4)
\draw[thick] \connt(180, 0, 0, (\r+0.2), 0.3) ++(0.5,0) .. (-1.5,0.0) node[blue, anchor = east]{\footnotesize $ a^*_{1, 3} $}; \extl(-1.5, 0)
\draw[thick] \connt(120, 1, -2, (\rC+0.2), 0.8) ++(0.5,0) .. (-1.2,-0.8) -- (-1.5,-0.8) node[blue, anchor = east]{\footnotesize $ a^*_{3, 3} $}; \extl(-1.5, -0.8)
\draw[thick] \connt(150, 1, -2, (\rC+0.2), 0.8) ++(0.5,0) .. (-1, -1.2) -- (-1.5,-1.2) node[blue, anchor = east]{\footnotesize $ a^*_{3, 2} $}; \extl(-1.5, -1.2)
\draw[thick] \connt(220, 0, 0, (\r+0.2)), 0.8) ++(1,0) .. (-1.5,-1.6) node[blue, anchor = east]{\footnotesize $ a^*_{1, 1} $}; \extl(-1.5, -1.6)
\draw[thick] \connt(180, 1, -2, (\rC+0.2), 0.3) ++(0.5,0) .. (-1.5,-2) node[blue, anchor = east]{\footnotesize $ a^*_{3, 1} $}; \extl(-1.5, -2)

\draw[thick] \connt(40, 2.5, -0.5, (\r+0.2), 0.3) ++(-0.5,0) .. (4,0.4) -- (4.3,0.4) node[green!50!black, anchor = west]{\footnotesize $ a_{2, 1} $} ; \extr(4.3, 0.4)
\draw[thick] \connt(20, 2.5, -0.5, (\r+0.2), 0.3) ++(-0.3,0) .. (4.1,0) -- (4.3,0) node[green!50!black, anchor = west]{\footnotesize $ a_{2, 2} $}; \extr(4.3, 0)
\draw[thick] \connt(0, 1, -2, (\rC+0.2), 2) ++(-0.5,0) .. (4.2,-0.8) -- (4.3,-0.8) node[green!50!black, anchor = west]{\footnotesize $ a_{3, 2} $}; \extr(4.3, -0.8)
\draw[thick] \connt(-20, 2.5, -0.5, (\rB+0.2), 0.3) ++(-0.5,0) .. (4.2,-1.2) -- (4.3,-1.2) node[green!50!black, anchor = west]{\footnotesize $ a_{2, 3} $}; \extr(4.3, -1.2)

\draw[red, opacity = .8, line width = 1] \connt(200, 0, 0, (\r+0.2), 1.8) \connte(-40, 2.5, -0.5, (\rB+0.2), 1.5);
\draw[red, opacity = .8, line width = 1] \connt(0, 0, 0, (\r+0.2), 1) \connte(220, 2.5, -0.5, (\rB+0.2), 1);
\draw[red, opacity = .8, line width = 1] \connt(-20, 0, 0, (\r+0.2), 0.3) \connte(160, 2.5, -0.5, (\rB+0.2), 0.3);
\draw[red, opacity = .8, line width = 1] \connt(-40, 0, 0, (\r+0.2), 0.3) \connte(180, 2.5, -0.5, (\rB+0.2), 0.3);
\draw[red, opacity = .8, line width = 1] \connt(40, 1, -2, (\rC+0.2), 0.3) \connte(200, 2.5, -0.5, (\rB+0.2), 0.3);

\draw[rounded corners = 30, dashed, thick, red!50!blue] (-1.2,-2.6) rectangle (4,1);

\filldraw[fill = yellow!50!white, thick] (0,0) circle (\r) node{$A_1$} ;

\draw[line width = 2, red!50!blue] \conn(160, 0, 0, \r);
\draw[line width = 2, red!50!blue] \conn(180, 0, 0, \r);
\draw[line width = 2, red!50!blue] \conn(200, 0, 0, \r);
\draw[line width = 2, red!50!blue] \conn(220, 0, 0, \r);

\draw[line width = 2, red!50!blue] \conn(0, 0, 0, \r);
\draw[line width = 2, red!50!blue] \conn(-20, 0, 0, \r);
\draw[line width = 2, red!50!blue] \conn(-40, 0, 0, \r);

\filldraw[fill = yellow!50!white, thick] (2.5,-0.5) circle (\rB) node{$A_2$} ;

\draw[line width = 2, red!50!blue] \conn(140, 2.5, -0.5, \rB);
\draw[line width = 2, red!50!blue] \conn(160, 2.5, -0.5, \rB);
\draw[line width = 2, red!50!blue] \conn(180, 2.5, -0.5, \rB);
\draw[line width = 2, red!50!blue] \conn(200, 2.5, -0.5, \rB);
\draw[line width = 2, red!50!blue] \conn(220, 2.5, -0.5, \rB);

\draw[line width = 2, red!50!blue] \conn(40, 2.5, -0.5, \rB);
\draw[line width = 2, red!50!blue] \conn(20, 2.5, -0.5, \rB);
\draw[line width = 2, red!50!blue] \conn(-20, 2.5, -0.5, \rB);
\draw[line width = 2, red!50!blue] \conn(-40, 2.5, -0.5, \rB);

\filldraw[fill = yellow!50!white, thick] (1,-2) circle (\rC) node{$A_3$} ;

\draw[line width = 2, red!50!blue] \conn(120, 1, -2, \rC);
\draw[line width = 2, red!50!blue] \conn(150, 1, -2, \rC);
\draw[line width = 2, red!50!blue] \conn(180, 1, -2, \rC);

\draw[line width = 2, red!50!blue] \conn(40, 1, -2, \rC);
\draw[line width = 2, red!50!blue] \conn(0, 1, -2, \rC);

\draw[->, line width = 3, red!50!blue] (5.5, -0.8) -- ++(1.3,0);
\node[red!50!blue] at (3.8,1) {$ G $};

\draw[thick] \connt(120, 10, -0.8, (\rG+0.2), 0.3) ++(0.5,0) .. (8.2,0.8) -- (\xGA,0.8) node[blue, anchor = east]{\footnotesize $ a^*_{G, 7} $}; \extl(\xGA, 0.8)
\draw[thick] \connt(135, 10, -0.8, (\rG+0.2), 0.3) ++(0.5,0) .. (\xGA,0.4) node[blue, anchor = east]{\footnotesize $ a^*_{G, 6} $}; \extl(\xGA, 0.4)
\draw[thick] \connt(150, 10, -0.8, (\rG+0.2), 0.3) ++(0.5,0) .. (\xGA,0.0) node[blue, anchor = east]{\footnotesize $ a^*_{G, 5} $}; \extl(\xGA, 0)
\draw[thick] \connt(180, 10, -0.8, (\rG+0.2), 0.3) ++(0.5,0) .. (8.2,-0.8) -- (\xGA,-0.8) node[blue, anchor = east]{\footnotesize $ a^*_{G, 4} $}; \extl(\xGA, -0.8)
\draw[thick] \connt(195, 10, -0.8, (\rG+0.2), 0.3) ++(0.5,0) .. (8.2, -1.2) -- (\xGA,-1.2) node[blue, anchor = east]{\footnotesize $ a^*_{G, 3} $}; \extl(\xGA, -1.2)
\draw[thick] \connt(210, 10, -0.8, (\rG+0.2)), 0.3) ++(0.5,0) .. (8.2, -1.6) --(\xGA,-1.6) node[blue, anchor = east]{\footnotesize $ a^*_{G, 2} $}; \extl(\xGA, -1.6)
\draw[thick] \connt(225, 10, -0.8, (\rG+0.2), 0.3) ++(0.5,0) .. (8.2, -2) --(\xGA,-2) node[blue, anchor = east]{\footnotesize $ a^*_{G, 1} $}; \extl(\xGA, -2)

\draw[thick] \connt(40, 10, -0.8, (\rG+0.2), 0.3) ++(-0.5,0) .. (11.8,0.4) -- (\xGB,0.4) node[green!50!black, anchor = west]{\footnotesize $ a_{G, 1} $} ; \extr(\xGB, 0.4)
\draw[thick] \connt(25, 10, -0.8, (\rG+0.2), 0.3) ++(-0.5,0) .. (11.8,0) -- (\xGB,0) node[green!50!black, anchor = west]{\footnotesize $ a_{G, 2} $}; \extr(\xGB, 0)
\draw[thick] \connt(0, 10, -0.8, (\rG+0.2), 0.3) ++(-0.5,0) .. (11.8,-0.8) -- (\xGB,-0.8) node[green!50!black, anchor = west]{\footnotesize $ a_{G, 3} $}; \extr(\xGB, -0.8)
\draw[thick] \connt(-15, 10, -0.8, (\rG+0.2), 0.3) ++(-0.5,0) .. (11.8,-1.2) -- (\xGB,-1.2) node[green!50!black, anchor = west]{\footnotesize $ a_{G, 4} $}; \extr(\xGB, -1.2)

\filldraw[fill = yellow!50!white, thick] (10,-0.8) circle (\rG) node{$G$} ;

\draw[line width = 2, red!50!blue] \conn(120, 10, -0.8, \rG);
\draw[line width = 2, red!50!blue] \conn(135, 10, -0.8, \rG);
\draw[line width = 2, red!50!blue] \conn(150, 10, -0.8, \rG);
\draw[line width = 2, red!50!blue] \conn(180, 10, -0.8, \rG);
\draw[line width = 2, red!50!blue] \conn(195, 10, -0.8, \rG);
\draw[line width = 2, red!50!blue] \conn(210, 10, -0.8, \rG);
\draw[line width = 2, red!50!blue] \conn(225, 10, -0.8, \rG);

\draw[line width = 2, red!50!blue] \conn(40, 10, -0.8, \rG);
\draw[line width = 2, red!50!blue] \conn(25, 10, -0.8, \rG);
\draw[line width = 2, red!50!blue] \conn(0, 10, -0.8, \rG);
\draw[line width = 2, red!50!blue] \conn(-15, 10, -0.8, \rG);

\end{tikzpicture}}
	\caption{The diagram corresponding to $ G $ is equivalent to a single--vertex diagram.\\}
	\label{fig:Friedrichs_contG}
\end{figure}

In order to formulate the commutator formulas, we will introduce the notion of an attached product as in \cite{Friedrichs1965}: Consider two operators $ A, B $ of form \eqref{eq:A}, with kernels $ f_A $ and $ f_B $. The coordinate index sets will be denoted by
\begin{equation}
	\cI_A := \{ (A, 1), \ldots, (A, n_A) \}, \qquad
	\cI'_A := \{ (A, 1), \ldots, (A, m_A) \}
\label{eq:cIA}
\end{equation}
and the same for $ B $. Further, we introduce a set that indexes all possible ways to contract the $ \by $--coordinates of $ A $ to the $ \bx $--coordinates of $ B $:
\begin{equation}
\begin{aligned}
	\cC := \big\{ (\pi, \pi') \; \big\vert \;
	&\pi: \{1, \ldots, C \} \to \cI_B, \;
	\pi': \{1, \ldots, C\} \to \cI'_A, \\
	&1 \le C \le \min(m_A, n_B), \; |\mathrm{imag}(\pi')|=C,\;
	\pi(1) > \ldots > \pi(C) \big\},
\end{aligned}
\label{eq:cC}
\end{equation}
where the ordering relation $ (B, j) > (B, j') $ is to be understood as $ j > j' $. Diagrammatically, $ (\pi, \pi') \in \cC $ is represented by $ C $ lines, each from connector $ \pi'(c) $ of $ A $ to connector $ \pi(c) $ of $ B $. Here, $ (\pi, \pi') \in \cC $ is also called a ``contraction configuration'', or ``config'' for short. For a given config $ (\pi, \pi') $, we denote the sets of indices belonging to contractable but uncontracted left/right--connectors by
\begin{equation}
\begin{aligned}
	\cU := &\big\{ (B, j) \in \cI_B \; \mid \nexists\;  c \in \{1, \ldots, C\} : \pi(c) = (B, j) \big\},\\
	\cU' := &\big\{ (A, k) \in \cI'_A \; \mid \nexists \; c \in \{1, \ldots, C\} : \pi'(c) = (A, k) \big\}.
\end{aligned}
\label{eq:cU}
\end{equation}
Further, we set $ \bX := (\bX_A, \bX_B) \in X^{n_A + n_B}, \bY := (\bY_A, \bY_B) \in X^{m_A + m_B} $. Then, the \textbf{bosonic attached product} is defined as
\begin{equation}
\begin{aligned}
	A \cont B := \sum_{(\pi, \pi') \in \cC} \int &f_A(\bX_A, \bY_A) f_B(\bX_B, \bY_B)
	\prod_{c = 1}^C \delta(\bx_{\pi(c)} - \by_{\pi'(c)}) \times \\
	&\times \left(\prod_{\ell = 1}^{n_A} a_{A, \ell}\right)^*
	\left(\prod_{u \in \cU} a_u \right)^*
	\prod_{u' \in \cU'} a_{u'}
	\prod_{\ell' = 1}^{m_B} a_{B, \ell'} \; \mathrm{d}\bX \mathrm{d}\bY.
\end{aligned}
\label{eq:cont}
\end{equation}
Here, we use the convention that, within products over a set (like $ \prod_{u' \in \cU'} $ or $ \prod_{u \in \cU} $), the indices $ (A, k) $ or $ (B, j) $ are arranged in increasing order in $ k $ or $ j $. Diagrammatically, \eqref{eq:cont} corresponds to a sum over all different diagrams in which at least one right--connector of $ A $ is contracted with a left--connector of $ B $. Uncontracted connectors are turned into external legs while keeping their order.

In the fermionic version of the commutator formula, we will encounter a similar sum over all possible contraction configs, which differs from \eqref{eq:contfer} by a sign change for certain configs. To specify the signs, we introduce the notion of a ``maximally crossed'' diagram, as shown in Figure \ref{fig:Friedrichs_maxcrossing}: A diagram together with its contribution, indexed by the config $ (\pi, \pi') \in \cC $, is called \textbf{maximally crossed} if $ \pi(c) = (B, n_B - c + 1) $ and $ \pi'(c) = (A, m_A - c + 1) $ for all $ 1 \le c \le C $. So we have a ``maximal crossing of contraction lines'' in the sense that the $ c $--th bottom--most right--connector of $ A $ is exactly contracted to the $ c $--th top--most left--connector of $ B $. Now, let $ \sigma: \cI_B \to \cI_B, \sigma': \cI'_A \to \cI'_A $ be the unique index permutations that take the diagram into a maximally crossed form while preserving the order of all uncontracted indices, that is,
\begin{equation}
\begin{aligned}
	\sigma(\pi(c)) &= (B, n_B - c + 1) \quad
	&\text{and} \quad u_1 < u_2 &\Rightarrow \sigma(u_1) < \sigma(u_2) \quad
	&&\forall \; u_1, u_2 \in \cU,\\
	\sigma(\pi'(c)) &= (A, m_A + c - 1) \quad
	&\text{and} \quad u'_1 < u'_2 &\Rightarrow \sigma'(u'_1) < \sigma'(u'_2) \quad
	&&\forall \; u'_1, u'_2 \in \cU'.
\end{aligned}
\label{eq:fullcrossing}
\end{equation}

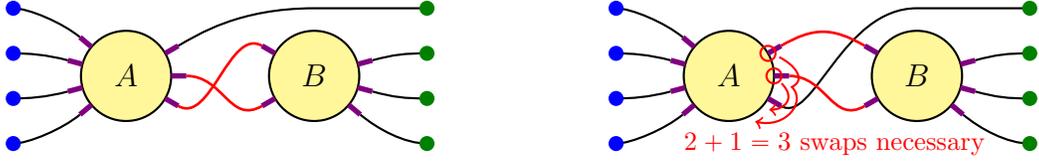
\begin{figure}
	\centering
	\hspace{-1cm}
	\scalebox{1.0}{\def\r{0.6} 
\def\rB{0.8} 
\def\rC{0.5} 
\def\extl(#1, #2){\fill (#1, #2) circle (0.05); \fill[opacity = 0.3, blue] (#1, #2) circle (0.1);  } 
\def\extr(#1, #2){\fill (#1, #2) circle (0.05); \fill[opacity = 0.3, green!50!black] (#1, #2) circle (0.1);  } 
\def\conn(#1, #2, #3, #4){({#4*cos(#1) + #2},{#4*sin(#1) + #3}) --  ({(#4+0.2)*cos(#1) + #2},{(#4+0.2)*sin(#1) + #3})} 
\def\connt(#1, #2, #3, #4, #5){({#4*cos(#1) + #2},{#4*sin(#1) + #3}) .. controls  ({(#4+#5)*cos(#1) + #2},{(#4+#5)*sin(#1) + #3}) and} 
\def\connte(#1, #2, #3, #4, #5){ ({(#4+#5)*cos(#1) + #2},{(#4+#5)*sin(#1) + #3}) .. ({#4*cos(#1) + #2},{#4*sin(#1) + #3}) } 
\begin{tikzpicture}
\useasboundingbox (-2,-1) rectangle (4.5,1);

\draw[thick] \connt(140, 0, 0, (\r+0.2), 0.3) ++(0.2,0) .. (-1.5,0.9); \extl(-1.5, 0.9)
\draw[thick] \connt(165, 0, 0, (\r+0.2), 0.3) ++(0.2,0) .. (-1.5,0.3); \extl(-1.5, 0.3)
\draw[thick] \connt(195, 0, 0, (\r+0.2), 0.3) ++(0.2,0) .. (-1.5,-0.3); \extl(-1.5, -0.3)
\draw[thick] \connt(220, 0, 0, (\r+0.2), 0.3) ++(0.2,0) .. (-1.5,-0.9); \extl(-1.5, -0.9)

\draw[thick] \connt(30, 0, 0, (\r+0.2), 0.3) ++(-1,0) .. (2.5,0.9) -- (4,0.9); \extr(4, 0.9)
\draw[thick] \connt(15, 2.5, 0, (\r+0.2), 0.3) ++(-0.3,0) .. (4,0.3); \extr(4, 0.3)
\draw[thick] \connt(-15, 2.5, 0, (\r+0.2), 0.3) ++(-0.3,0) .. (4,-0.3); \extr(4, -0.3)
\draw[thick] \connt(-40, 2.5, 0, (\r+0.2), 0.3) ++(-0.3,0) .. (4,-0.9); \extr(4, -0.9)

\draw[red, opacity = .8, line width = 1] \connt(0, 0, 0, (\r+0.2), 0.5) \connte(210, 2.5, 0, (\r+0.2), 0.5);
\draw[red, opacity = .8, line width = 1] \connt(-30, 0, 0, (\r+0.2), 0.5) \connte(150, 2.5, 0, (\r+0.2), 0.5);

\filldraw[fill = yellow!50!white, thick] (0,0) circle (\r) node{$A$} ;

\draw[line width = 2, red!50!blue] \conn(140, 0, 0, \r);
\draw[line width = 2, red!50!blue] \conn(165, 0, 0, \r);
\draw[line width = 2, red!50!blue] \conn(195, 0, 0, \r);
\draw[line width = 2, red!50!blue] \conn(220, 0, 0, \r);

\draw[line width = 2, red!50!blue] \conn(30, 0, 0, \r);
\draw[line width = 2, red!50!blue] \conn(0, 0, 0, \r);
\draw[line width = 2, red!50!blue] \conn(-30, 0, 0, \r);

\filldraw[fill = yellow!50!white, thick] (2.5,0) circle (\r) node{$B$} ;

\draw[line width = 2, red!50!blue] \conn(150, 2.5, 0, \r);
\draw[line width = 2, red!50!blue] \conn(210, 2.5, 0, \r);

\draw[line width = 2, red!50!blue] \conn(15, 2.5, 0, \r);
\draw[line width = 2, red!50!blue] \conn(-15, 2.5, 0, \r);
\draw[line width = 2, red!50!blue] \conn(-40, 2.5, 0, \r);

\end{tikzpicture}}
	\scalebox{1.0}{\def\r{0.6} 
\def\rB{0.8} 
\def\rC{0.5} 
\def\extl(#1, #2){\fill (#1, #2) circle (0.05); \fill[opacity = 0.3, blue] (#1, #2) circle (0.1);  } 
\def\extr(#1, #2){\fill (#1, #2) circle (0.05); \fill[opacity = 0.3, green!50!black] (#1, #2) circle (0.1);  } 
\def\conn(#1, #2, #3, #4){({#4*cos(#1) + #2},{#4*sin(#1) + #3}) --  ({(#4+0.2)*cos(#1) + #2},{(#4+0.2)*sin(#1) + #3})} 
\def\connt(#1, #2, #3, #4, #5){({#4*cos(#1) + #2},{#4*sin(#1) + #3}) .. controls  ({(#4+#5)*cos(#1) + #2},{(#4+#5)*sin(#1) + #3}) and} 
\def\connp(#1, #2, #3, #4, #5){({ #4*cos(#1) + #2},{#4*sin(#1) + #3}) } 
\def\connte(#1, #2, #3, #4, #5){ ({(#4+#5)*cos(#1) + #2},{(#4+#5)*sin(#1) + #3}) .. ({#4*cos(#1) + #2},{#4*sin(#1) + #3}) } 
\begin{tikzpicture}
\useasboundingbox (-2,-1) rectangle (4.5,1);

\draw[thick] \connt(140, 0, 0, (\r+0.2), 0.3) ++(0.2,0) .. (-1.5,0.9); \extl(-1.5, 0.9)
\draw[thick] \connt(165, 0, 0, (\r+0.2), 0.3) ++(0.2,0) .. (-1.5,0.3); \extl(-1.5, 0.3)
\draw[thick] \connt(195, 0, 0, (\r+0.2), 0.3) ++(0.2,0) .. (-1.5,-0.3); \extl(-1.5, -0.3)
\draw[thick] \connt(220, 0, 0, (\r+0.2), 0.3) ++(0.2,0) .. (-1.5,-0.9); \extl(-1.5, -0.9)

\draw[thick] \connt(-30, 0, 0, (\r+0.2), 0.5) ++(-1,0) .. (2.5,0.9) -- (4,0.9); \extr(4, 0.9)
\draw[thick] \connt(15, 2.5, 0, (\r+0.2), 0.3) ++(-0.3,0) .. (4,0.3); \extr(4, 0.3)
\draw[thick] \connt(-15, 2.5, 0, (\r+0.2), 0.3) ++(-0.3,0) .. (4,-0.3); \extr(4, -0.3)
\draw[thick] \connt(-40, 2.5, 0, (\r+0.2), 0.3) ++(-0.3,0) .. (4,-0.9); \extr(4, -0.9)

\draw[red, opacity = .8, line width = 1] \connt(30, 0, 0, (\r+0.2), 0.5) \connte(150, 2.5, 0, (\r+0.2), 0.5);
\draw[red, opacity = .8, line width = 1] \connt(0, 0, 0, (\r+0.2), 0.5) \connte(210, 2.5, 0, (\r+0.2), 0.5);

\filldraw[fill = yellow!50!white, thick] (0,0) circle (\r) node{$A$} ;

\draw[line width = 2, red!50!blue] \conn(140, 0, 0, \r);
\draw[line width = 2, red!50!blue] \conn(165, 0, 0, \r);
\draw[line width = 2, red!50!blue] \conn(195, 0, 0, \r);
\draw[line width = 2, red!50!blue] \conn(220, 0, 0, \r);

\draw[line width = 2, red!50!blue] \conn(30, 0, 0, \r);
\draw[line width = 2, red!50!blue] \conn(0, 0, 0, \r);
\draw[line width = 2, red!50!blue] \conn(-30, 0, 0, \r);

\filldraw[fill = yellow!50!white, thick] (2.5,0) circle (\r) node{$B$} ;

\draw[line width = 2, red!50!blue] \conn(150, 2.5, 0, \r);
\draw[line width = 2, red!50!blue] \conn(210, 2.5, 0, \r);

\draw[line width = 2, red!50!blue] \conn(15, 2.5, 0, \r);
\draw[line width = 2, red!50!blue] \conn(-15, 2.5, 0, \r);
\draw[line width = 2, red!50!blue] \conn(-40, 2.5, 0, \r);

\draw[red, opacity = .8, thick] \connp(0, 0, 0, \r, 0 ) circle (0.1);
\draw[red, opacity = .8, thick, ->] \connp(0, 0.1, -0.1, \r, 0 ) .. controls ++(0.2,-0.2)  and ++(0.2,0)  .. \connp(-40, 0, 0, (\r+0.1), 0) ;
\draw[red, opacity = .8, thick] \connp(30, 0, 0, \r, 0) circle (0.1);
\draw[red, opacity = .8, thick, ->] \connp(30, 0.15, -0.05, \r, 0 ) .. controls ++(0.2,-0.1)  and ++(0.2,0.1)  .. \connp(-10, 0, 0, (\r+0.25), 0) .. controls ++(0.2,-0.2)  and ++(0.4,-0.1)  .. \connp(-60, 0, 0, (\r+0.1), 0) ;
\node[red, opacity = .8] at (1.4,-0.9) {\footnotesize $2 + 1 = 3$ swaps necessary};

\end{tikzpicture}}
	\caption{Left: A maximally crossed Friedrichs diagram.\\ Right: There are 3 swaps necessary in $ \sigma' $ to achieve a maximally crossed form. So $ \sgn(\sigma') = -1 $.}
	\label{fig:Friedrichs_maxcrossing}
\end{figure}

By $ \sgn(\sigma), \sgn(\sigma') \in \{1, -1\} $, we denote the signs of these permutations (see also Figure \ref{fig:Friedrichs_maxcrossing}). Then, the sign of the config $ (\pi, \pi') \in \cC $ is defined as
\begin{equation}
	\sgn(\pi, \pi') = \sgn(\pi, \pi', m_A, n_B) := (-1)^{(m_A - C)(n_B - C)} \cdot \sgn(\sigma) \cdot \sgn(\sigma').
\label{eq:sgnpipi}
\end{equation}
In Appendix \ref{app:heuristiccontfer}, a heuristic motivation for how this sign function arises can be found. We now define the \textbf{fermionic attached product} as
\begin{equation}
\begin{aligned}
	A \contfer B := \sum_{(\pi, \pi') \in \cC} \sgn(\pi, \pi') \int &f_A(\bX_A, \bY_A) f_B(\bX_B, \bY_B)
	\prod_{c = 1}^C \delta(\bx_{\pi(c)} - \by_{\pi'(c)}) \times \\
	&\times \left( \prod_{\ell = 1}^{n_A} a_{A, \ell} \right)^*
	\left( \prod_{u \in \cU} a_u \right)^*
	\prod_{u' \in \cU'} a_{u'}
	\prod_{\ell' = 1}^{m_B} a_{B, \ell'} \; \mathrm{d}\bX \mathrm{d}\bY.
\end{aligned}
\label{eq:contfer}
\end{equation}
Finally, we define the \textbf{normal ordered product} as
\begin{equation}
\begin{aligned}
	\normalordered{AB} \;:= \int &f_A(\bX_A, \bY_A) f_B(\bX_B, \bY_B)\times \\
	&\times \left( \prod_{\ell = 1}^{n_A} a_{A, \ell} \right)^*
	\left( \prod_{\ell = 1}^{n_B} a_{B, \ell} \right)^*
	\left( \prod_{\ell' = 1}^{m_A} a_{A, \ell'} \right)
	\left( \prod_{\ell' = 1}^{m_B} a_{B, \ell'} \right) \; \mathrm{d}\bX \mathrm{d}\bY.
\end{aligned}
\label{eq:normalordered}
\end{equation}

\section{Main Results}	
\label{sec:mainresults}

Recall the definitions of the bosonic/fermionic attached products $ A \cont B $ \eqref{eq:cont} and $ A \contfer B $ \eqref{eq:contfer}, as well as the normal ordered product $ \normalordered{AB} $ \eqref{eq:normalordered}. Our two main commutator formulas are now the following.

\begin{theorem}[Bosonic Commutator Formula]
Consider $ A, B \in \cA_+ $ of the form \eqref{eq:A}, i.e., the CCR hold. Then,
\begin{equation}
	[A, B] = A \cont B - B \cont A.
\label{eq:contthm}
\end{equation}
\label{thm:cont}
\end{theorem}

\begin{theorem}[Fermionic Commutator Formula]
Consider $ A, B \in \cA_- $ of the form \eqref{eq:A}, i.e., the CAR hold. Then,
\begin{equation}
\begin{aligned}
	{[A, B]} &= A \contfer B - B \contfer A \qquad &&\text{\rm if} \; (n_A + m_A)(n_B + m_B) \; \text{\rm is even},\\
	\{A, B\} &= A \contfer B + B \contfer A \qquad &&\text{\rm if} \; (n_A + m_A)(n_B + m_B) \; \text{\rm is odd}.\\
\end{aligned}
\label{eq:contferthm}
\end{equation}
\label{thm:contfer}
\end{theorem}
These theorems are immediate consequences of the following two lemmas:
\begin{lemma}[{Bosonic Contraction Formula, \cite[p.~54]{Friedrichs1965}, \cite[Thm.~1.1]{Hepp1969}}]
Consider $ A, B \in \cA_+ $, i.e., the CCR hold. Then,
\begin{equation}
	AB = \; \normalordered{AB} + A \cont B.
\label{eq:contproduct}
\end{equation}
\label{lem:cont}
\end{lemma}

\begin{lemma}[Fermionic Contraction Formula]
Consider $ A, B \in \cA_- $, i.e., the CAR hold. Then,
\begin{equation}
	AB = (-1)^{m_A n_B} \; \normalordered{AB} + A \contfer B.
\label{eq:contferproduct}
\end{equation}
\label{lem:contfer}
\end{lemma}

\paragraph{Remarks.} 

\begin{enumerate}
\setcounter{enumi}{\theremarks}

\item \label{rem: algorithm} \textit{Algorithm for commutator evaluation}: As formulas \eqref{eq:contthm} and \eqref{eq:contferthm} are rather abstract, let us briefly sketch the algorithms to be used for practical evaluations of (anti--) commutators by means of Friedrichs diagrams. A \textbf{fermionic (anti--) commutator} is evaluated as follows:
\begin{algorithm}[H]
\caption{Fermionic Commutator/Anticommutator Evaluation}
\label{alg:fer}
\begin{algorithmic}[1]\vspace{6pt}
\State Check whether \eqref{eq:contferthm} is valid for commutators or anticommutators: If $ m_A n_B + m_B n_A $ is even, it applies to commutators; if it is odd, then it applies to anticomutators.\vspace{6pt}
\State Draw all diagrams corresponding to $ A \contfer B $ and $ B \contfer A $, i.e., which involve at least one contraction, and translate them into contributions \eqref{eq:G}. If $ m_A n_B + m_B n_A $ is even, then the contributions of $ B \contfer A $ get a factor of $ (-1) $.\vspace{6pt}
\State For each diagram, multiply its contribution by $ (-1) $ for each swap of connectors necessary to take the diagram into a maximally crossed form.\vspace{6pt}
\State For each diagram, multiply its contribution by a ``normal ordering factor'' of $ (-1)^{(m_A - C)(n_B - C)} $ (here, $ (m_A - C)(n_B - C) $ is the product of the numbers of potentially contractable but uncontracted connectors).\vspace{6pt}
\State Add up all contributions.\vspace{6pt}
\end{algorithmic}
\end{algorithm}

In the \textbf{bosonic case}, only \textbf{commutators} may be evaluated by \eqref{eq:contthm}. The algorithm is much simpler:
\begin{algorithm}[H]
\caption{Bosonic Commutator Evaluation}
\label{alg:bos}
\begin{algorithmic}[1]\vspace{6pt}
\State Draw all diagrams corresponding to $ A \cont B $ and $ B \cont A $ and translate them into a contribution \eqref{eq:G}. The contribution $ B \cont A $ always gets a factor of $ (-1) $.\vspace{6pt}
\State Add up all contributions.\vspace{6pt}
\end{algorithmic}
\end{algorithm}

\item \label{rem:distkernel} \textit{Distributions as integral kernels}: In more physical situations, such as many--body systems with Coulomb interaction or QFT models, one may encounter operator products of the form \eqref{eq:A}, where the integral kernel $ f_A $ is not an element of $ L^2 $, but rather a distribution, e.g., a tempered distribution in $ \cS'(\RRR^{(n_A + m_A) d}) $. In that case, a contraction as in \eqref{eq:G} may or may not make sense, depending on whether the distribution multiplication is allowed, see also \cite[Chap.~8]{Hoermander}. If the distribution multiplication is allowed for all contractions in all $ A \cont B $ or $ A \contfer B $, then the respective formulas in Lemmas \ref{lem:cont} and \ref{lem:contfer}, as well as in Theorems \ref{thm:cont} and \ref{thm:contfer} remain valid.

\item \label{rem:contnumber} \textit{Number of contractions}: In each attached product $ A \cont B $ or $ A \contfer B $, the number of contributions, here called $ \cN_{(A,B)} $, is easily obtained combinatorically: If a diagram shall have $ C $ contractions, then there are $ {n_B \choose C} $ choices for the values $ n_B \ge \pi(1) > \ldots > \pi(C) \ge 1 $. The number of choice options for associated indices for contraction $ \pi'(1), \ldots, \pi'(C) $ is $ m_A \cdot (m_A - 1) \cdot \ldots \cdot (m_A - C + 1) $. So the total number of contributions is
\begin{equation}
\begin{aligned}
	\cN_{(A,B)}
	= &\sum_{C = 1}^{\min(m_A, n_B)} {n_B \choose C} \cdot m_A \cdot \ldots \cdot (m_A - C + 1)\\
	= &\sum_{C = 1}^{\min(m_A, n_B)} \frac{n_B! m_A!}{(n_B - C)! C! (m_A - C)!}
	= \sum_{C = 1}^{\min(m_A, n_B)} C! {n_B \choose C} {m_A \choose C}.
\end{aligned}
\label{eq:cNAB}
\end{equation}

\end{enumerate}	
\setcounter{remarks}{\theenumi}

\section{Proofs}	
\label{sec:proofs}

The bosonic contraction formula (Lemma \ref{lem:cont}) was already stated in \cite{Friedrichs1965} and \cite{Hepp1969}, where \cite[p.~13]{Hepp1969} provides the idea for a proof by induction. For completeness, we give the explicit proof by induction here, as well as a second proof based on exponentials of derivatives.

\begin{proof}[Proof of Lemma \ref{lem:cont}]
Plugging in the definitions of $ A, B $ \eqref{eq:A}, $ \normalordered{AB} $ \eqref{eq:normalordered} and $ A \cont B $ \eqref{eq:cont}, we can equivalently transform the statement to be shown \eqref{eq:contproduct} into
\begin{equation}
\begin{aligned}
	&a_{A, 1} \ldots a_{A, m_A} a^*_{B, n_B} \ldots a^*_{B, 1}\\
	= &a^*_{B, n_B} \ldots a^*_{B, 1} a_{A, 1} \ldots a_{A, m_A}
	+ \sum_{(\pi, \pi') \in \cC} \prod_{c = 1}^C \delta(\bx_{\pi(c)} - \by_{\pi'(c)})
	\left( \prod_{u \in \cU} a_u \right)^*
	\prod_{u' \in \cU'} a_{u'}.
\end{aligned}
\label{eq:contproduct2}
\end{equation}
The proof of \eqref{eq:contproduct2} is done by an induction over $ n_B, m_A \in \NNN_0 $. We start with the case where $ n_B = 0 $ or $ m_A = 0 $, corresponding to the axes of the $ m_A $--$ n_B $--quadrant. Then we establish that \eqref{eq:contproduct2} holds for $ (m_A, n_B) $, assuming that it holds for $ (m_A, n_B - 1) $ and $ (m_A - 1, n_B - 1) $. With this induction step, we can ``fill up'' the quadrant of index pairs $ (m_A, n_B) $, for which \eqref{eq:contproduct2} is valid, line by line.\\

On the axes $ n_B = 0 $ or $ m_A = 0 $, the statement is trivially satisfied, since $ AB $ is already normal ordered and the sum over $ \cC $ in \eqref{eq:contproduct2} is empty.\\

For the induction step $ (m_A - 1, n_B - 1) \land (m_A, n_B - 1) \mapsto (m_A, n_B) $, we consider the left--hand side of \eqref{eq:contproduct2} and commute $ a^*_{B, n_B} $ to the very left, before using the induction assumption:
\begin{equation}
\begin{aligned}
	&a_{A, 1} \ldots a_{A, m_A} a^*_{B, n_B} \ldots a^*_{B, 1}\\
	= &\sum_{k = 1}^{m_A} a_{A, 1} \ldots [a_{A, k}, a^*_{B, n_B}] \ldots a_{A, m_A} a^*_{B, n_B-1} \ldots a^*_{B, 1}
	+ a^*_{B, n_B} a_{A, 1} \ldots a_{A, m_A} a^*_{B, n_B-1} \ldots a^*_{B, 1}\\
	= & \underbrace{\sum_{k = 1}^{m_A} \delta(\bx_{B, n_B} - \by_{A, k}) \!\! \left( \sum_{(\pi, \pi') \in \tilde\cC_k} \prod_{c = 1}^C \delta(\bx_{\pi(c)} - \by_{\pi'(c)})
	\Bigg( \! \prod_{u \in \tilde\cU} \! a_u \!\Bigg)^* \!\!\!\!
	\prod_{u' \in \tilde\cU'_k} \!\! a_{u'}
	\! + \! \Bigg( \! \prod_{u \neq (B, n_B)} \!\!\!\!\! a_u \! \Bigg)^* \!\!\!\!
	\prod_{u' \neq (A, k)} \!\!\!\!\! a_{u'} \! \right)}_{\circled{1}}\\
	&+ \underbrace{\sum_{(\pi, \pi') \in \tilde\cC} \prod_{c = 1}^C \delta(\bx_{\pi(c)} - \by_{\pi'(c)}) a^*_{B, n_B}
	\left( \prod_{u \in \tilde\cU} a_u \right)^*
	\prod_{u' \in \cU'} a_{u'}}_{\circled{2}}
	+ \underbrace{a^*_{B, n_B} \ldots a^*_{B, 1} a_{A, 1} \ldots a_{A, m_A}}_{\circled{3}}.
\end{aligned}
\label{eq:evaluatecont}
\end{equation}
Here, we used the modified sets of contraction configs and uncontracted legs
\begin{equation}
\begin{aligned}
	\tilde\cC_k := &\big\{ (\pi, \pi') \; \big\vert \;
	\pi: \{1, \ldots, C\} \to \cI_B \setminus \{ (B, n_B) \}, \;
	\pi': \{1, \ldots, C\} \to \cI'_A \setminus \{ (A, k) \},\\
	&\qquad \qquad 1 \le C \le \min(m_A, n_B) - 1, \;|\mathrm{imag}(\pi')|=C,\; 
	\pi(1) > \ldots > \pi(C) \big\},\\
	\tilde\cC := &\big\{ (\pi, \pi') \; \big\vert \;
	\pi: \{1, \ldots, C\} \to \cI_B \setminus \{ (B, n_B) \}, \;
	\pi': \{1, \ldots, C\} \to \cI'_A,\\
	&\qquad \qquad 1 \le C \le \min(m_A, n_B - 1), \; |\mathrm{imag}(\pi')|=C, \;
	\pi(1) > \ldots > \pi(C) \big\},\\
\end{aligned}
\label{eq:tildecC}
\end{equation}
\begin{equation}
\begin{aligned}
	\tilde\cU'_k = \tilde\cU'_{k, \pi'} := &\big\{ (A, k') \neq (A, k) \; \mid \nexists \; c \in \{1, \ldots, C\} : \pi'(c) = (A, k') \big\} \subseteq \cI'_A,\\
	\tilde\cU = \tilde\cU_\pi := &\big\{ (B, j) \neq (B, n_B) \; \mid \nexists\;  c \in \{1, \ldots, C\} : \pi(c) = (B, j) \big\} \subseteq \cI_B.\\
\end{aligned}
\label{eq:tildecU}
\end{equation}
Now in \eqref{eq:evaluatecont}, the term \circled{1} comprises all those contraction configs between all coordinates of $ \cI_A' $ and $ \cI_B $ into which $ (B, n_B) $ is involved: The sum over $ \tilde\cC_k $ contains all terms where there is at least one further contraction in addition to the one involving $ (B, n_B) $, and the operator product over $ u \neq (B, n_B), u' \neq (A, k) $ is that term, where $ (B, n_B) $ is involved into the only contraction.\\
The term \circled{2} comprises all configs of $ \cI_A' $ and $ \cI_B $, into which $ (B, n_B) $ is not involved. So $ \circled{1} + \circled{2} $ just amounts to all contraction configs, i.e., to the second term on the r.h.s of \eqref{eq:contproduct2}.\\
Finally, \circled{3} is the normal ordered first term on the r.h.s of \eqref{eq:contproduct2}. This establishes \eqref{eq:contproduct2} for $ (m_A, n_B) $ and thus the induction step.\\
\end{proof}

\begin{proof}[Proof of Theorem \ref{thm:cont}]
Keeping in mind that definition \eqref{eq:normalordered} entails $ \normalordered{AB} = \normalordered{BA} $, Lemma \ref{lem:cont} readily implies
\begin{equation}
	[A, B] = AB - BA = \; \normalordered{AB} + A \cont B - \normalordered{BA} - B \cont A = A \cont B - B \cont A.
\end{equation}
\end{proof}

Before we are going to verify the fermionic case, let us present an alternative proof of Lemma \ref{lem:cont} based on coherent state techniques.

\begin{proof}[Alternative proof of Lemma \ref{lem:cont}]
As in the previous proof, let us assume w.l.o.g. that $A$ and $B$ are of the form $A=a_{J_1} \ldots a_{J_{m_A}} $ and $B=a^*_{I_{n_B}} \ldots a^*_{I_1}$. Since $A$ and $B$ only involve the finitely many modes $\{J_1,\dots ,J_{m_A},I_1,\dots I_{n_B}\}$, we can further assume that the one particle Hilbert space is finite dimensional, i.e., given by $\mathfrak{h}=\mathbb{C}^N$. In the following, let $|z\rangle $ denote the coherent state corresponding to $z\in \mathbb{C}^N$ defined (modulo a complex phase) by the eigen-equations $a_k |z\rangle=z_k |z\rangle$, and let us define the lower symbol of an operator $X$ as the function $z\mapsto \langle{z|X|z\rangle}$. Clearly the lower symbol of the operator $\normalordered{AB}$ is given by $\langle z|\normalordered{AB}|z\rangle=\prod_{i=1}^{n_B}\overline{z}_{I_i}\prod_{j=1}^{m_A}z_{J_j}$, and therefore the lower symbol of $AB$ reads
\begin{equation}
\label{eq: symbol}
\langle z|AB|z\rangle\! =\! e^{\nabla_{\overline{z}}\nabla_z}\! \! \left(\prod_{i=1}^{n_B}\overline{z}_{I_i}\prod_{j=1}^{m_A}z_{J_j}\! \right)\! \! =\!  \langle z| \normalordered{AB} |z\rangle\, +\! \! \! \! \! \sum_{C=1}^{\min\{n_B,m_A\}} \! \frac{\left(\nabla_{\overline{z}}\nabla_z\right)^C}{C!} \! \left(\prod_{i=1}^{n_B}\overline{z}_{I_i}\prod_{j=1}^{m_A}z_{J_j}\! \right),
\end{equation}
see e.g., Proposition 4.4 in \cite{Derezinski}. We are going to prove by induction in $ C $ that 
\begin{equation}
\label{eq: coherent}
    \left(\nabla_{\overline{z}}\nabla_z\right)^C \left(\prod_{i=1}^{n_B}\overline{z}_{I_i}\prod_{j=1}^{m_A}z_{J_j}\! \right)=\sum_{\pi\in \cC^B_C,\pi'\in \cC^A_C}f_{\pi,\pi'}(z),
\end{equation}
where $\cC^A_C := \big\{ \pi' \; \big\vert \;
	\pi': \{1, \ldots, C \} \to \{1,\ldots , m_A\},\, \text{$\pi'$ is injective}
	 \big\}$, $\cC^B_C := \big\{ \pi \; \big\vert \;
	\pi: \{1, \ldots, C \} \to \{1,\ldots , n_B\},\, \text{$\pi$ is injective} \big\}$ and $f_{\pi,\pi'}$ is defined as
 \begin{equation}
     f_{\pi,\pi'}(z):=\prod_{p=1}^C \delta\left(I_{\pi(p)}-J_{\pi'(p)} \right)\prod_{i\notin \mathrm{imag}(\pi) }\overline{z}_{I_i}\prod_{j\notin \mathrm{imag}(\pi')} z_{J_j}.
 \end{equation}
 The case $C=0$ is trivial. The induction step $C\mapsto C+1$ follows from the observation that we can express the term $\left(\nabla_{\overline{z}}\nabla_z\right)^{C+1} \left(\prod_{i=1}^{n_B}\overline{z}_{I_i}\prod_{j=1}^{m_A}z_{J_j}\! \right)= \sum_{\pi\in \cC^B_C,\pi'\in \cC^A_C}\nabla_{\overline{z}}\nabla_z f_{\pi,\pi'}(z)$ as
 \begin{equation}
 \begin{aligned}
    \sum_{\pi\in \cC^B_C,\pi'\in \cC^A_C}& \prod_{p=1}^C \delta\left(I_{\pi(p)}-J_{\pi'(p)} \right)\! \! \! \sum_{\ell \notin \mathrm{imag}(\pi),k \notin \mathrm{imag}(\pi')}\! \! \! \delta\left(I_{\ell}-J_{k} \right)\! \! \! \prod_{i\notin \mathrm{imag}(\pi) \cup \{\ell\}}\! \! \! \overline{z}_{I_i}\! \! \! \prod_{j\notin \mathrm{imag}(\pi')\cup \{k\}}\! \! \!  z_{J_j}\\
    &=\sum_{\pi\in \cC^B_{C},\pi'\in \cC^A_{C}}\sum_{\ell \notin \mathrm{imag}(\pi),k \notin \mathrm{imag}(\pi')} f_{\pi_\ell,\pi'_k}(z)=\sum_{\Pi\in \cC^B_{C+1},\Pi'\in \cC^A_{C+1}}f_{\Pi,\Pi'}(z),
 \end{aligned}
 \end{equation}
 where $\pi_\ell:\{1, \ldots, C+1 \} \to \{1,\ldots , n_B\}$ is an extension of the function $\pi$ by $\pi_\ell(C+1):=\ell$, and $\pi'_k$ is an extension of the function $\pi'$ by $\pi'_k(C+1):=k$. This finishes the induction.\\
 Combining Eq.~(\ref{eq: coherent}) with the observation that any pair $(\Pi,\Pi')\in \cC^B_C\times \cC^A_C$ can be uniquely written as $\Pi=\tilde{\pi}\circ \sigma$ and $\Pi'=\tilde{\pi}'\circ \sigma$, where $(\tilde{\pi},\tilde{\pi}')\in \cC_C:=\big\{ (\tilde{\pi},\tilde{\pi}')\in \cC \; \big\vert \; |\mathrm{imag}(\tilde{\pi})|=C \big\}$ and $\sigma\in S_C$ is a permutation of the set $\{1,\ldots,C\}$, yields
 \begin{equation}
     \frac{\left(\nabla_{\overline{z}}\nabla_z\right)^C}{C!} \! \left(\prod_{i=1}^{n_B}\overline{z}_{I_i}\prod_{j=1}^{m_A}z_{J_j}\! \right)=\sum_{(\tilde{\pi},\tilde{\pi}')\in \cC_C}\sum_{\sigma\in S_C}\frac{1}{C!}f_{\tilde{\pi}\circ \sigma ,\tilde{\pi}'\circ \sigma }(z)=\sum_{(\tilde{\pi},\tilde{\pi}')\in \cC_C}f_{\tilde{\pi} ,\tilde{\pi}' }(z),
 \end{equation}
 where we have used $f_{\tilde{\pi}\circ \sigma ,\tilde{\pi}'\circ \sigma }=f_{\tilde{\pi} ,\tilde{\pi}' }$. Since $\sum_{(\tilde{\pi},\tilde{\pi}')\in \cC}f_{\tilde{\pi} ,\tilde{\pi}' }(z)$ is the lower symbol of the operator $ A \cont B =\sum_{(\tilde{\pi},\tilde{\pi}')\in \cC}\prod_{p=1}^C \delta\left(I_{\pi(p)}-J_{\pi'(p)} \right) \left( \prod_{i\notin \mathrm{imag}(\tilde{\pi}) }a_{I_i} \right)^* \prod_{j\notin \mathrm{imag}(\tilde{\pi}')} a_{J_j}$, we obtain by Eq.~(\ref{eq: symbol}) that $\langle z|AB|z\rangle=\langle z|\normalordered{AB}|z\rangle+\langle z|A \cont B|z\rangle$ for all $z\in \mathbb{C}^N$, and therefore $AB=\normalordered{AB}+A \cont B$.
\end{proof}

The proofs for the fermionic cases are similar to those for the bosonic cases,  subject to sign changes.\\

\begin{proof}[Proof of Lemma \ref{lem:contfer}]
The definitions of $ A, B $ \eqref{eq:A}, $ \normalordered{AB} $ \eqref{eq:normalordered} and $ A \contfer B $ \eqref{eq:contfer} allow to equivalently reformulate the target formula \eqref{eq:contferproduct} into
\begin{equation}
\begin{aligned}
	a_{A, 1} \ldots a_{A, m_A} a^*_{B, n_B} \ldots a^*_{B, 1}
	= & (-1)^{m_A n_B} a^*_{B, n_B} \ldots a^*_{B, 1} a_{A, 1} \ldots a_{A, m_A}\\
	&+ \sum_{(\pi, \pi') \in \cC} \sgn(\pi, \pi') \prod_{c = 1}^C \delta(\bx_{\pi(c)} - \by_{\pi'(c)})
	\left( \prod_{u \in \cU} a_u \right)^*
	\prod_{u '\in \cU'} a_{u'},
\end{aligned}
\label{eq:contferproduct2}
\end{equation}
with $ \sgn(\pi, \pi') = \sgn(\pi, \pi', m_A, n_B) $ given by \eqref{eq:sgnpipi}. As in the proof of Lemma \ref{lem:cont}, we perform an induction over $ m_A, n_B \in \NNN_0 $, which amounts to verifying \eqref{eq:contferproduct2} for $ m_A = 0 $ or $ n_B = 0 $ and then establishing the induction step $ (m_A - 1, n_B - 1) \land (m_A, n_B - 1) \mapsto (m_A, n_B) $.\\

The cases $ m_A = 0 $ and $ n_B = 0 $ are trivial, since the sum over $ (\pi, \pi') \in \cC $ is empty.\\

The induction step $ (m_A - 1, n_B - 1) \land (m_A, n_B - 1) \mapsto (m_A, n_B) $ is performed in similarity to \eqref{eq:evaluatecont}. We pull $ a^*_{B, n_B} $ to the left, which generates several anticommutators
\begin{equation}
\begin{aligned}
	&a_{A, 1} \ldots a_{A, m_A} a^*_{B, n_B} \ldots a^*_{B, 1}\\
	= &\sum_{k = 1}^{m_A} (-1)^{m_A - k} a_{A, 1} \ldots \{ a_{A, k}, a^*_{B, n_B} \} \ldots a_{A, m_A} a^*_{B, n_B-1} \ldots a^*_{B, 1}\\
	& \quad + (-1)^{m_A} a^*_{B, n_B} a_{A, 1} \ldots a_{A, m_A} a^*_{B, n_B-1} \ldots a^*_{B, 1}\\
	= & \sum_{k = 1}^{m_A} (-1)^{m_A - k} \delta(\bx_{B, n_B} - \by_{A, k}) \Bigg( \sum_{(\pi, \pi') \in \tilde\cC_k} \sgn(\pi, \pi', m_A - 1, n_B - 1) \prod_{c = 1}^C \delta(\bx_{\pi(c)} - \by_{\pi'(c)}) \times\\
	&\quad \times \left( \prod_{u \in \tilde\cU} a_u \right)^*
	\prod_{u' \in \tilde\cU'_k} a_{u'}
	+ (-1)^{(m_A - 1)(n_B - 1)} \Bigg( \prod_{u \neq (B, n_B)} a_u \Bigg)^*
	\prod_{u' \neq (A, m_A)} a_{u'} \Bigg)\\
	&\quad + (-1)^{m_A} \sum_{(\pi, \pi') \in \tilde\cC} \sgn(\pi, \pi', m_A, n_B - 1) \prod_{c = 1}^C \delta(\bx_{\pi(c)} - \by_{\pi'(c)}) a^*_{B, n_B}
	\left( \prod_{u \in \tilde\cU} a_u \right)^*
	\prod_{u '\in \cU'} a_{u'}\\
	&\quad + (-1)^{m_A} (-1)^{m_A (n_B - 1)} a^*_{B, n_B} \ldots a^*_{B, 1} a_{A, 1} \ldots a_{A, m_A},
\end{aligned}
\label{eq:evaluatecontfer}
\end{equation}
using the contraction config sets $ \tilde\cC_k, \tilde\cC $ \eqref{eq:tildecC} and the uncontracted leg sets $ \tilde\cU'_k, \tilde\cU $ \eqref{eq:tildecU}. In the last step, we applied \eqref{eq:contferproduct2} with coefficient pairs $ (m_A - 1, n_B - 1) $ and $ (m_A, n_B - 1) $, which makes it necessary to adjust the respective indices in $ \sgn(\pi, \pi') $.\\
We now include the contraction $ \delta(\bx_{B, n_B} - \by_{A, k}) $ into each config in $ \tilde\cC_k $. This way, the set of configs $ \tilde\cC_k $ changes to
\begin{equation}
\begin{aligned}
	\tilde\cC_{k, +} := &\big\{ (\pi_+, \pi'_+) \; \big\vert \;
	\pi_+: \{1, \ldots, C\} \to \cI_B, \;
	\pi'_+: \{1, \ldots, C\} \to \cI'_A , \;
	\pi_+(1) = (B, n_B),\\
	&\;  \pi'_+(1) = (A, k), \;
	2 \le C \le \min(m_A, n_B), \; |\mathrm{imag}(\pi')|=C, \;
	\pi_+(1) > \ldots > \pi_+(C) \big\}.\\
\end{aligned}
\label{eq:tildecCk+}
\end{equation}
Concerning the sign of $ (\pi_+, \pi'_+) $, it is easy to see that under the change $ \tilde\cC_k \mapsto \tilde\cC_{k, +} $, the uncontracted leg number product $ (m_A - C)(n_B - C) $ stays invariant. Since $ \pi_+(1) = (B, n_B) $, it is not necessary to additionally swap left--connectors of $ B $ to get a maximally crossed form, so $ \sgn(\sigma) $ also stays invariant. However, in order to achieve a maximally crossed form, the right--connector of $ A $ corresponding to $ (A, k) $ must be moved by $ (m_A - k) $ positions down to get to the bottom, before the original $ \sigma' $ can be applied for bringing all other connectors into the correct position (see Figure \ref{fig:Friedrichs_piplus}). So the sign of $ (\pi_+, \pi'_+) $ amounts to
\begin{equation}
	\sgn(\pi_+, \pi'_+, m_A, n_B) = (-1)^{m_A - k} \sgn(\pi, \pi', m_A - 1, n_B - 1).
\label{eq:sgnpi+pi+}
\end{equation}

\begin{figure}[ht]
	\centering
	\hspace{-2.3cm}
	\scalebox{0.9}{\def\r{0.6} 
\def\rB{0.8} 
\def\rC{0.5} 
\def\extl(#1, #2){\fill (#1, #2) circle (0.05); \fill[opacity = 0.3, blue] (#1, #2) circle (0.1);  } 
\def\extr(#1, #2){\fill (#1, #2) circle (0.05); \fill[opacity = 0.3, green!50!black] (#1, #2) circle (0.1);  } 
\def\conn(#1, #2, #3, #4){({#4*cos(#1) + #2},{#4*sin(#1) + #3}) --  ({(#4+0.2)*cos(#1) + #2},{(#4+0.2)*sin(#1) + #3})} 
\def\connp(#1, #2, #3, #4, #5){({ #4*cos(#1) + #2},{#4*sin(#1) + #3}) } 
\def\connt(#1, #2, #3, #4, #5){({#4*cos(#1) + #2},{#4*sin(#1) + #3}) .. controls  ({(#4+#5)*cos(#1) + #2},{(#4+#5)*sin(#1) + #3}) and} 
\def\connte(#1, #2, #3, #4, #5){ ({(#4+#5)*cos(#1) + #2},{(#4+#5)*sin(#1) + #3}) .. ({#4*cos(#1) + #2},{#4*sin(#1) + #3}) } 
\begin{tikzpicture}
\useasboundingbox (-2,-1.2) rectangle (4.5,1.2);

\draw[thick] \connt(130, 0, 0, (\r+0.2), 0.3) ++(0.2,0) .. (-1.2,1); \extl(-1.2, 1)
\draw[thick] \connt(150, 0, 0, (\r+0.2), 0.3) ++(0.2,0) .. (-1.2,0.6); \extl(-1.2, 0.6)
\draw[thick] \connt(170, 0, 0, (\r+0.2), 0.3) ++(0.2,0) .. (-1.2,0.2); \extl(-1.2, 0.2)
\draw[thick] \connt(190, 0, 0, (\r+0.2), 0.3) ++(0.2,0) .. (-1.2,-0.2); \extl(-1.2, -0.2)
\draw[thick] \connt(210, 0, 0, (\r+0.2), 0.3) ++(0.2,0) .. (-1.2,-0.6); \extl(-1.2, -0.6)
\draw[thick] \connt(190, 2.5, 0, (\r+0.2), 0.3) ++(1,0) .. (0,-1) -- (-1.2,-1); \extl(-1.2, -1.0)

\draw[thick] \connt(-10, 0, 0, (\r+0.2), 0.3) ++(-1,0) .. (2,1) -- (3.7,1); \extr(3.7, 1)
\draw[thick] \connt(-50, 0, 0, (\r+0.2), 0.8) ++(-1,0) .. (2,0.7) -- (3.7,0.7); \extr(3.7, 0.7)
\draw[thick] \connt(10, 2.5, 0, (\r+0.2), 0.3) ++(-0.2,0) .. (3.7,0.2); \extr(3.7, 0.2)
\draw[thick] \connt(-10, 2.5, 0, (\r+0.2), 0.3) ++(-0.2,0) .. (3.7,-0.2); \extr(3.7, -0.2)
\draw[thick] \connt(-30, 2.5, 0, (\r+0.2), 0.3) ++(-0.2,0) .. (3.7,-0.6); \extr(3.7, -0.6)
\draw[thick] \connt(-50, 2.5, 0, (\r+0.2), 0.3) ++(-0.2,0) .. (3.7,-1); \extr(3.7, -1)

\draw[red, opacity = .8, line width = 1] \connt(50, 0, 0, (\r+0.2), 0.5) \connte(230, 2.5, 0, (\r+0.2), 0.5);
\draw[red, opacity = .8, line width = 1] \connt(-30, 0, 0, (\r+0.2), 0.5) \connte(210, 2.5, 0, (\r+0.2), 0.5);
\draw[red, opacity = .8, line width = 1] \connt(30, 0, 0, (\r+0.2), 0.5) \connte(170, 2.5, 0, (\r+0.2), 0.5);
\draw[red, opacity = .8, line width = 1, dashed] \connt(10, 0, 0, (\r+0.2), 0.5) \connte(150, 2.5, 0, (\r+0.2), 0.5);

\filldraw[fill = yellow!50!white, thick] (0,0) circle (\r) node{$A$} ;

\draw[line width = 2, red!50!blue] \conn(130, 0, 0, \r);
\draw[line width = 2, red!50!blue] \conn(150, 0, 0, \r);
\draw[line width = 2, red!50!blue] \conn(170, 0, 0, \r);
\draw[line width = 2, red!50!blue] \conn(190, 0, 0, \r);
\draw[line width = 2, red!50!blue] \conn(210, 0, 0, \r);

\draw[line width = 2, red!50!blue] \conn(50, 0, 0, \r);
\draw[line width = 2, red!50!blue] \conn(30, 0, 0, \r);
\draw[line width = 2, red!50!blue] \conn(10, 0, 0, \r);
\draw[line width = 2, red!50!blue] \conn(-10, 0, 0, \r);
\draw[line width = 2, red!50!blue] \conn(-30, 0, 0, \r);
\draw[line width = 2, red!50!blue] \conn(-50, 0, 0, \r);

\filldraw[fill = yellow!50!white, thick] (2.5,0) circle (\r) node{$B$} ;

\draw[line width = 2, red!50!blue] \conn(150, 2.5, 0, \r);
\draw[line width = 2, red!50!blue] \conn(170, 2.5, 0, \r);
\draw[line width = 2, red!50!blue] \conn(190, 2.5, 0, \r);
\draw[line width = 2, red!50!blue] \conn(210, 2.5, 0, \r);
\draw[line width = 2, red!50!blue] \conn(230, 2.5, 0, \r);

\draw[line width = 2, red!50!blue] \conn(10, 2.5, 0, \r);
\draw[line width = 2, red!50!blue] \conn(-10, 2.5, 0, \r);
\draw[line width = 2, red!50!blue] \conn(-30, 2.5, 0, \r);
\draw[line width = 2, red!50!blue] \conn(-50, 2.5, 0, \r);

\draw[red!50!blue, line width = 1] \connp(10, 0, 0, \r, 0) circle (0.1);
\draw[red!50!blue] (0.5,0.14) -- ++(-0.4,0.6) node[anchor = south] {\scriptsize $ (A, k) $};
\draw[red!50!blue, line width = 1] \connp(150, 2.5, 0, \r, 0) circle (0.1);
\draw[red!50!blue] (2.06,0.24) -- ++(0.3,-1) node[anchor = north] {\scriptsize $ (B, n_B) $};
\node[red!50!blue] at (-0.2,-0.8) {\scriptsize 3 swaps};
\draw[red!50!blue, line width = 1, ->] (0.45,0.05) .. controls ++(-0.2,-0.1) and ++(-0.1,0.2) .. (0.34,-0.36);
\end{tikzpicture}} \hspace{-1.9cm}
	\scalebox{0.9}{\def\r{0.6} 
\def\rB{0.8} 
\def\rC{0.5} 
\def\extl(#1, #2){\fill (#1, #2) circle (0.05); \fill[opacity = 0.3, blue] (#1, #2) circle (0.1);  } 
\def\extr(#1, #2){\fill (#1, #2) circle (0.05); \fill[opacity = 0.3, green!50!black] (#1, #2) circle (0.1);  } 
\def\conn(#1, #2, #3, #4){({#4*cos(#1) + #2},{#4*sin(#1) + #3}) --  ({(#4+0.2)*cos(#1) + #2},{(#4+0.2)*sin(#1) + #3})} 
\def\connt(#1, #2, #3, #4, #5){({#4*cos(#1) + #2},{#4*sin(#1) + #3}) .. controls  ({(#4+#5)*cos(#1) + #2},{(#4+#5)*sin(#1) + #3}) and} 
\def\connte(#1, #2, #3, #4, #5){ ({(#4+#5)*cos(#1) + #2},{(#4+#5)*sin(#1) + #3}) .. ({#4*cos(#1) + #2},{#4*sin(#1) + #3}) } 
\begin{tikzpicture}
\useasboundingbox (-2,-1.2) rectangle (4.5,1.2);

\draw[thick] \connt(130, 0, 0, (\r+0.2), 0.3) ++(0.2,0) .. (-1.2,1); \extl(-1.2, 1)
\draw[thick] \connt(150, 0, 0, (\r+0.2), 0.3) ++(0.2,0) .. (-1.2,0.6); \extl(-1.2, 0.6)
\draw[thick] \connt(170, 0, 0, (\r+0.2), 0.3) ++(0.2,0) .. (-1.2,0.2); \extl(-1.2, 0.2)
\draw[thick] \connt(190, 0, 0, (\r+0.2), 0.3) ++(0.2,0) .. (-1.2,-0.2); \extl(-1.2, -0.2)
\draw[thick] \connt(210, 0, 0, (\r+0.2), 0.3) ++(0.2,0) .. (-1.2,-0.6); \extl(-1.2, -0.6)
\draw[thick] \connt(190, 2.5, 0, (\r+0.2), 0.3) ++(1,0) .. (0,-1) -- (-1.2,-1); \extl(-1.2, -1.0)

\draw[thick] \connt(10, 0, 0, (\r+0.2), 0.3) ++(-1,0) .. (2,1) -- (3.7,1); \extr(3.7, 1)
\draw[thick] \connt(-30, 0, 0, (\r+0.2), 0.3) ++(-1,0) .. (2,0.7) -- (3.7,0.7); \extr(3.7, 0.7)
\draw[thick] \connt(10, 2.5, 0, (\r+0.2), 0.3) ++(-0.2,0) .. (3.7,0.2); \extr(3.7, 0.2)
\draw[thick] \connt(-10, 2.5, 0, (\r+0.2), 0.3) ++(-0.2,0) .. (3.7,-0.2); \extr(3.7, -0.2)
\draw[thick] \connt(-30, 2.5, 0, (\r+0.2), 0.3) ++(-0.2,0) .. (3.7,-0.6); \extr(3.7, -0.6)
\draw[thick] \connt(-50, 2.5, 0, (\r+0.2), 0.3) ++(-0.2,0) .. (3.7,-1); \extr(3.7, -1)

\draw[red, opacity = .8, line width = 1] \connt(50, 0, 0, (\r+0.2), 0.5) \connte(230, 2.5, 0, (\r+0.2), 0.5);
\draw[red, opacity = .8, line width = 1] \connt(-10, 0, 0, (\r+0.2), 0.5) \connte(210, 2.5, 0, (\r+0.2), 0.5);
\draw[red, opacity = .8, line width = 1] \connt(30, 0, 0, (\r+0.2), 0.5) \connte(170, 2.5, 0, (\r+0.2), 0.5);
\draw[red, opacity = .8, line width = 1, dashed] \connt(-50, 0, 0, (\r+0.2), 0.5) \connte(150, 2.5, 0, (\r+0.2), 0.5);

\filldraw[fill = yellow!50!white, thick] (0,0) circle (\r) node{$A$} ;

\draw[line width = 2, red!50!blue] \conn(130, 0, 0, \r);
\draw[line width = 2, red!50!blue] \conn(150, 0, 0, \r);
\draw[line width = 2, red!50!blue] \conn(170, 0, 0, \r);
\draw[line width = 2, red!50!blue] \conn(190, 0, 0, \r);
\draw[line width = 2, red!50!blue] \conn(210, 0, 0, \r);

\draw[line width = 2, red!50!blue] \conn(50, 0, 0, \r);
\draw[line width = 2, red!50!blue] \conn(30, 0, 0, \r);
\draw[line width = 2, red!50!blue] \conn(10, 0, 0, \r);
\draw[line width = 2, red!50!blue] \conn(-10, 0, 0, \r);
\draw[line width = 2, red!50!blue] \conn(-30, 0, 0, \r);
\draw[line width = 2, red!50!blue] \conn(-50, 0, 0, \r);

\filldraw[fill = yellow!50!white, thick] (2.5,0) circle (\r) node{$B$} ;

\draw[line width = 2, red!50!blue] \conn(150, 2.5, 0, \r);
\draw[line width = 2, red!50!blue] \conn(170, 2.5, 0, \r);
\draw[line width = 2, red!50!blue] \conn(190, 2.5, 0, \r);
\draw[line width = 2, red!50!blue] \conn(210, 2.5, 0, \r);
\draw[line width = 2, red!50!blue] \conn(230, 2.5, 0, \r);

\draw[line width = 2, red!50!blue] \conn(10, 2.5, 0, \r);
\draw[line width = 2, red!50!blue] \conn(-10, 2.5, 0, \r);
\draw[line width = 2, red!50!blue] \conn(-30, 2.5, 0, \r);
\draw[line width = 2, red!50!blue] \conn(-50, 2.5, 0, \r);

\draw[red!50!blue, dashed, thick, rotate around = {10:(0.55,0.1)}] (0.55,0.1) ellipse (0.3 and 0.6);
\draw[red!50!blue] (0.34,0.62) -- ++(-0.1,0.2) node[anchor = south] {\scriptsize apply original $ \sigma' $};
\draw[red!50!blue, dashed, thick, rotate around = {20:(1.9,-0.2)} ] (1.9,-0.2) ellipse (0.3 and 0.5);
\draw[red!50!blue] (2.14,-0.62) -- ++(0.1,-0.2) node[anchor = north] {\scriptsize apply original $ \sigma $};

\end{tikzpicture}} \hspace{-1.9cm}
	\scalebox{0.9}{\def\r{0.6} 
\def\rB{0.8} 
\def\rC{0.5} 
\def\extl(#1, #2){\fill (#1, #2) circle (0.05); \fill[opacity = 0.3, blue] (#1, #2) circle (0.1);  } 
\def\extr(#1, #2){\fill (#1, #2) circle (0.05); \fill[opacity = 0.3, green!50!black] (#1, #2) circle (0.1);  } 
\def\conn(#1, #2, #3, #4){({#4*cos(#1) + #2},{#4*sin(#1) + #3}) --  ({(#4+0.2)*cos(#1) + #2},{(#4+0.2)*sin(#1) + #3})} 
\def\connt(#1, #2, #3, #4, #5){({#4*cos(#1) + #2},{#4*sin(#1) + #3}) .. controls  ({(#4+#5)*cos(#1) + #2},{(#4+#5)*sin(#1) + #3}) and} 
\def\connte(#1, #2, #3, #4, #5){ ({(#4+#5)*cos(#1) + #2},{(#4+#5)*sin(#1) + #3}) .. ({#4*cos(#1) + #2},{#4*sin(#1) + #3}) } 
\begin{tikzpicture}
\useasboundingbox (-2,-1.2) rectangle (4.5,1.2);

\draw[thick] \connt(130, 0, 0, (\r+0.2), 0.3) ++(0.2,0) .. (-1.2,1); \extl(-1.2, 1)
\draw[thick] \connt(150, 0, 0, (\r+0.2), 0.3) ++(0.2,0) .. (-1.2,0.6); \extl(-1.2, 0.6)
\draw[thick] \connt(170, 0, 0, (\r+0.2), 0.3) ++(0.2,0) .. (-1.2,0.2); \extl(-1.2, 0.2)
\draw[thick] \connt(190, 0, 0, (\r+0.2), 0.3) ++(0.2,0) .. (-1.2,-0.2); \extl(-1.2, -0.2)
\draw[thick] \connt(210, 0, 0, (\r+0.2), 0.3) ++(0.2,0) .. (-1.2,-0.6); \extl(-1.2, -0.6)
\draw[thick] \connt(230, 2.5, 0, (\r+0.2), 0.3) ++(1,0) .. (0,-1) -- (-1.2,-1); \extl(-1.2, -1.0)

\draw[thick] \connt(50, 0, 0, (\r+0.2), 0.3) ++(-1,0) .. (2.5,1) -- (3.7,1); \extr(3.7, 1)
\draw[thick] \connt(30, 0, 0, (\r+0.2), 0.3) ++(-1,0) .. (2.5,0.7) -- (3.7,0.7); \extr(3.7, 0.7)
\draw[thick] \connt(10, 2.5, 0, (\r+0.2), 0.3) ++(-0.2,0) .. (3.7,0.2); \extr(3.7, 0.2)
\draw[thick] \connt(-10, 2.5, 0, (\r+0.2), 0.3) ++(-0.2,0) .. (3.7,-0.2); \extr(3.7, -0.2)
\draw[thick] \connt(-30, 2.5, 0, (\r+0.2), 0.3) ++(-0.2,0) .. (3.7,-0.6); \extr(3.7, -0.6)
\draw[thick] \connt(-50, 2.5, 0, (\r+0.2), 0.3) ++(-0.2,0) .. (3.7,-1); \extr(3.7, -1)

\draw[red, opacity = .8, line width = 1] \connt(10, 0, 0, (\r+0.2), 0.5) \connte(210, 2.5, 0, (\r+0.2), 0.5);
\draw[red, opacity = .8, line width = 1] \connt(-10, 0, 0, (\r+0.2), 0.5) \connte(190, 2.5, 0, (\r+0.2), 0.5);
\draw[red, opacity = .8, line width = 1] \connt(-30, 0, 0, (\r+0.2), 0.5) \connte(170, 2.5, 0, (\r+0.2), 0.5);
\draw[red, opacity = .8, line width = 1, dashed] \connt(-50, 0, 0, (\r+0.2), 0.5) \connte(150, 2.5, 0, (\r+0.2), 0.5);

\filldraw[fill = yellow!50!white, thick] (0,0) circle (\r) node{$A$} ;

\draw[line width = 2, red!50!blue] \conn(130, 0, 0, \r);
\draw[line width = 2, red!50!blue] \conn(150, 0, 0, \r);
\draw[line width = 2, red!50!blue] \conn(170, 0, 0, \r);
\draw[line width = 2, red!50!blue] \conn(190, 0, 0, \r);
\draw[line width = 2, red!50!blue] \conn(210, 0, 0, \r);

\draw[line width = 2, red!50!blue] \conn(50, 0, 0, \r);
\draw[line width = 2, red!50!blue] \conn(30, 0, 0, \r);
\draw[line width = 2, red!50!blue] \conn(10, 0, 0, \r);
\draw[line width = 2, red!50!blue] \conn(-10, 0, 0, \r);
\draw[line width = 2, red!50!blue] \conn(-30, 0, 0, \r);
\draw[line width = 2, red!50!blue] \conn(-50, 0, 0, \r);

\filldraw[fill = yellow!50!white, thick] (2.5,0) circle (\r) node{$B$} ;

\draw[line width = 2, red!50!blue] \conn(150, 2.5, 0, \r);
\draw[line width = 2, red!50!blue] \conn(170, 2.5, 0, \r);
\draw[line width = 2, red!50!blue] \conn(190, 2.5, 0, \r);
\draw[line width = 2, red!50!blue] \conn(210, 2.5, 0, \r);
\draw[line width = 2, red!50!blue] \conn(230, 2.5, 0, \r);

\draw[line width = 2, red!50!blue] \conn(10, 2.5, 0, \r);
\draw[line width = 2, red!50!blue] \conn(-10, 2.5, 0, \r);
\draw[line width = 2, red!50!blue] \conn(-30, 2.5, 0, \r);
\draw[line width = 2, red!50!blue] \conn(-50, 2.5, 0, \r);

\end{tikzpicture}} \hspace{-1cm}
	\caption{Left to right: Taking the diagram into a maximally crossed form.}
	\label{fig:Friedrichs_piplus}
\end{figure}
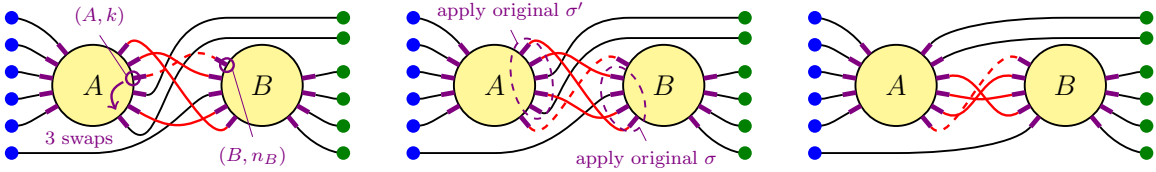

Further, the sum over configs in $ \tilde\cC $ with $ \pi: \{1, \ldots, C\} \to \cI_B \setminus \{ (B, n_B) \} $, can be recast into a sum with $ \pi_+: \{1, \ldots, C\} \to \cI_B $ over contractions in
\begin{equation}
\begin{aligned}
	\tilde\cC_+ := &\big\{ (\pi_+, \pi'_+) \; \big\vert \;
	\pi_+: \{1, \ldots, C\} \to \cI_B, \;
	\pi'_+: \{1, \ldots, C\} \to \cI'_A , \;
	\pi_+(c) \neq (B, n_B),\\
	&\qquad \qquad  \quad \;
	1 \le C \le \min(m_A, n_B - 1), \; |\mathrm{imag}(\pi')|=C, \;
	\pi(1) > \ldots > \pi(C) \big\}.\\
\end{aligned}
\label{eq:tildecC+}
\end{equation}
Replacing $ \tilde\cC $ by $ \tilde\cC_+ $, $ (B, n_B) $ becomes an additional uncontracted connector, so we have to replace $ \tilde\cU $ by $ \cU $ (compare definitions \eqref{eq:tildecU} and \eqref{eq:cU}). Concerning the sign of $ (\pi_+, \pi'_+) $, the additional uncontracted connector $ (B, n_B) $ leads to a replacement $ (-1)^{(m_A - C)(n_B - 1 - C)} \mapsto (-1)^{(m_A - C)(n_B - C)} $ in the definition of \eqref{eq:sgnpipi}, which renders an additional factor of $ (-1)^{m_A - C} $. Since there are no changes in the right--connectors of $ A $, $ \sgn(\sigma') $ stays invariant. However, the inclusion of $ (B, n_B) $ makes it necessary to pull the associated connector down, past all $ C $ contracted connectors (see Figure \ref{fig:Friedrichs_piplus2}), in order to get the diagram into a maximally crossed form. Hence, we gain an additional factor of $ (-1)^C $ in $ \sgn(\sigma) $. The overall sign change is thus given by
\begin{equation}
	\sgn(\pi_+, \pi'_+, m_A, n_B) = (-1)^{m_A} \sgn(\pi, \pi', m_A, n_B - 1).
\label{eq:sgnpi+pi+2}
\end{equation}

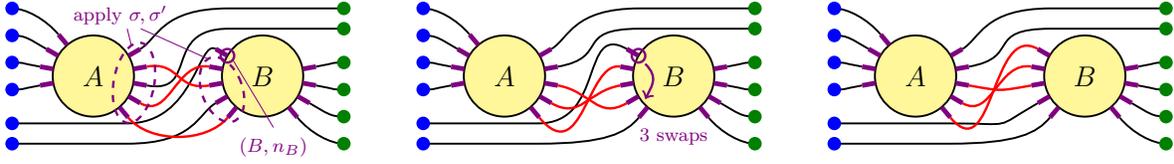
\begin{figure}[ht]
	\centering
	\hspace{-2.3cm}
	\scalebox{0.9}{\def\r{0.6} 
\def\rB{0.8} 
\def\rC{0.5} 
\def\extl(#1, #2){\fill (#1, #2) circle (0.05); \fill[opacity = 0.3, blue] (#1, #2) circle (0.1);  } 
\def\extr(#1, #2){\fill (#1, #2) circle (0.05); \fill[opacity = 0.3, green!50!black] (#1, #2) circle (0.1);  } 
\def\conn(#1, #2, #3, #4){({#4*cos(#1) + #2},{#4*sin(#1) + #3}) --  ({(#4+0.2)*cos(#1) + #2},{(#4+0.2)*sin(#1) + #3})} 
\def\connp(#1, #2, #3, #4, #5){({ #4*cos(#1) + #2},{#4*sin(#1) + #3}) } 
\def\connt(#1, #2, #3, #4, #5){({#4*cos(#1) + #2},{#4*sin(#1) + #3}) .. controls  ({(#4+#5)*cos(#1) + #2},{(#4+#5)*sin(#1) + #3}) and} 
\def\connte(#1, #2, #3, #4, #5){ ({(#4+#5)*cos(#1) + #2},{(#4+#5)*sin(#1) + #3}) .. ({#4*cos(#1) + #2},{#4*sin(#1) + #3}) } 
\begin{tikzpicture}
\useasboundingbox (-2,-1.2) rectangle (4.5,1.2);

\draw[thick] \connt(130, 0, 0, (\r+0.2), 0.3) ++(0.2,0) .. (-1.2,1); \extl(-1.2, 1)
\draw[thick] \connt(150, 0, 0, (\r+0.2), 0.3) ++(0.2,0) .. (-1.2,0.6); \extl(-1.2, 0.6)
\draw[thick] \connt(170, 0, 0, (\r+0.2), 0.3) ++(0.2,0) .. (-1.2,0.2); \extl(-1.2, 0.2)
\draw[thick] \connt(190, 0, 0, (\r+0.2), 0.3) ++(0.2,0) .. (-1.2,-0.2); \extl(-1.2, -0.2)
\draw[thick] \connt(150, 2.5, 0, (\r+0.2), 0.3) ++(1,0) .. (0.5,-0.7) -- (-1.2,-0.7); \extl(-1.2, -0.7)
\draw[thick] \connt(210, 2.5, 0, (\r+0.2), 0.3) ++(1,0) .. (0.5,-1) -- (-1.2,-1); \extl(-1.2, -1.0)

\draw[thick] \connt(30, 0, 0, (\r+0.2), 0.3) ++(-1,0) .. (2,1) -- (3.7,1); \extr(3.7, 1)
\draw[thick] \connt(-10, 0, 0, (\r+0.2), 0.8) ++(-1,0) .. (2,0.7) -- (3.7,0.7); \extr(3.7, 0.7)
\draw[thick] \connt(10, 2.5, 0, (\r+0.2), 0.3) ++(-0.2,0) .. (3.7,0.2); \extr(3.7, 0.2)
\draw[thick] \connt(-10, 2.5, 0, (\r+0.2), 0.3) ++(-0.2,0) .. (3.7,-0.2); \extr(3.7, -0.2)
\draw[thick] \connt(-30, 2.5, 0, (\r+0.2), 0.3) ++(-0.2,0) .. (3.7,-0.6); \extr(3.7, -0.6)
\draw[thick] \connt(-50, 2.5, 0, (\r+0.2), 0.3) ++(-0.2,0) .. (3.7,-1); \extr(3.7, -1)

\draw[red, opacity = .8, line width = 1] \connt(-30, 0, 0, (\r+0.2), 0.5) \connte(170, 2.5, 0, (\r+0.2), 0.5);
\draw[red, opacity = .8, line width = 1] \connt(-50, 0, 0, (\r+0.2), 0.5) \connte(230, 2.5, 0, (\r+0.2), 0.5);
\draw[red, opacity = .8, line width = 1] \connt(10, 0, 0, (\r+0.2), 0.5) \connte(190, 2.5, 0, (\r+0.2), 0.5);

\filldraw[fill = yellow!50!white, thick] (0,0) circle (\r) node{$A$} ;

\draw[line width = 2, red!50!blue] \conn(130, 0, 0, \r);
\draw[line width = 2, red!50!blue] \conn(150, 0, 0, \r);
\draw[line width = 2, red!50!blue] \conn(170, 0, 0, \r);
\draw[line width = 2, red!50!blue] \conn(190, 0, 0, \r);

\draw[line width = 2, red!50!blue] \conn(30, 0, 0, \r);
\draw[line width = 2, red!50!blue] \conn(10, 0, 0, \r);
\draw[line width = 2, red!50!blue] \conn(-10, 0, 0, \r);
\draw[line width = 2, red!50!blue] \conn(-30, 0, 0, \r);
\draw[line width = 2, red!50!blue] \conn(-50, 0, 0, \r);

\filldraw[fill = yellow!50!white, thick] (2.5,0) circle (\r) node{$B$} ;

\draw[line width = 2, red!50!blue] \conn(150, 2.5, 0, \r);
\draw[line width = 2, red!50!blue] \conn(170, 2.5, 0, \r);
\draw[line width = 2, red!50!blue] \conn(190, 2.5, 0, \r);
\draw[line width = 2, red!50!blue] \conn(210, 2.5, 0, \r);
\draw[line width = 2, red!50!blue] \conn(230, 2.5, 0, \r);

\draw[line width = 2, red!50!blue] \conn(10, 2.5, 0, \r);
\draw[line width = 2, red!50!blue] \conn(-10, 2.5, 0, \r);
\draw[line width = 2, red!50!blue] \conn(-30, 2.5, 0, \r);
\draw[line width = 2, red!50!blue] \conn(-50, 2.5, 0, \r);

\draw[red!50!blue, line width = 1] \connp(150, 2.5, 0, \r, 0) circle (0.1);
\draw[red!50!blue] (2.06,0.24) -- ++(0.6,-1) node[anchor = north] {\scriptsize $ (B, n_B) $};
\draw[red!50!blue, line width = 1, dashed, rotate around = {-10:(0.6,-0.1)} ] (0.6,-0.1) ellipse (0.3 and 0.6);
\draw[red!50!blue, line width = 1, dashed, rotate around = {20:(1.9,-0.2)} ] (1.9,-0.2) ellipse (0.3 and 0.5);
\node[red!50!blue] at (0.4,0.9) {\scriptsize apply $ \sigma, \sigma' $};
\draw[red!50!blue] (0.5,0.7) -- (0.56,0.46);
\draw[red!50!blue] (0.7,0.7) -- (1.68,0.26);

\end{tikzpicture}} \hspace{-1.9cm}
	\scalebox{0.9}{\def\r{0.6} 
\def\rB{0.8} 
\def\rC{0.5} 
\def\extl(#1, #2){\fill (#1, #2) circle (0.05); \fill[opacity = 0.3, blue] (#1, #2) circle (0.1);  } 
\def\extr(#1, #2){\fill (#1, #2) circle (0.05); \fill[opacity = 0.3, green!50!black] (#1, #2) circle (0.1);  } 
\def\conn(#1, #2, #3, #4){({#4*cos(#1) + #2},{#4*sin(#1) + #3}) --  ({(#4+0.2)*cos(#1) + #2},{(#4+0.2)*sin(#1) + #3})} 
\def\connp(#1, #2, #3, #4, #5){({ #4*cos(#1) + #2},{#4*sin(#1) + #3}) } 
\def\connt(#1, #2, #3, #4, #5){({#4*cos(#1) + #2},{#4*sin(#1) + #3}) .. controls  ({(#4+#5)*cos(#1) + #2},{(#4+#5)*sin(#1) + #3}) and} 
\def\connte(#1, #2, #3, #4, #5){ ({(#4+#5)*cos(#1) + #2},{(#4+#5)*sin(#1) + #3}) .. ({#4*cos(#1) + #2},{#4*sin(#1) + #3}) } 
\begin{tikzpicture}
\useasboundingbox (-2,-1.2) rectangle (4.5,1.2);

\draw[thick] \connt(130, 0, 0, (\r+0.2), 0.3) ++(0.2,0) .. (-1.2,1); \extl(-1.2, 1)
\draw[thick] \connt(150, 0, 0, (\r+0.2), 0.3) ++(0.2,0) .. (-1.2,0.6); \extl(-1.2, 0.6)
\draw[thick] \connt(170, 0, 0, (\r+0.2), 0.3) ++(0.2,0) .. (-1.2,0.2); \extl(-1.2, 0.2)
\draw[thick] \connt(190, 0, 0, (\r+0.2), 0.3) ++(0.2,0) .. (-1.2,-0.2); \extl(-1.2, -0.2)
\draw[thick] \connt(150, 2.5, 0, (\r+0.2), 0.8) ++(1,0) .. (0.5,-0.7) -- (-1.2,-0.7); \extl(-1.2, -0.7)
\draw[thick] \connt(230, 2.5, 0, (\r+0.2), 0.3) ++(1,0) .. (0.5,-1) -- (-1.2,-1); \extl(-1.2, -1.0)

\draw[thick] \connt(30, 0, 0, (\r+0.2), 0.3) ++(-1,0) .. (2,1) -- (3.7,1); \extr(3.7, 1)
\draw[thick] \connt(10, 0, 0, (\r+0.2), 0.8) ++(-1,0) .. (2,0.7) -- (3.7,0.7); \extr(3.7, 0.7)
\draw[thick] \connt(10, 2.5, 0, (\r+0.2), 0.3) ++(-0.2,0) .. (3.7,0.2); \extr(3.7, 0.2)
\draw[thick] \connt(-10, 2.5, 0, (\r+0.2), 0.3) ++(-0.2,0) .. (3.7,-0.2); \extr(3.7, -0.2)
\draw[thick] \connt(-30, 2.5, 0, (\r+0.2), 0.3) ++(-0.2,0) .. (3.7,-0.6); \extr(3.7, -0.6)
\draw[thick] \connt(-50, 2.5, 0, (\r+0.2), 0.3) ++(-0.2,0) .. (3.7,-1); \extr(3.7, -1)

\draw[red, opacity = .8, line width = 1] \connt(-10, 0, 0, (\r+0.2), 0.5) \connte(210, 2.5, 0, (\r+0.2), 0.5);
\draw[red, opacity = .8, line width = 1] \connt(-30, 0, 0, (\r+0.2), 0.5) \connte(190, 2.5, 0, (\r+0.2), 0.5);
\draw[red, opacity = .8, line width = 1] \connt(-50, 0, 0, (\r+0.2), 1) \connte(170, 2.5, 0, (\r+0.2), 0.5);

\filldraw[fill = yellow!50!white, thick] (0,0) circle (\r) node{$A$} ;

\draw[line width = 2, red!50!blue] \conn(130, 0, 0, \r);
\draw[line width = 2, red!50!blue] \conn(150, 0, 0, \r);
\draw[line width = 2, red!50!blue] \conn(170, 0, 0, \r);
\draw[line width = 2, red!50!blue] \conn(190, 0, 0, \r);

\draw[line width = 2, red!50!blue] \conn(30, 0, 0, \r);
\draw[line width = 2, red!50!blue] \conn(10, 0, 0, \r);
\draw[line width = 2, red!50!blue] \conn(-10, 0, 0, \r);
\draw[line width = 2, red!50!blue] \conn(-30, 0, 0, \r);
\draw[line width = 2, red!50!blue] \conn(-50, 0, 0, \r);

\filldraw[fill = yellow!50!white, thick] (2.5,0) circle (\r) node{$B$} ;

\draw[line width = 2, red!50!blue] \conn(150, 2.5, 0, \r);
\draw[line width = 2, red!50!blue] \conn(170, 2.5, 0, \r);
\draw[line width = 2, red!50!blue] \conn(190, 2.5, 0, \r);
\draw[line width = 2, red!50!blue] \conn(210, 2.5, 0, \r);
\draw[line width = 2, red!50!blue] \conn(230, 2.5, 0, \r);

\draw[line width = 2, red!50!blue] \conn(10, 2.5, 0, \r);
\draw[line width = 2, red!50!blue] \conn(-10, 2.5, 0, \r);
\draw[line width = 2, red!50!blue] \conn(-30, 2.5, 0, \r);
\draw[line width = 2, red!50!blue] \conn(-50, 2.5, 0, \r);

\draw[red!50!blue, line width = 1] \connp(150, 2.5, 0, \r, 0) circle (0.1);
\draw[red!50!blue, line width = 1, ->] (2.08, 0.18) .. controls ++(0.2, -0.2) and ++(0.1,0.2) .. (2.08, -0.34);
\node[blue!50!red] at (2.5,-0.9) {\scriptsize 3 swaps};

\end{tikzpicture}} \hspace{-1.9cm}
	\scalebox{0.9}{\def\r{0.6} 
\def\rB{0.8} 
\def\rC{0.5} 
\def\extl(#1, #2){\fill (#1, #2) circle (0.05); \fill[opacity = 0.3, blue] (#1, #2) circle (0.1);  } 
\def\extr(#1, #2){\fill (#1, #2) circle (0.05); \fill[opacity = 0.3, green!50!black] (#1, #2) circle (0.1);  } 
\def\conn(#1, #2, #3, #4){({#4*cos(#1) + #2},{#4*sin(#1) + #3}) --  ({(#4+0.2)*cos(#1) + #2},{(#4+0.2)*sin(#1) + #3})} 
\def\connp(#1, #2, #3, #4, #5){({ #4*cos(#1) + #2},{#4*sin(#1) + #3}) } 
\def\connt(#1, #2, #3, #4, #5){({#4*cos(#1) + #2},{#4*sin(#1) + #3}) .. controls  ({(#4+#5)*cos(#1) + #2},{(#4+#5)*sin(#1) + #3}) and} 
\def\connte(#1, #2, #3, #4, #5){ ({(#4+#5)*cos(#1) + #2},{(#4+#5)*sin(#1) + #3}) .. ({#4*cos(#1) + #2},{#4*sin(#1) + #3}) } 
\begin{tikzpicture}
\useasboundingbox (-2,-1.2) rectangle (4.5,1.2);

\draw[thick] \connt(130, 0, 0, (\r+0.2), 0.3) ++(0.2,0) .. (-1.2,1); \extl(-1.2, 1)
\draw[thick] \connt(150, 0, 0, (\r+0.2), 0.3) ++(0.2,0) .. (-1.2,0.6); \extl(-1.2, 0.6)
\draw[thick] \connt(170, 0, 0, (\r+0.2), 0.3) ++(0.2,0) .. (-1.2,0.2); \extl(-1.2, 0.2)
\draw[thick] \connt(190, 0, 0, (\r+0.2), 0.3) ++(0.2,0) .. (-1.2,-0.2); \extl(-1.2, -0.2)
\draw[thick] \connt(210, 2.5, 0, (\r+0.2), 0.8) ++(1,0) .. (0.5,-0.7) -- (-1.2,-0.7); \extl(-1.2, -0.7)
\draw[thick] \connt(230, 2.5, 0, (\r+0.2), 0.3) ++(1,0) .. (0.5,-1) -- (-1.2,-1); \extl(-1.2, -1.0)

\draw[thick] \connt(30, 0, 0, (\r+0.2), 0.3) ++(-1,0) .. (2,1) -- (3.7,1); \extr(3.7, 1)
\draw[thick] \connt(10, 0, 0, (\r+0.2), 0.8) ++(-1,0) .. (2,0.7) -- (3.7,0.7); \extr(3.7, 0.7)
\draw[thick] \connt(10, 2.5, 0, (\r+0.2), 0.3) ++(-0.2,0) .. (3.7,0.2); \extr(3.7, 0.2)
\draw[thick] \connt(-10, 2.5, 0, (\r+0.2), 0.3) ++(-0.2,0) .. (3.7,-0.2); \extr(3.7, -0.2)
\draw[thick] \connt(-30, 2.5, 0, (\r+0.2), 0.3) ++(-0.2,0) .. (3.7,-0.6); \extr(3.7, -0.6)
\draw[thick] \connt(-50, 2.5, 0, (\r+0.2), 0.3) ++(-0.2,0) .. (3.7,-1); \extr(3.7, -1)

\draw[red, opacity = .8, line width = 1] \connt(-10, 0, 0, (\r+0.2), 0.5) \connte(190, 2.5, 0, (\r+0.2), 0.5);
\draw[red, opacity = .8, line width = 1] \connt(-30, 0, 0, (\r+0.2), 0.5) \connte(170, 2.5, 0, (\r+0.2), 0.5);
\draw[red, opacity = .8, line width = 1] \connt(-50, 0, 0, (\r+0.2), 1) \connte(150, 2.5, 0, (\r+0.2), 0.8);

\filldraw[fill = yellow!50!white, thick] (0,0) circle (\r) node{$A$} ;

\draw[line width = 2, red!50!blue] \conn(130, 0, 0, \r);
\draw[line width = 2, red!50!blue] \conn(150, 0, 0, \r);
\draw[line width = 2, red!50!blue] \conn(170, 0, 0, \r);
\draw[line width = 2, red!50!blue] \conn(190, 0, 0, \r);

\draw[line width = 2, red!50!blue] \conn(30, 0, 0, \r);
\draw[line width = 2, red!50!blue] \conn(10, 0, 0, \r);
\draw[line width = 2, red!50!blue] \conn(-10, 0, 0, \r);
\draw[line width = 2, red!50!blue] \conn(-30, 0, 0, \r);
\draw[line width = 2, red!50!blue] \conn(-50, 0, 0, \r);

\filldraw[fill = yellow!50!white, thick] (2.5,0) circle (\r) node{$B$} ;

\draw[line width = 2, red!50!blue] \conn(150, 2.5, 0, \r);
\draw[line width = 2, red!50!blue] \conn(170, 2.5, 0, \r);
\draw[line width = 2, red!50!blue] \conn(190, 2.5, 0, \r);
\draw[line width = 2, red!50!blue] \conn(210, 2.5, 0, \r);
\draw[line width = 2, red!50!blue] \conn(230, 2.5, 0, \r);

\draw[line width = 2, red!50!blue] \conn(10, 2.5, 0, \r);
\draw[line width = 2, red!50!blue] \conn(-10, 2.5, 0, \r);
\draw[line width = 2, red!50!blue] \conn(-30, 2.5, 0, \r);
\draw[line width = 2, red!50!blue] \conn(-50, 2.5, 0, \r);

\end{tikzpicture}} \hspace{-1cm}
	\caption{Left to right: Taking the diagram into a maximally crossed form.}
	\label{fig:Friedrichs_piplus2}
\end{figure}

Performing the summation replacements $ \tilde\cC_k \mapsto \tilde\cC_{k, +} $ and $ \tilde\cC \mapsto \tilde\cC_+ $, we end up with
\begin{equation}
\begin{aligned}
	&a_{A, 1} \ldots a_{A, m_A} a^*_{B, n_B} \ldots a^*_{B, 1}\\
	= & \sum_{k = 1}^{m_A} \Bigg( \underbrace{\sum_{(\pi_+, \pi'_+) \in \tilde\cC_{k, +}} \sgn(\pi_+, \pi'_+, m_A, n_B) \prod_{c = 1}^C \delta(\bx_{\pi_+(c)} - \by_{\pi'_+(c)})
	\left( \prod_{u \in \tilde\cU} a_u \right)^*
	\prod_{u' \in \tilde\cU'_k} a_{u'}}_{\circled{1a}}\\
	&\qquad + \underbrace{(-1)^{m_A - k} (-1)^{(m_A - 1)(n_B - 1)} \delta(\bx_{B, n_B} - \by_{A, k}) \left( \prod_{u \neq (B, n_B)} a_u \right)^*
	\prod_{u' \neq (A, m_A)} a_{u'}}_{\circled{1b}} \Bigg)\\
	&+ \underbrace{\sum_{(\pi_+, \pi'_+) \in \tilde\cC_+} \sgn(\pi_+, \pi'_+, m_A, n_B) \prod_{c = 1}^C \delta(\bx_{\pi_+(c)} - \by_{\pi'_+(c)}) a^*_{B, n_B}
	\left( \prod_{u \in \cU} a_u \right)^*
	\prod_{u' \in \cU'} a_{u'}}_{\circled{2}}\\
	&+ \underbrace{(-1)^{m_A n_B} a^*_{B, n_B} \ldots a^*_{B, 1} a_{A, 1} \ldots a_{A, m_A}}_{\circled{3}}.
\end{aligned}
\label{eq:evaluatecontfer2}
\end{equation}
Comparing with \eqref{eq:contferproduct2}, the term \circled{1a} contains all configs $ (\pi, \pi') \in \tilde\cC $, where $ (B, n_B) $ is contracted to some connector $ (A, k) $, and at least one further contraction exists. The term \circled{1b} contains exactly those configs where only $ (B, n_B) $ is contracted to some $ (A, k) $. In this case, a maximal crossing can be achieved by pulling the $ (A, k) $--connector down by $ m_A - k $ positions, so following \eqref{eq:sgnpipi}, the sign of these configs amounts to
\begin{equation}
	\sgn(\pi, \pi') = (-1)^{(m_A - 1)(n_B - 1)} \cdot 1 \cdot (-1)^{m_A - k},
\end{equation}
which coincides with the sign factor in \eqref{eq:evaluatecontfer2}. \circled{2} consists of exactly those configs where $ (B, n_B) $ is not contracted at all. Since all signs $ \sgn(\pi_+, \pi'_+, m_A, n_B) = \sgn(\pi, \pi') $ agree, the sum $ \circled{1a} + \circled{1b} + \circled{2} $ equals the second term on the r.h.s. of \eqref{eq:contferproduct2}.\\
Finally, \circled{3} is just the normal ordered first term on the r.h.s. of \eqref{eq:contferproduct2}, which establishes the equality and thus the induction step.\\
\end{proof}

\begin{proof}[Proof of Theorem \ref{thm:contfer}]
Considering the definition of $ \normalordered{AB} $ \eqref{eq:normalordered}, it is easy to see that $ \normalordered{AB} = (-1)^{n_A m_A + n_B m_B} \normalordered{BA} $. So in case $ (n_A + m_A)(n_B + m_B) $ is even, we have $ (-1)^{m_A n_B} \normalordered{AB} = (-1)^{m_B n_A} \normalordered{BA} $ and therefore obtain
\begin{equation}
	[A, B] = (-1)^{m_A n_B} \normalordered{AB} + A \contfer B - (-1)^{m_B n_A} \normalordered{BA} - B \contfer A = A \contfer B - B \contfer A.
\end{equation}
If $ (n_A + m_A)(n_B + m_B) $ is odd, then $ (-1)^{m_A n_B} \normalordered{AB} = - (-1)^{m_B n_A} \normalordered{BA} $, so
\begin{equation}
	\{ A, B \} = (-1)^{m_A n_B} \normalordered{AB} + A \contfer B + (-1)^{m_B n_A} \normalordered{BA} + B \contfer A = A \contfer B + B \contfer A.
\end{equation}
\end{proof}

\section{Applications}	
\label{sec:applications}

\subsection{Hartree Equation}
\label{subsec:Hartree}
In the following section we are going to use our diagrammatic approach in order to justify the time--dependent Hartree approximation. For this purpose, let $T$ be a self--adjoint operator on the one particle Hilbert space $\fh$ and let $V$ be a bounded and self--adjoint operator on $\fh^{\otimes 2}$, which in addition is symmetric under a coordinate permutation and which we will refer to as the interaction potential. Moreover we define the second--quantized version of $T$ as $\widehat{T}:=\bigoplus_{N=0}^\infty\left(\sum_{j=1}^N T_j\right)$, where $T_j$ acts on the $j$--th factor in the tensor product $\fh^{\otimes N}$ and $\widehat{T}$ can be realized as a self--adjoint operator on $ \sF_+ $. The potential in the interaction picture is given by a bounded operator $V^{(t)}:=e^{-i t (T_1+T_2)}V e^{i t (T_1+T_2)}$ with second quantization $\widehat{V^{(t)}}:=\int V^{(t)}(x_1,x_2,y_1,y_2)a_{x_2}^* a_{x_1}^* a_{y_1} a_{y_2}\mathrm{d}x_1\mathrm{d}x_2\mathrm{d}y_1\mathrm{d}y_2$, where we write $V^{(t)}(x_1,x_2,y_1,y_2)$ for the integral kernel of the operator $V^{(t)}$. Let furthermore $A$ be an operator with integral kernel $f_A\in \fh\otimes \fh$ and let us write the conjugation with the (interaction picture) time evolution $\mathcal{U}_t:=e^{-it\left(\widehat{T}+\widehat{V}\right)}e^{it \widehat{T}}$ using Duhamel’s formula as
\begin{equation}
\begin{aligned}
\mathcal{U}_t^{-1}\, A\, \mathcal{U}_t & =\sum_{k=0}^\infty \frac{1}{i^k }\int\limits_{0\leq t_1\leq \dots \leq t_k\leq t}\left[\dots \left[A,\frac{1}{2}\widehat{V^{(t_1)}}\right],\dots , \frac{1}{2}\widehat{V^{(t_k)}}\right]\mathrm{d}t_1\dots \mathrm{d}t_k\\
&=\sum_{k=0}^\infty \sum_{\ell=1}^{k+1}\int\limits_{0\leq t_1\leq \dots \leq t_k\leq t} A_{k,\ell}^{t_1,\dots , t_k}\, \mathrm{d}t_1\dots \mathrm{d}t_k.
\end{aligned}
\label{eq:nested}
\end{equation}
Here, $A_{k,\ell}^{t_1,\dots , t_k}$ is defined as containing all contributions in $\frac{1}{i^k}\left[\dots \left[A,\frac{1}{2}\widehat{V^{(t_1)}}\right],\dots , \frac{1}{2}\widehat{V^{(t_k)}}\right]$ whose diagrams have $\ell$ incoming and outgoing legs. That means, we write the nested commutator according to Theorem \ref{thm:cont} as a sum of connected diagrams, see also the definition of the bosonic attached product in Eq.~(\ref{eq:cont}) and the comment below. Each of the $k$ newly added vertices can be contracted to one or two of the existing ones, so we finally end up with a diagram having $ \ell \in \{2, \ldots, k+1 \}$ uncontracted legs on each side (i.e., incoming/outgoing legs). In the subsequent Proposition \ref{prop:classic}, we are going to verify that the sum over all acyclic graphs (i.e., those with $ \ell = k+1 $ legs on both sides) appearing in Eq.~(\ref{eq:nested}) amounts to the time evolution given by the non--linear Hartree dynamics (see also Figure \ref{fig:Friedrichs_evolution}). In order to formulate Proposition \ref{prop:classic}, let us define the mean field potential $V_u$ associated with a one particle state $u$ via its integral kernel as $V_u(x,y):=\int V(x,x',y,y')\overline{u}(x')u(y')\mathrm{d}x'\mathrm{d}y'$.
\begin{figure}
	\centering
	\scalebox{1.0}{\def\sCl{-1.5} 
\def\sQ{6.7} 
\def\r{0.55} 
\def\rB{0.8} 
\def\rC{0.5} 
\def\extl(#1, #2){\fill (#1, #2) circle (0.05); \fill[opacity = 0.3, blue] (#1, #2) circle (0.1);  } 
\def\extr(#1, #2){\fill (#1, #2) circle (0.05); \fill[opacity = 0.3, green!50!black] (#1, #2) circle (0.1);  } 
\def\conn(#1, #2, #3, #4){({#4*cos(#1) + #2},{#4*sin(#1) + #3}) --  ({(#4+0.2)*cos(#1) + #2},{(#4+0.2)*sin(#1) + #3})} 
\def\connt(#1, #2, #3, #4, #5){({#4*cos(#1) + #2},{#4*sin(#1) + #3}) .. controls  ({(#4+#5)*cos(#1) + #2},{(#4+#5)*sin(#1) + #3}) and} 
\def\connte(#1, #2, #3, #4, #5){ ({(#4+#5)*cos(#1) + #2},{(#4+#5)*sin(#1) + #3}) .. ({#4*cos(#1) + #2},{#4*sin(#1) + #3}) } 
\begin{tikzpicture}
\useasboundingbox (-2,-1) rectangle (4.5,1);


\draw[thick] \connt(150, \sCl-2, 2, (\r+0.2), 0.8) ++(0.2,0) .. (\sCl-5.5, 3); \extl(\sCl-5.5, 3)
\draw[thick] \connt(150, \sCl-4.2, 1, (\r+0.2), 0.2) ++(0.2,0) .. (\sCl-5.5, 1+0.5); \extl(\sCl-5.5, 1+0.5)
\draw[thick] \connt(210, \sCl-4.2, 1, (\r+0.2), 0.2) ++(0.2,0) .. (\sCl-5.5, 1-0.5); \extl(\sCl-5.5, 1-0.5)
\draw[thick] \connt(250, \sCl+0.2, 1, (\r+0.2), 1.7) ++(0.2,0) .. (\sCl-5.5, -0.8); \extl(\sCl-5.5, -0.8)

\draw[thick] \connt(30, \sCl+0.2, 1, (\r+0.2), 0.2) ++(0.2,0) .. (\sCl+1.5, 1+0.5); \extr(\sCl+1.5, 1+0.5)
\draw[thick] \connt(-30, \sCl+0.2, 1, (\r+0.2), 0.2) ++(0.2,0) .. (\sCl+1.5, 1-0.5); \extr(\sCl+1.5, 1-0.5)
\draw[thick] \connt(30, \sCl-2, 2, (\r+0.2), 0.8) ++(0.2,0) .. (\sCl+1.5, 3); \extr(\sCl+1.5, 3)
\draw[thick] \connt(-30, \sCl-2, 2, (\r+0.2), 0.7) ++(0.2,0) .. (\sCl+1.5, 2.2); \extr(\sCl+1.5, 2.2)

\draw[thick] \connt(150, \sQ-2, 2, (\r+0.2), 0.8) ++(0.2,0) .. (\sQ-5.5, 3); \extl(\sQ-5.5, 3)
\draw[thick] \connt(150, \sQ-4.2, 1, (\r+0.2), 0.2) ++(0.2,0) .. (\sQ-5.5, 1+0.5); \extl(\sQ-5.5, 1+0.5)
\draw[thick] \connt(210, \sQ-4.2, 1, (\r+0.2), 0.2) ++(0.2,0) .. (\sQ-5.5, 1-0.5); \extl(\sQ-5.5, 1-0.5)

\draw[thick] \connt(30, \sQ+0.2, 1, (\r+0.2), 0.2) ++(0.2,0) .. (\sQ+1.5, 1+0.5); \extr(\sQ+1.5, 1+0.5)
\draw[thick] \connt(-30, \sQ+0.2, 1, (\r+0.2), 0.2) ++(0.2,0) .. (\sQ+1.5, 1-0.5); \extr(\sQ+1.5, 1-0.5)
\draw[thick] \connt(30, \sQ-2, 2, (\r+0.2), 0.8) ++(0.2,0) .. (\sQ+1.5, 3); \extr(\sQ+1.5, 3)

\draw[red, opacity = .8, line width = 1] \connt(0, \sCl-2, 0, (\r+0.2), 0.5) \connte(150, \sCl+0.2, 1, (\r+0.2), 0.5);
\draw[red, opacity = .8, line width = 1] \connt(30, \sCl-4.2, 1, (\r+0.2), 0.5) \connte(210, \sCl-2, 2, (\r+0.2), 0.5);
\draw[red, opacity = .8, line width = 1] \connt(-30, \sCl-4.2, 1, (\r+0.2), 0.5) \connte(180, \sCl-2, 0, (\r+0.2), 0.5);

\draw[red, opacity = .8, line width = 1] \connt(0, \sQ-2, 0, (\r+0.2), 0.5) \connte(150, \sQ+0.2, 1, (\r+0.2), 0.5);
\draw[red, opacity = .8, line width = 1] \connt(30, \sQ-4.2, 1, (\r+0.2), 0.5) \connte(210, \sQ-2, 2, (\r+0.2), 0.5);
\draw[red, opacity = .8, line width = 1] \connt(-30, \sQ-4.2, 1, (\r+0.2), 0.5) \connte(180, \sQ-2, 0, (\r+0.2), 0.5);
\draw[red, opacity = .8, line width = 1] \connt(-30, \sQ-2, 2, (\r+0.2), 1) \connte(250, \sQ+0.2, 1, (\r+0.2), 0.5);

\filldraw[fill = yellow!50!white, thick] (\sCl-2,0) circle (\r) node{$A$} ;

\draw[line width = 2, red!50!blue] \conn(0, \sCl-2, 0, \r);

\draw[line width = 2, red!50!blue] \conn(180, \sCl-2, 0, \r);

\filldraw[fill = yellow!50!white, thick] (\sQ-2,0) circle (\r) node{$A$} ;

\draw[line width = 2, red!50!blue] \conn(0, \sQ-2, 0, \r);

\draw[line width = 2, red!50!blue] \conn(180, \sQ-2, 0, \r);

\filldraw[fill = yellow!50!white, thick] (\sCl-4.2,1) circle (\r) node{$V^{(t_1)}$} ;
\draw[line width = 2, red!50!blue] \conn(-30, \sCl-4.2, 1, \r);
\draw[line width = 2, red!50!blue] \conn(30, \sCl-4.2, 1, \r);

\draw[line width = 2, red!50!blue] \conn(150, \sCl-4.2, 1, \r);
\draw[line width = 2, red!50!blue] \conn(210, \sCl-4.2, 1, \r);

\filldraw[fill = yellow!50!white, thick] (\sQ-4.2,1) circle (\r) node{$V^{(t_1)}$} ;
\draw[line width = 2, red!50!blue] \conn(-30, \sQ-4.2, 1, \r);
\draw[line width = 2, red!50!blue] \conn(30, \sQ-4.2, 1, \r);

\draw[line width = 2, red!50!blue] \conn(150, \sQ-4.2, 1, \r);
\draw[line width = 2, red!50!blue] \conn(210, \sQ-4.2, 1, \r);

\filldraw[fill = yellow!50!white, thick] (\sCl-2,2) circle (\r) node{$V^{(t_2)}$} ;
\draw[line width = 2, red!50!blue] \conn(-30, \sCl-2, 2, \r);
\draw[line width = 2, red!50!blue] \conn(30, \sCl-2, 2, \r);

\draw[line width = 2, red!50!blue] \conn(150, \sCl-2, 2, \r);
\draw[line width = 2, red!50!blue] \conn(210, \sCl-2, 2, \r);

\filldraw[fill = yellow!50!white, thick] (\sQ-2,2) circle (\r) node{$V^{(t_2)}$} ;
\draw[line width = 2, red!50!blue] \conn(-30, \sQ-2, 2, \r);
\draw[line width = 2, red!50!blue] \conn(30, \sQ-2, 2, \r);

\draw[line width = 2, red!50!blue] \conn(150, \sQ-2, 2, \r);
\draw[line width = 2, red!50!blue] \conn(210, \sQ-2, 2, \r);

\filldraw[fill = yellow!50!white, thick] (\sCl+0.2,1) circle (\r) node{$V^{(t_3)}$} ;
\draw[line width = 2, red!50!blue] \conn(-30, \sCl+0.2, 1, \r);
\draw[line width = 2, red!50!blue] \conn(30, \sCl+0.2, 1, \r);

\draw[line width = 2, red!50!blue] \conn(150, \sCl+0.2, 1, \r);
\draw[line width = 2, red!50!blue] \conn(250, \sCl+0.2, 1, \r);

\filldraw[fill = yellow!50!white, thick] (\sQ+0.2,1) circle (\r) node{$V^{(t_3)}$} ;
\draw[line width = 2, red!50!blue] \conn(-30, \sQ+0.2, 1, \r);
\draw[line width = 2, red!50!blue] \conn(30, \sQ+0.2, 1, \r);

\draw[line width = 2, red!50!blue] \conn(150, \sQ+0.2, 1, \r);
\draw[line width = 2, red!50!blue] \conn(250, \sQ+0.2, 1, \r);

\end{tikzpicture}}
	\caption{Examples of an acyclic graph (left) and a graph containing a cycle (right) arising as part of the time evolution in Eq.~(\ref{eq:nested}).}
	\label{fig:Friedrichs_evolution}
\end{figure}
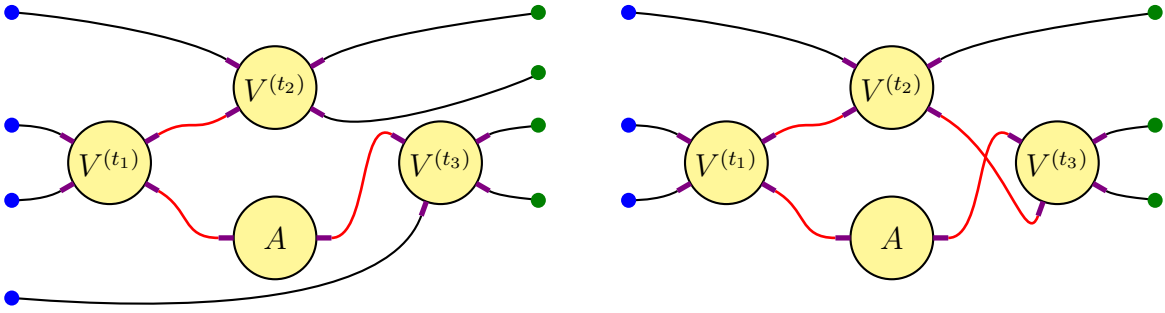
\begin{proposition}
\label{prop:classic}
Assume that $u_t$ is a solution to the Hartree equation $i\frac{\mathrm{d}}{\mathrm{d}t}u_t=(T+V_{u_t})u_t$ for $t\in (-L,L)$ and $L:=\frac{1}{2\|V\| \|u_0\|^2}$. Then we obtain for all $|t|<L$
\begin{equation}
    \langle A \rangle_{u_t^{\otimes N}}=\sum_{k=0}^\infty \int\limits_{0\leq t_1\leq \dots \leq t_k\leq t} \langle A_{k,k+1}^{t_1,\dots , t_k} \rangle_{\left(e^{-i t T}\! u_0\right)^{\otimes N}}\, \mathrm{d}t_1\dots \mathrm{d}t_k,
    \label{eq: classicEvo}
\end{equation}
where $\|V\|$ denotes the operator norm.
\end{proposition}
\begin{proof}
Let us introduce for a fixed $t\in (-L,L) $ the family of one particle states $v_s:=e^{i(s-t)T}u_s$, which interpolates between the states $v_0=e^{-itT}u_0$ and $v_t=u_t$. In the following we want to verify that the expression
\begin{equation}
    \sum_{k=0}^\infty \int\limits_{0\leq t_1\leq \dots \leq t_k\leq t-s} \langle A_{k,k+1}^{t_1,\dots , t_k} \rangle_{v_s^{\otimes N}}\, \mathrm{d}t_1\dots \mathrm{d}t_k
    \label{eq:interpolation}
\end{equation}
is constant with respect to $s$, which immediately concludes the proof by plugging in the values $s=0$ and $s=t$ (note that for $k=0$, $A_{0,1}=A$). In order to do this, we define $\{A,B\}_*$ as the sum over all graphs appearing in the expression $ A \cont B - B \cont A$ that have only a single connection between the connectors of $f_A$ and $f_B$, where we add the subscript $*$ in order to distinguish $\{A,B\}_*$ from the anti--commutator and we purposefully use curly brackets in analogy to the Poisson bracket from classical mechanics. Using this notation, we obtain $A_{k,k+1}^{t_1,\dots , t_k}=\frac{1}{2i}\{A_{k-1,k}^{t_1,\dots , t_{k-1}},V^{(t_k)}\}_*$. Since there are $4k$ different graphs appearing in $\{A_{k-1,k}^{t_1,\dots , t_{k-1}},V^{(t_k)}\}_*$, we have the estimate $\|A_{k,k+1}^{t_1,\dots , t_k}\|\leq 4k\frac{1}{2}\|V\|\|A_{k-1,k}^{t_1,\dots , t_{k-1}}\|$, where we have used that the operator norm of a contracted operator, defined in Eq.~(\ref{eq:G}), is bounded from above by the product of the operator norms of the individual operators. Consequently, $\|A_{k,k+1}^{t_1,\dots , t_k}\|\leq k!\|A\|(2\|V\|)^k$. Using the control of the norm and $\Vert v_s \Vert=\Vert u_0 \Vert$, we immediately obtain 
\begin{equation}
\begin{aligned}
  \bigg|\int\limits_{0\leq t_1\leq \dots \leq t_k\leq t-s}\! \! \! \!\!  \!  \!  \langle A_{k,k+1}^{t_1,\dots , t_k}  & \rangle_{v_s^{\otimes N}}\, \mathrm{d}t_1\dots \mathrm{d}t_k\bigg| \leq  \! \! \! \! \! \int\limits_{0\leq t_1\leq \dots \leq t_k\leq t-s}\! \! \! \! \! \! \! \! \! k!\|A\|(2\|V\|)^k\|u_0\|^{2(k+1)}\, \mathrm{d}t_1\dots \mathrm{d}t_k\\
   &\leq\!  \|A\|\|u_0\|^2 \left(2(t-s)\|V\|\|u_0\|^2\right)^k.
   \end{aligned}
\end{equation}
Therefore the sum in Eq.~(\ref{eq:interpolation}) converges absolutely for $|t-s|<\frac{1}{2\|V\|\Vert u_0 \Vert^2}$, and we can express the derivative $\frac{\mathrm{d}}{\mathrm{d}s}\left(\sum_{k=0}^\infty \int\limits_{0\leq t_1\leq \dots \leq t_k\leq t-s} \langle A_{k,k+1}^{t_1,\dots , t_k} \rangle_{v_s^{\otimes N}}\, \mathrm{d}t_1\dots \mathrm{d}t_k\right)$ as 
\begin{equation}
    \begin{aligned}
\sum_{k=0}^\infty & \int\limits_{0\leq t_1\leq \dots \leq t_k\leq t-s}\! \! \! \!   \! \!  \! \! \! \! \! \frac{\mathrm{d}}{\mathrm{d}s}\langle A_{k,k+1}^{t_1,\dots , t_k} \rangle_{v_s^{\otimes N}} \mathrm{d}t_1\! \dots\!  \mathrm{d}t_k \! -\! \sum_{k=1}^\infty \int\limits_{0\leq t_1\leq \dots \leq t_{k-1}\leq t-s}\! \! \! \! \! \! \! \!  \! \! \! \! \! \! \frac{1}{2i}\langle \{A_{k-1,k}^{t_1,\dots , t_{k-1}},V^{(t-s)}\}_* \rangle_{v_s^{\otimes N}} \mathrm{d}t_1\! \dots \! \mathrm{d}t_{k-1}\\
       &=\sum_{k=0}^\infty  \int\limits_{0\leq t_1\leq \dots \leq t_k\leq t-s}\! \! \! \!   \! \!  \! \! \! \! \left(\frac{\mathrm{d}}{\mathrm{d}s}\langle A_{k,k+1}^{t_1,\dots , t_k} \rangle_{v_s^{\otimes N}}-\frac{1}{2i}\langle \{A_{k,k+1}^{t_1,\dots , t_{k}},V^{(t-s)}\}_* \rangle_{v_s^{\otimes N}}\right)\mathrm{d}t_1\! \dots \! \mathrm{d}t_{k}=0.
    \end{aligned}
\end{equation}
In order to show the last equality, we use that $i\frac{\mathrm{d}}{\mathrm{d}s}v_s=V^{(t-s)}_{v_s}v_s$, with $V^{(t)}_u(x,y):=\int V^{(t)}(x,x',y,y')\overline{u}(x')u(y')\mathrm{d}x'\mathrm{d}y'$, which allows us to express $\frac{\mathrm{d}}{\mathrm{d}s}\langle B \rangle_{v_s^{\otimes N}}$ for any $f_B\in \fh^{\otimes n} \otimes \fh^{\otimes n}$ as
\begin{equation}
    \begin{aligned}
        & \frac{1}{i}\sum_{j=1}^n\int f_B(X,Y) \prod_{k=1}^n \overline{v_s(x_k)} v_s(y_1)\dots v_s(y_{j-1})(V^{(t-s)}_{v_s}v_s)(y_j)v_s(y_{j+1})\dots v_s(y_n)\mathrm{d}X\mathrm{d}Y\\
        &\ \ \ -\frac{1}{i}\sum_{j=1}^n\int f_B(X,Y)\overline{v_s(x_1)}\dots \overline{v_s(x_{j-1})}\overline{(V^{(t-s)}_{v_s}v_s)(x_j)}\overline{v_s(x_{j+1})}\dots \overline{v_s(x_n)}\prod_{k=1}^n v_s(y_k)\mathrm{d}X\mathrm{d}Y\\
        &=\frac{1}{i}\! \sum_{j=1}^n\! \int\!  f_B(X,Y)V^{(t-s)}(x_{n+1},x_{n+2},y_{n+1},y_{n+2})\delta(y_j-x_{n+2})\! \! \prod_{k=1}^{n+1}\overline{v_s(x_k)}\! \! \! \! \! \! \! \prod_{\ell\in \{1,\dots ,n+2\}\setminus \{j\}}\! \! \! \! \! \! \! \! \! \! v_s(y_\ell) \mathrm{d}X'\mathrm{d}Y'\\
        &-\frac{1}{i}\! \sum_{j=1}^n\! \int\!  f_B(X,Y)V^{(t-s)}(x_{n+1},x_{n+2},y_{n+1},y_{n+2})\delta(x_j-y_{n+2})\! \! \! \! \! \! \! \! \! \!  \! \! \prod_{k\in \{1,\dots ,n+2\}\setminus \{j\}}\! \! \! \! \! \! \!\overline{v_s(x_k)} \prod_{\ell=1}^{n+1}v_s(y_\ell) \mathrm{d}X'\mathrm{d}Y'\\
        &=\frac{1}{2i}\langle \{B,V^{(t-s)}\}_* \rangle_{v_s^{\otimes N}},
    \end{aligned}
\end{equation}
with $X'=(x_1,\dots ,x_{n+2})$ and $Y'=(y_1,\dots ,y_{n+2})$. Note that in the last line, we have used our definition of $\{\cdot,\cdot \}_*$ in terms of the attached product $ \cont $ defined in Eq.~(\ref{eq:cont}).
\end{proof}
Proposition \ref{prop:classic} states that every acyclic graph appearing in the quantum time evolution has to be considered as a contribution coming from the Hartree dynamics, while graphs having at least one cycle constitute the quantum correction to the Hartree dynamics. Since cyclic graphs have less open connectors relative to their number of vertices (see also Figure \ref{fig:Friedrichs_evolution}), and each open connector is of order $\sqrt{N}$, one observes that the quantum corrections are of subleading order in case $V$ is small compared to the number of particles $N$. While this approach of establishing the non--linear Hartree dynamics is comparable to the one presented in \cite{spohn}, we want to stress the graphical interpretation of the Hartree dynamics as a subset of the diagrams arising in the quantum time evolution. Furthermore, we want to note that Friedrichs diagrams have been used previously to establish the classical limit of non--relativistic bosons in \cite{ginibre}.

\subsection{Multicommutators and Bosonization}
\label{subsec:multicomm}
Let us comment a bit more on how calculations from the recent literature on Bose and Fermi gases can be re--phrased in terms of Friedrichs diagrams. A situation that often appears, e.g., in \cite[Sect.~3]{BoccatoBrenneckeCenatiempoSchlein2020}, \cite[Sect.~3]{AdhikariBrenneckeSchlein2021}, \cite[Sect.~4]{FalconiGiacomelliHainzlPorta2021}, \cite[Sect.~5]{ChristiansenHainzlNam2022}, \cite[Sect.~7]{BenedikterPortaSchleinSeiringer2022} is the evaluation of an expectation value of some operator $ A $ inside some trial state $ \psi := e^{-B} \Omega $ with $ B $ being an antisymmetric operator and $ \Omega \in \sF_- $ the vacuum vector of Fock space. Typically, a momentum lattice $ X = \ZZZ^d $ is chosen with momenta denoted $ k, p \in \ZZZ^d $. The operator $ A $ could be, for instance, $ a_k^*, a_k $ or the number operator $ \cN = \sum_k a_k^* a_k $ and $ B $ could be a quadratic/cubic/quartic transformation of the form
\begin{equation}
\begin{aligned}
	B_2 := &\sum_{k_1, k_2 \in \ZZZ^d} \eta_{k_1, k_2} \left( a^*_{k_1} a^*_{k_2} - \text{h.c.} \right),\\
	B_3 := &\sum_{k_1, k_2, k_3 \in \ZZZ^d} \eta_{k_1, k_2, k_3} \left( a^*_{k_1} a^*_{k_2} a^*_{k_3} - \text{h.c.} \right),\\
	B_4 := &\sum_{k_1, k_2, k_3, k_4 \in \ZZZ^d} \eta_{k_1, k_2, k_3, k_4} \left( a^*_{k_1} a^*_{k_2} a^*_{k_3} a^*_{k_4} - \text{h.c.} \right),\\
\end{aligned}
\end{equation}
see Figure \ref{fig:Friedrichs_aaNB2B3B4}, with suitable integral kernels $ \eta $. Here, $ e^{-B_2} $ is a Bogoliubov transformation, while $ e^{-B_3} $ and $ e^{-B_4} $ can be seen as ``generalized Bogoliubov transformations''.\\

\begin{figure}[hbt]
	\centering
	\scalebox{1.0}{\begin{tikzpicture}

\filldraw[fill = yellow!60!white, thick] (0,0.25) circle (0.15);
\draw[line width = 2, red!50!blue] (-0.15,0.25) -- ++(-0.2,0);
\node at (-0.2,1.2) {$ a_k^* $};

\filldraw[fill = yellow!60!white, thick] (1,0.25) circle (0.15);
\draw[line width = 2, red!50!blue] (1.15,0.25) -- ++(0.2,0);  
\node at (1.2,1.2) {$ a_k $};

\filldraw[fill = yellow!60!white, thick] (3,0.25) circle (0.2);
\draw[line width = 2, red!50!blue] (2.8,0.25) -- ++(-0.2,0);
\draw[line width = 2, red!50!blue] (3.2,0.25) -- ++(0.2,0);  
\node at (3,1.2) {$ \cN $};

\filldraw[thick, rounded corners = 5, fill = yellow!30!white] (5.3,0) rectangle ++(0.4,0.5);
\draw[line width = 2, red!50!blue] (5.3,0.1) -- ++(-0.2,0);
\draw[line width = 2, red!50!blue] (5.3,0.4) -- ++(-0.2,0);
\node at (6,0.25) {$ - $};
\filldraw[thick, rounded corners = 5, fill = yellow!30!white] (6.3,0) rectangle ++(0.4,0.5);
\draw[line width = 2, red!50!blue] (6.7,0.1) -- ++(0.2,0);
\draw[line width = 2, red!50!blue] (6.7,0.4) -- ++(0.2,0);
\node at (6,1.2) {$ B_2 $};

\filldraw[thick, rounded corners = 5, fill = yellow!30!white] (8.3,-0.1) rectangle ++(0.4,0.7);
\draw[line width = 2, red!50!blue] (8.3,0) -- ++(-0.2,0);
\draw[line width = 2, red!50!blue] (8.3,0.25) -- ++(-0.2,0);
\draw[line width = 2, red!50!blue] (8.3,0.5) -- ++(-0.2,0);
\node at (9,0.25) {$ - $};
\filldraw[thick, rounded corners = 5, fill = yellow!30!white] (9.3,-0.1) rectangle ++(0.4,0.7);
\draw[line width = 2, red!50!blue] (9.7,0) -- ++(0.2,0);
\draw[line width = 2, red!50!blue] (9.7,0.25) -- ++(0.2,0);
\draw[line width = 2, red!50!blue] (9.7,0.5) -- ++(0.2,0);
\node at (9,1.2) {$ B_3 $};

\filldraw[thick, rounded corners = 5, fill = yellow!30!white] (11.3,-0.25) rectangle ++(0.4,1);
\draw[line width = 2, red!50!blue] (11.3,-0.1) -- ++(-0.2,0);
\draw[line width = 2, red!50!blue] (11.3,0.1) -- ++(-0.2,0);
\draw[line width = 2, red!50!blue] (11.3,0.4) -- ++(-0.2,0);
\draw[line width = 2, red!50!blue] (11.3,0.6) -- ++(-0.2,0);
\node at (12,0.25) {$ - $};
\filldraw[thick, rounded corners = 5, fill = yellow!30!white] (12.3,-0.25) rectangle ++(0.4,1);
\draw[line width = 2, red!50!blue] (12.7,-0.1) -- ++(0.2,0);
\draw[line width = 2, red!50!blue] (12.7,0.1) -- ++(0.2,0);
\draw[line width = 2, red!50!blue] (12.7,0.4) -- ++(0.2,0);
\draw[line width = 2, red!50!blue] (12.7,0.6) -- ++(0.2,0);
\node at (12,1.2) {$ B_4 $};

\end{tikzpicture}}
	\caption{The Friedrichs vertices (without external legs) of the example operators for $ A $ and $ B $. Note that the $ B $--operators are differences of two vertices, respectively.}
	\label{fig:Friedrichs_aaNB2B3B4}
\end{figure}
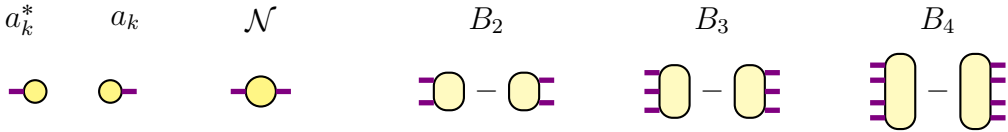
Recursive application of Duhamel's formula renders the multicommutator series
\begin{equation}
	\langle \Omega, e^B A e^{-B} \Omega \rangle
	= \sum_{n = 0}^\infty \frac{1}{n!} \langle \Omega, \underbrace{[B, \ldots [B, [B, A]] \ldots ]}_{n \text{ commutators}} \Omega \rangle.
\end{equation}
The $ n $--th term in this series can be written as a sum over all diagrams, which are built, starting from an $ A $--vertex and successively contracting  $ n $ vertices of type $ B $ into it.\\

\begin{figure}[hbt]
	\centering
	\scalebox{1.0}{\def\extl(#1, #2){\fill (#1, #2) circle (0.05); \fill[opacity = 0.3, blue] (#1, #2) circle (0.1);  } 
\def\extr(#1, #2){\fill (#1, #2) circle (0.05); \fill[opacity = 0.3, green!50!black] (#1, #2) circle (0.1);  } 

\begin{tikzpicture}
\filldraw[fill = yellow!60!white, thick] (1,0.1) circle (0.15);

\filldraw[thick, rounded corners = 5, fill = yellow!30!white] (-1.2,0) rectangle ++(0.4,0.5);
\filldraw[thick, rounded corners = 5, fill = yellow!30!white] (0.8,0.4) rectangle ++(0.4,0.5);
\filldraw[thick, rounded corners = 5, fill = yellow!30!white] (-1.2,0.8) rectangle ++(0.4,0.5);
\filldraw[thick, rounded corners = 5, fill = yellow!30!white] (0.8,1.2) rectangle ++(0.4,0.5);
\filldraw[thick, rounded corners = 5, fill = yellow!30!white] (-1.2,1.6) rectangle ++(0.4,0.5);
\filldraw[thick, rounded corners = 5, fill = yellow!30!white] (0.8,2) rectangle ++(0.4,0.5);
\filldraw[thick, rounded corners = 5, fill = yellow!30!white] (-1.2,2.4) rectangle ++(0.4,0.5);

\draw[line width = 2, red!50!blue] (0.85,0.1) -- ++(-0.2,0);
\draw[line width = 2, red!50!blue] (-0.8,0.1) -- ++(0.2,0);
\draw[line width = 2, red!50!blue] (-0.8,0.4) -- ++(0.2,0);
\draw[line width = 2, red!50!blue] (0.8,0.5) -- ++(-0.2,0);
\draw[line width = 2, red!50!blue] (0.8,0.8) -- ++(-0.2,0);
\draw[line width = 2, red!50!blue] (-0.8,0.9) -- ++(0.2,0);
\draw[line width = 2, red!50!blue] (-0.8,1.2) -- ++(0.2,0);
\draw[line width = 2, red!50!blue] (0.8,1.3) -- ++(-0.2,0);
\draw[line width = 2, red!50!blue] (0.8,1.6) -- ++(-0.2,0);
\draw[line width = 2, red!50!blue] (-0.8,1.7) -- ++(0.2,0);
\draw[line width = 2, red!50!blue] (-0.8,2) -- ++(0.2,0);
\draw[line width = 2, red!50!blue] (0.8,2.1) -- ++(-0.2,0);
\draw[line width = 2, red!50!blue] (0.8,2.4) -- ++(-0.2,0);
\draw[line width = 2, red!50!blue] (-0.8,2.5) -- ++(0.2,0);
\draw[line width = 2, red!50!blue] (-0.8,2.8) -- ++(0.2,0);

\draw[red, opacity = .8, line width = 1] (-0.6,0.1) -- ++(1.25,0);
\draw[red, opacity = .8, line width = 1] (-0.6,0.4) .. controls ++(0.5,0) and ++(-0.5,0) .. (0.6,0.5);
\draw[red, opacity = .8, line width = 1] (-0.6,0.9) .. controls ++(0.5,0) and ++(-0.5,0) .. (0.6,0.8);
\draw[red, opacity = .8, line width = 1] (-0.6,1.2) .. controls ++(0.5,0) and ++(-0.5,0) .. (0.6,1.3);
\draw[red, opacity = .8, line width = 1] (-0.6,1.7) .. controls ++(0.5,0) and ++(-0.5,0) .. (0.6,1.6);
\draw[red, opacity = .8, line width = 1] (-0.6,2) .. controls ++(0.5,0) and ++(-0.5,0) .. (0.6,2.1);
\draw[red, opacity = .8, line width = 1] (-0.6,2.5) .. controls ++(0.5,0) and ++(-0.5,0) .. (0.6,2.4);

\draw[thick] (-0.6,2.8) -- (1.4,2.8); \extr(1.4, 2.8)

\end{tikzpicture}}
	\scalebox{1.0}{\def\extl(#1, #2){\fill (#1, #2) circle (0.05); \fill[opacity = 0.3, blue] (#1, #2) circle (0.1);  } 
\def\extr(#1, #2){\fill (#1, #2) circle (0.05); \fill[opacity = 0.3, green!50!black] (#1, #2) circle (0.1);  } 

\begin{tikzpicture}
\filldraw[fill = yellow!60!white, thick] (-1,0.1) circle (0.15);

\filldraw[thick, rounded corners = 5, fill = yellow!30!white] (0.8,0) rectangle ++(0.4,0.5);
\filldraw[thick, rounded corners = 5, fill = yellow!30!white] (-1.2,0.4) rectangle ++(0.4,0.5);
\filldraw[thick, rounded corners = 5, fill = yellow!30!white] (0.8,0.8) rectangle ++(0.4,0.5);
\filldraw[thick, rounded corners = 5, fill = yellow!30!white] (-1.2,1.2) rectangle ++(0.4,0.5);
\filldraw[thick, rounded corners = 5, fill = yellow!30!white] (0.8,1.6) rectangle ++(0.4,0.5);
\filldraw[thick, rounded corners = 5, fill = yellow!30!white] (-1.2,2) rectangle ++(0.4,0.5);
\filldraw[thick, rounded corners = 5, fill = yellow!30!white] (0.8,2.4) rectangle ++(0.4,0.5);

\draw[line width = 2, red!50!blue] (-0.85,0.1) -- ++(0.2,0);
\draw[line width = 2, red!50!blue] (0.8,0.1) -- ++(-0.2,0);
\draw[line width = 2, red!50!blue] (0.8,0.4) -- ++(-0.2,0);
\draw[line width = 2, red!50!blue] (-0.8,0.5) -- ++(0.2,0);
\draw[line width = 2, red!50!blue] (-0.8,0.8) -- ++(0.2,0);
\draw[line width = 2, red!50!blue] (0.8,0.9) -- ++(-0.2,0);
\draw[line width = 2, red!50!blue] (0.8,1.2) -- ++(-0.2,0);
\draw[line width = 2, red!50!blue] (-0.8,1.3) -- ++(0.2,0);
\draw[line width = 2, red!50!blue] (-0.8,1.6) -- ++(0.2,0);
\draw[line width = 2, red!50!blue] (0.8,1.7) -- ++(-0.2,0);
\draw[line width = 2, red!50!blue] (0.8,2) -- ++(-0.2,0);
\draw[line width = 2, red!50!blue] (-0.8,2.1) -- ++(0.2,0);
\draw[line width = 2, red!50!blue] (-0.8,2.4) -- ++(0.2,0);
\draw[line width = 2, red!50!blue] (0.8,2.5) -- ++(-0.2,0);
\draw[line width = 2, red!50!blue] (0.8,2.8) -- ++(-0.2,0);

\draw[red, opacity = .8, line width = 1] (-0.65,0.1) -- ++(1.25,0);
\draw[red, opacity = .8, line width = 1] (-0.6,0.5) .. controls ++(0.5,0) and ++(-0.5,0) .. (0.6,0.4);
\draw[red, opacity = .8, line width = 1] (-0.6,0.8) .. controls ++(0.5,0) and ++(-0.5,0) .. (0.6,0.9);
\draw[red, opacity = .8, line width = 1] (-0.6,1.3) .. controls ++(0.5,0) and ++(-0.5,0) .. (0.6,1.2);
\draw[red, opacity = .8, line width = 1] (-0.6,1.6) .. controls ++(0.5,0) and ++(-0.5,0) .. (0.6,1.7);
\draw[red, opacity = .8, line width = 1] (-0.6,2.1) .. controls ++(0.5,0) and ++(-0.5,0) .. (0.6,2);
\draw[red, opacity = .8, line width = 1] (-0.6,2.4) .. controls ++(0.5,0) and ++(-0.5,0) .. (0.6,2.5);

\draw[thick] (0.6,2.8) -- (-1.4,2.8); \extl(-1.4, 2.8)

\end{tikzpicture}}
	\hspace{2cm}
	\scalebox{1.0}{\begin{tikzpicture}

\filldraw[fill = yellow!60!white, thick] (-0.5,-0.1) circle (0.2);

\filldraw[thick, rounded corners = 5, fill = yellow!30!white] (-1.7,-0.25) rectangle ++(0.4,1);
\filldraw[thick, rounded corners = 5, fill = yellow!30!white] (-1.7,1) rectangle ++(0.4,1);
\filldraw[thick, rounded corners = 5, fill = yellow!30!white] (-1.7,2.25) rectangle ++(0.4,1);

\filldraw[thick, rounded corners = 5, fill = yellow!30!white] (1.3,-0.25) rectangle ++(0.4,1);
\filldraw[thick, rounded corners = 5, fill = yellow!30!white] (1.3,1) rectangle ++(0.4,1);
\filldraw[thick, rounded corners = 5, fill = yellow!30!white] (1.3,2.25) rectangle ++(0.4,1);

\draw[line width = 2, red!50!blue] (-0.7,-0.1) -- ++(-0.2,0);
\draw[line width = 2, red!50!blue] (-0.3,-0.1) -- ++(0.2,0);

\draw[line width = 2, red!50!blue] (-1.1,-0.1) -- ++(-0.2,0);
\draw[line width = 2, red!50!blue] (-1.1,0.1) -- ++(-0.2,0);
\draw[line width = 2, red!50!blue] (-1.1,0.4) -- ++(-0.2,0);
\draw[line width = 2, red!50!blue] (-1.1,0.6) -- ++(-0.2,0);

\draw[line width = 2, red!50!blue] (-1.1,1.15) -- ++(-0.2,0);
\draw[line width = 2, red!50!blue] (-1.1,1.35) -- ++(-0.2,0);
\draw[line width = 2, red!50!blue] (-1.1,1.65) -- ++(-0.2,0);
\draw[line width = 2, red!50!blue] (-1.1,1.85) -- ++(-0.2,0);

\draw[line width = 2, red!50!blue] (-1.1,2.4) -- ++(-0.2,0);
\draw[line width = 2, red!50!blue] (-1.1,2.6) -- ++(-0.2,0);
\draw[line width = 2, red!50!blue] (-1.1,2.9) -- ++(-0.2,0);
\draw[line width = 2, red!50!blue] (-1.1,3.1) -- ++(-0.2,0);

\draw[line width = 2, red!50!blue] (1.1,-0.1) -- ++(0.2,0);
\draw[line width = 2, red!50!blue] (1.1,0.1) -- ++(0.2,0);
\draw[line width = 2, red!50!blue] (1.1,0.4) -- ++(0.2,0);
\draw[line width = 2, red!50!blue] (1.1,0.6) -- ++(0.2,0);

\draw[line width = 2, red!50!blue] (1.1,1.15) -- ++(0.2,0);
\draw[line width = 2, red!50!blue] (1.1,1.35) -- ++(0.2,0);
\draw[line width = 2, red!50!blue] (1.1,1.65) -- ++(0.2,0);
\draw[line width = 2, red!50!blue] (1.1,1.85) -- ++(0.2,0);

\draw[line width = 2, red!50!blue] (1.1,2.4) -- ++(0.2,0);
\draw[line width = 2, red!50!blue] (1.1,2.6) -- ++(0.2,0);
\draw[line width = 2, red!50!blue] (1.1,2.9) -- ++(0.2,0);
\draw[line width = 2, red!50!blue] (1.1,3.1) -- ++(0.2,0);

\draw[red, opacity = .8, line width = 1] (-1.1,-0.1) -- (-0.9,-0.1);
\draw[red, opacity = .8, line width = 1] (-1.1,0.1) .. controls ++(2,0) and ++(-1,0) .. (1.1,-0.1);
\draw[red, opacity = .8, line width = 1] (1.1,0.1) .. controls ++(-1.5,0) and ++(1.5,0) .. (-1.1,0.4);
\draw[red, opacity = .8, line width = 1] (-1.1,0.6) .. controls ++(1.5,0) and ++(-1.5,0) .. (1.1,2.6);
\draw[red, opacity = .8, line width = 1] (1.1,2.4) .. controls ++(-1.5,0) and ++(1.5,0) .. (-1.1,3.1);
\draw[red, opacity = .8, line width = 1] (-1.1,2.9) .. controls ++(1.5,0) and ++(-1.5,0) .. (1.1,0.6);
\draw[red, opacity = .8, line width = 1] (1.1,0.4) .. controls ++(-1,0) and ++(1,0) .. (-0.1,-0.1);

\draw[red, opacity = .8, line width = 1] (-1.1,1.35) .. controls ++(1.5,0) and ++(-1.5,0) .. (1.1,1.65);
\draw[red, opacity = .8, line width = 1] (1.1,1.85) .. controls ++(-1.5,0) and ++(1.5,0) .. (-1.1,1.65);
\draw[red, opacity = .8, line width = 1] (-1.1,1.85) .. controls ++(1.5,0) and ++(-1.5,0) .. (1.1,2.9);
\draw[red, opacity = .8, line width = 1] (1.1,3.1) .. controls ++(-1.5,0) and ++(1.5,0) .. (-1.1,1.15);

\draw[red, opacity = .8, line width = 1] (-1.1,2.6) .. controls ++(1.5,0) and ++(-1.5,0) .. (1.1,1.15);
\draw[red, opacity = .8, line width = 1] (1.1,1.35) .. controls ++(-1.5,0) and ++(1.5,0) .. (-1.1,2.4);

\end{tikzpicture}}
	\caption{Left: The two diagrams appearing within the $ n = 7 $--fold multicommutator for $ e^{B_2} a^*_k e^{-B_2} $ and $ e^{B_2} a_k e^{-B_2} $, respectively.\\ Right: A diagram appearing for $ n = 6 $, when evaluating $ \langle \Omega, e^{B_4} \cN e^{-B_4} \Omega \rangle $.}
	\label{fig:Friedrichs_3loop}
\end{figure}
Figure \ref{fig:Friedrichs_3loop} shows certain diagrams appearing for $ A = a_k^*, a_k $ and $ B = B_2 $, which arise when Bogoliubov--transforming the operators $ a_k^* $ and $ a_k $. Every new $ B_2 $--vertex can only be contracted to exactly one leg, leaving again exactly one external leg. So inductively, one quickly sees that both $ e^{B_2} a_k^* e^{-B_2} $ and $ e^{B_2} a_k e^{-B_2} $ must again be linear combinations of $ a^* $-- and $ a $--operators.\\
By contrast, for other choices of $ A $ and $ B $, e.g., $ A = \cN $ and $ B =  B_4 $, a much bigger number of contractions occurs, which leads to more diagrams, one of them being shown in Figure \ref{fig:Friedrichs_3loop}. However, when evaluating the expectation value with respect to $ \Omega $, all diagrams with at least one external leg vanish as $ a_k \Omega = 0 $.\\

For fermions, commutator evaluations can be strongly simplified by a bosonization technique, which was recently applied in \cite{Giacomelli2022short, Giacomelli2022, FalconiGiacomelliHainzlPorta2021, ChristiansenHainzlNam2022, BenedikterPortaSchleinSeiringer2022}. Introducing the bosonized operators for $ f \in \ell^2(\ZZZ^d) $,
\begin{equation}
	c^*_k(f) := \sum_{p \in \ZZZ^d} f_p a^*_p a^*_{p - k}, \qquad
	c_k(f) := \sum_{p \in \ZZZ^d} \overline{f_p} a_{p - k} a_p,
\end{equation}
we have the following commutation relations, which are a generalized case of \cite[(4.16)]{ChristiansenHainzlNam2022} and \cite[(5.6)]{BenedikterPortaSchleinSeiringer2022} with $ f, g \in \ell^2(\ZZZ^d) $:
\begin{equation}
\begin{aligned}
	[c_k(f), c_{k'}(g)] = &[c_k^*(f), c^*_{k'}(g)] = 0,\\
	[c_k(f), c^*_{k'}(g)]
	= & \delta_{k, k'} \langle f, g \rangle
	- \sum_p \overline{f_p} g_{p-k} \delta_{k, -k'}
	- \sum_p \overline{f_p} g_p a^*_{p-k'} a_{p-k}
	- \sum_p \overline{f_p} g_{p-k+k'} a^*_{p-k+k'} a_p\\
	& \qquad + \sum_p \overline{f_p} g_{p-k} a^*_{p-k-k'} a_p
	+ \sum_p \overline{f_p} g_{p+k'} a^*_{p+k'} a_{p-k}	
\end{aligned}
\label{eq:approximateCCR}
\end{equation}

\begin{figure}[hbt]
	\centering
	\scalebox{1.0}{\def\r{0.2} 
\def\extl(#1, #2){\fill (#1, #2) circle (0.05); \fill[opacity = 0.3, blue] (#1, #2) circle (0.1);  } 
\def\extr(#1, #2){\fill (#1, #2) circle (0.05); \fill[opacity = 0.3, green!50!black] (#1, #2) circle (0.1);  } 
\def\conn(#1, #2, #3, #4){({#4*cos(#1) + #2},{#4*sin(#1) + #3}) --  ({(#4+0.2)*cos(#1) + #2},{(#4+0.2)*sin(#1) + #3})} 
\def\connt(#1, #2, #3, #4, #5){({#4*cos(#1) + #2},{#4*sin(#1) + #3}) .. controls  ({(#4+#5)*cos(#1) + #2},{(#4+#5)*sin(#1) + #3}) and} 
\def\connte(#1, #2, #3, #4, #5){ ({(#4+#5)*cos(#1) + #2},{(#4+#5)*sin(#1) + #3}) .. ({#4*cos(#1) + #2},{#4*sin(#1) + #3}) } 
\begin{tikzpicture}

\node at (-0.4,1.2) {$\big[$};

\filldraw[fill = yellow!50!white, thick] (0,1.2) circle (\r);
\draw[line width = 2, red!50!blue] \conn(20, 0, 1.2, \r);
\draw[line width = 2, red!50!blue] \conn(-20, 0, 1.2, \r);

\node at (0.6,1) {$,$};

\filldraw[fill = yellow!50!white, thick] (1,1.2) circle (\r);
\draw[line width = 2, red!50!blue] \conn(20, 1, 1.2, \r);
\draw[line width = 2, red!50!blue] \conn(-20, 1, 1.2, \r);

\node at (1.6,1.2) {$\big]$};
\node at (2,1.2) {$=$};
\node at (2.4,1.2) {$0$};

\node at (3.6,1.2) {$\big[$};

\filldraw[fill = yellow!50!white, thick] (4.2,1.2) circle (\r);
\draw[line width = 2, red!50!blue] \conn(160, 4.2, 1.2, \r);
\draw[line width = 2, red!50!blue] \conn(200, 4.2, 1.2, \r);

\node at (4.6,1) {$,$};

\filldraw[fill = yellow!50!white, thick] (5.2,1.2) circle (\r);
\draw[line width = 2, red!50!blue] \conn(160, 5.2, 1.2, \r);
\draw[line width = 2, red!50!blue] \conn(200, 5.2, 1.2, \r);

\node at (5.6,1.2) {$\big]$};
\node at (6,1.2) {$=$};
\node at (6.4,1.2) {$0$};

\node at (-0.4,0) {$\big[$};

\filldraw[fill = yellow!50!white, thick] (0,0) circle (\r);
\draw[line width = 2, red!50!blue] \conn(20, 0, 0, \r);
\draw[line width = 2, red!50!blue] \conn(-20, 0, 0, \r);

\node at (0.6,-0.2) {$,$};

\filldraw[fill = yellow!50!white, thick] (1.2,0) circle (\r);
\draw[line width = 2, red!50!blue] \conn(160, 1.2, 0, \r);
\draw[line width = 2, red!50!blue] \conn(200, 1.2, 0, \r);

\node at (1.6,0) {$\big]$};
\node at (2,0) {$=$};

\filldraw[fill = yellow!50!white, thick] (2.5,0) circle (\r);
\draw[line width = 2, red!50!blue] \conn(20, 2.5, 0, \r);
\draw[line width = 2, red!50!blue] \conn(-20, 2.5, 0, \r);
\filldraw[fill = yellow!50!white, thick] (4,0) circle (\r);
\draw[line width = 2, red!50!blue] \conn(160, 4, 0, \r);
\draw[line width = 2, red!50!blue] \conn(200, 4, 0, \r);
\draw[red, opacity = .8, line width = 1] \connt(20, 2.5, 0, (\r+0.2), 0.3) \connte(200, 4, 0, (\r+0.2), 0.3);
\draw[red, opacity = .8, line width = 1] \connt(-20, 2.5, 0, (\r+0.2), 0.3) \connte(160, 4, 0, (\r+0.2), 0.3);

\node at (4.5,0) {$-$};

\filldraw[fill = yellow!50!white, thick] (5,0) circle (\r);
\draw[line width = 2, red!50!blue] \conn(20, 5, 0, \r);
\draw[line width = 2, red!50!blue] \conn(-20, 5, 0, \r);
\filldraw[fill = yellow!50!white, thick] (6.5,0) circle (\r);
\draw[line width = 2, red!50!blue] \conn(160, 6.5, 0, \r);
\draw[line width = 2, red!50!blue] \conn(200, 6.5, 0, \r);
\draw[red, opacity = .8, line width = 1] \connt(20, 5, 0, (\r+0.2), 0.3) \connte(160, 6.5, 0, (\r+0.2), 0.3);
\draw[red, opacity = .8, line width = 1] \connt(-20, 5, 0, (\r+0.2), 0.3) \connte(200, 6.5, 0, (\r+0.2), 0.3);

\node at (7,0) {$-$};

\filldraw[fill = yellow!50!white, thick] (7.5,0) circle (\r);
\draw[line width = 2, red!50!blue] \conn(20, 7.5, 0, \r);
\draw[line width = 2, red!50!blue] \conn(-20, 7.5, 0, \r);
\filldraw[fill = yellow!50!white, thick] (9,0) circle (\r);
\draw[line width = 2, red!50!blue] \conn(160, 9, 0, \r);
\draw[line width = 2, red!50!blue] \conn(200, 9, 0, \r);
\draw[red, opacity = .8, line width = 1] \connt(-20, 7.5, 0, (\r+0.2), 0.3) \connte(160, 9, 0, (\r+0.2), 0.3);
\draw[thick] \connt(20, 7.5, 0, (\r+0.2), 0.3)  ++(-0.6,0) .. (8.8,0.3) -- (9.2,0.3); \extr(9.2, 0.3)
\draw[thick] \connt(200, 9, 0, (\r+0.2), 0.3)  ++(0.6,0) .. (7.7,-0.3) -- (7.3,-0.3); \extl(7.3, -0.3)

\node at (9.5,0) {$-$};

\filldraw[fill = yellow!50!white, thick] (10,0) circle (\r);
\draw[line width = 2, red!50!blue] \conn(20, 10, 0, \r);
\draw[line width = 2, red!50!blue] \conn(-20, 10, 0, \r);
\filldraw[fill = yellow!50!white, thick] (11.5,0) circle (\r);
\draw[line width = 2, red!50!blue] \conn(160, 11.5, 0, \r);
\draw[line width = 2, red!50!blue] \conn(200, 11.5, 0, \r);
\draw[red, opacity = .8, line width = 1] \connt(20, 10, 0, (\r+0.2), 0.3) \connte(200, 11.5, 0, (\r+0.2), 0.3);
\draw[thick] \connt(-20, 10, 0, (\r+0.2), 0.3)  ++(-0.6,0) .. (11.3,0.3) -- (11.7,0.3); \extr(11.7, 0.3)
\draw[thick] \connt(160, 11.5, 0, (\r+0.2), 0.3)  ++(0.6,0) .. (10.2,-0.3) -- (9.8,-0.3); \extl(9.8, -0.3)

\node at (3,-1.0) {$+$};

\filldraw[fill = yellow!50!white, thick] (3.5,-1.0) circle (\r);
\draw[line width = 2, red!50!blue] \conn(20, 3.5, -1.0, \r);
\draw[line width = 2, red!50!blue] \conn(-20, 3.5, -1.0, \r);
\filldraw[fill = yellow!50!white, thick] (5,-1.0) circle (\r);
\draw[line width = 2, red!50!blue] \conn(160, 5, -1.0, \r);
\draw[line width = 2, red!50!blue] \conn(200, 5, -1.0, \r);
\draw[red, opacity = .8, line width = 1] \connt(20, 3.5, -1.0, (\r+0.2), 0.3) \connte(160, 5, -1.0, (\r+0.2), 0.3);
\draw[thick] \connt(-20, 3.5, -1.0, (\r+0.2), 0.3)  ++(-0.6,0) .. (4.8,-0.7) -- (5.2,-0.7); \extr(5.2, -0.7)
\draw[thick] \connt(200, 5, -1.0, (\r+0.2), 0.3)  ++(0.6,0) .. (3.7,-1.3) -- (3.3,-1.3); \extl(3.3, -1.3)

\node at (5.5,-1.0) {$+$};

\filldraw[fill = yellow!50!white, thick] (6,-1.0) circle (\r);
\draw[line width = 2, red!50!blue] \conn(20, 6, -1.0, \r);
\draw[line width = 2, red!50!blue] \conn(-20, 6, -1.0, \r);
\filldraw[fill = yellow!50!white, thick] (7.5,-1.0) circle (\r);
\draw[line width = 2, red!50!blue] \conn(160, 7.5, -1.0, \r);
\draw[line width = 2, red!50!blue] \conn(200, 7.5, -1.0, \r);
\draw[red, opacity = .8, line width = 1] \connt(-20, 6, -1.0, (\r+0.2), 0.3) \connte(200, 7.5, -1.0, (\r+0.2), 0.3);
\draw[thick] \connt(20, 6, -1.0, (\r+0.2), 0.3)  ++(-0.6,0) .. (7.3,-0.7) -- (7.7,-0.7); \extr(7.7, -0.7)
\draw[thick] \connt(160, 7.5, -1.0, (\r+0.2), 0.3)  ++(0.6,0) .. (6.2,-1.3) -- (5.8,-1.3); \extl(5.8, -1.3)

\end{tikzpicture}}
	\caption{For $ [c, c] $ and $ [c^*, c^*] $, there are no legs to contract, so these terms vanish.\\ For $ [c, c^*] $, there are 6 contributions to the commutator.}
	\label{fig:Friedrichs_almostCCR}
\end{figure}
In fact, the first line can easily be seen diagrammatically, as there are no legs to contract between the corresponding Friedrichs vertices, see Figure \ref{fig:Friedrichs_almostCCR}. The Friedrichs diagrams for the second line are shown in the same figure. Here, the second term typically vanishes due to constraints, forcing the momenta to be inside or outside the Fermi ball. For the same reason, two of the last four terms typically vanish, while the other two can be seen as error terms. They become small, whenever taking expectation values within states $ \psi \in \sF_- $ with $ \Vert \cN \psi \Vert $ being small. In that case, \eqref{eq:approximateCCR} becomes approximately equivalent to the CCR \eqref{eq:CCRCAR}, so $ c^*_k(f), c_k(f) $ can be viewed as almost--bosonic operators. The process of re--expressing operators in terms of $ c^*_k(f), c_k(f) $ is exactly the above--mentioned ``bosonization'' technique.\\

\begin{appendices}
\addtocontents{toc}{\setcounter{tocdepth}{-2}}

\section{Heuristic Motivation of Fermionic Contraction Signs}
\label{app:heuristiccontfer}

The fermionic sign factors $ \sgn(\pi, \pi') $ in \eqref{eq:sgnpipi} and \eqref{eq:contfer}, as well as the factor $ (-1)^{n_B m_A} $ in \eqref{eq:contferproduct} play a central role in the evaluation of fermionic commutators. Here, we provide a simple heuristic motivation of how these factors come about. We believe that these heuristics might be useful to quickly determine the fermionic signs within practical fermionic commutator evaluations.\\
In order to arrive at the fermionic product formula \eqref{eq:contferproduct}, we could also have taken the following approach: Starting from $ AB $, we successively pull creation operators to the left by ``adding smart zeros'', until we end up with $ \normalordered{AB} $. As a simple example, let us consider $ A = a_1 a_2, B = a_3^* a_4^* $, so $ m_A = n_B = 2 $:
\begin{equation}
\begin{aligned}
    AB = &a_1 a_2 a_3^* a_4^*\\
    = &a_1 a_2 a_3^* a_4^* + a_1 a_3^* a_2 a_4^*
        - a_1 a_3^* a_2 a_4^* - a_3^* a_1 a_2 a_4^*\\
        &+ a_3^* a_1 a_2 a_4^* + a_3^* a_1 a_4^* a_2
        - a_3^* a_1 a_4^* a_2 - a_3^* a_4^* a_1 a_2
        + a_3^* a_4^* a_1 a_2\\
    = & \{a_2, a_3^*\} a_1 a_4^* - \{a_1, a_3^*\} a_2 a_4^* + \{a_2, a_4^*\} a_3^* a_1 - \{a_1, a_4^*\} a_3^* a_2 + \underbrace{a_3^* a_4^* a_1 a_2}_{= \normalordered{AB}}.
\label{eq:stepcomm}
\end{aligned}
\end{equation}
For general $ A, B $ of the form \eqref{eq:A}, we have to pull $ m_A $ creation operators past $ n_B $ annihilation operators, so the normal ordered term in the end picks up a factor of $ (-1)^{m_A n_B} $, which is the same one as in \eqref{eq:contferproduct}. Note that the $ n_A $ creation operators of $ A $ and the $ m_B $ annihilation operators of $ B $ are already in their normal ordered positions.\\

\begin{itemize}
    \item \textit{So, the factor of $ (-1)^{m_A n_B} $ in front of $ \normalordered{AB} $ in \eqref{eq:contferproduct}  heuristically arises when pulling the $ a, a^* $--operators into their normal ordered positions.}
\end{itemize}

\noindent Now, consider some generic $ A $ and $ B $, involving operator products $ a_{A, 1} \ldots a_{A, n_A} $ and $ a^*_{B, n_B} \ldots a^*_{B, 1} $. The product $ AB $ then takes the form:
\begin{equation}
    AB
  = A \contfer B + (-1)^{m_A n_B} :AB:
  = \sum_{(\pi, \pi') \in \cC} \sgn(\pi, \pi') G_{\pi, \pi'}  + (-1)^{m_A n_B} :AB:,
\end{equation}
where $ G_{\pi, \pi'} $ is an abbreviation for the diagram with contractions described by the maps $ \pi $ and $ \pi' $ as in \eqref{eq:contfer}. Our goal is to determine $ \sgn(\pi, \pi') $. We first consider the easier case where $ (\pi, \pi') $ corresponds to a maximally crossed diagram and then reduce the generic case to the maximally crossed one. A ``smart zero insertion'' as in \eqref{eq:stepcomm} will produce many anticommutators, where the first term appearing is
\begin{equation}
    \{a_{A, m_A}, a^*_{B, n_B} \} a_{A, 1} \ldots a_{A, m_A-1} a^*_{B, n_B-1} \ldots a^*_{B, 1},
\label{eq:anticommutators}
\end{equation}
with a sign of $ +1 $. In order to arrive at \eqref{eq:contferproduct}, the remaining products of $ (m_A - 1) + (n_B - 1) $ operators as in \eqref{eq:anticommutators} also have to be brought in a normal ordered form, which is done by repeatedly applying the zero insertion procedure as in \eqref{eq:stepcomm}. After $ C \in \NNN $ insertion procedures, the appearing terms have $ C $ anticommutators (taken ``from the inside out''), where the first appearing term is
\begin{equation}
    \{a_{A, m_A}, a^*_{B, n_B} \} \ldots \{a_{A, m_A-C}, a^*_{B, n_B-C} \} a_{A, 1} \ldots a_{A, m_A-C-1} a^*_{B, n_B-C-1} \ldots a^*_{B, 1}.
\label{eq:Ccontractions}
\end{equation}
This term carries a sign of $ +1 $ and contains the same contractions as a maximally crossed diagram, e.g., as in Figures \ref{fig:Friedrichs_maxcrossing}, \ref{fig:Friedrichs_piplus} and \ref{fig:Friedrichs_piplus2}. One further application of the zero insertion procedure now takes \eqref{eq:Ccontractions} into normal ordered form, where it exactly corresponds to a maximally crossed diagram. In this last procedure, we pick up a factor of $ (-1)^{(m_A-C)(n_B-C)} $.\\

\begin{itemize}
    \item \textit{So, the factor of $ (-1)^{(m_A-C)(n_B-C)} $ in \eqref{eq:sgnpipi} heuristically arises from pulling the uncontracted $ a, a^* $--operators into their normal ordered positions.}
\end{itemize}

\noindent If a generic $ G_{\pi, \pi'} $ in $ A \contfer B $ is not maximally crossed, then there are permutations $ \sigma, \sigma' $ of the index sets $ (1, \ldots, m_A) $ and $ (1, \ldots, n_B) $ that take $ G_{\pi, \pi'} $ into maximally crossed form. That means, we define $ \sigma $ and $ \sigma' $ such that $ \sigma'(m_A) $ is contracted to $ \sigma(n_B) $, $ \sigma'(m_A - 1) $ to $ \sigma(n_B - 1) $, and so on. If we apply $ \sigma $ and $ \sigma' $ to the indices of $ a_{A, 1} \ldots a_{A, m_A} $ and $ a^*_{B, n_B} \ldots a^*_{B, 1} $ before the evaluation, then within
\begin{equation}
    a_{A, \sigma'(1)} \ldots a_{A, \sigma'(m_A)} a^*_{B, \sigma(n_B)} \ldots a^*_{B, \sigma(1)},
\end{equation}
$ G_{\pi, \pi'} $ appears with contractions being taken ``from the inside out'', and hence as a maximally crossed diagram carrying a sign of $ (-1)^{(m_A-C)(n_B-C)} $. Now, since $ AB $ includes the product
\begin{equation}
    a_{A, 1} \ldots a_{A, m_A} a^*_{B, n_B} \ldots a^*_{B, 1}
    = \sgn(\sigma) \sgn(\sigma') a_{A, \sigma'(1)} \ldots a_{A, \sigma'(m_A)} a^*_{B, \sigma(n_B)} \ldots a^*_{B, \sigma(1)}
\end{equation}
and the order of all uncontracted operators is untouched by $ \sigma $ and $ \sigma' $, the overall sign of $ G_{\pi, \pi'} $ appearing in $ AB $ is given by $ \sgn(\pi, \pi') = \sgn(\sigma) \sgn(\sigma') (-1)^{(m_A-C)(n_B-C)} $. This is the same factor as in \eqref{eq:sgnpipi}.

\begin{itemize}
    \item \textit{So, the factor $ \sgn(\sigma) \sgn(\sigma') $ in \eqref{eq:sgnpipi} is heuristically required to permute the operators such that contractions can be taken ``from the inside out''. More generally, maximal crossing just corresponds to taking contractions ``from the inside out''.}
\end{itemize}

\end{appendices}

\bigskip

\noindent\textit{Conflict of Interest.}
On behalf of all authors, the corresponding author states that there is no conflict of interest.\\

\noindent\textit{Data Availability.}
Data sharing not applicable to this article as no datasets were generated or analysed during the current study.\\

\noindent\textit{Acknowledgments.}
SL was supported by the European Research Council (ERC) through the Starting Grant \textsc{FermiMath}, Grant Agreement No. 101040991 of Niels Benedikter.\\

\end{document}